\documentclass{article}[11pt]

\usepackage{fullpage}
\usepackage{cite,color,xspace}
\usepackage{todonotes}
\usepackage{amsmath,amssymb}
\usepackage{amsfonts}
\usepackage{enumitem}

\usepackage{bibunits}

\usepackage{subfig}

\usepackage{tikz,tikz-qtree,tikz-qtree-compat,tikz-3dplot}
\usetikzlibrary{patterns, snakes}

\usepackage{xcolor}
\usepackage{tabu}

\usepackage{amsthm}
\usepackage{cite}
\usepackage{url}
\usepackage{graphicx}
\usepackage{algorithm}
\usepackage[noend]{algpseudocode}

\usepackage{appendix}

\newcommand{\set}[1]{\left\{ #1 \right\}}

\DeclareMathOperator{\polylog}{polylog}
\newcommand{\floor}[1]{\left\lfloor #1 \right\rfloor}
\newcommand{\norm}[2]{\left\| {#2} \right\|_{{#1}}}

\newcommand{\mvec}[1]{\mathbf{#1}}

\newcommand{\vb}{\mvec{b}}
\newcommand{\vc}{\mvec{c}}

\newcommand{\ve}{\mvec{e}}
\newcommand{\vf}{\mvec{f}}

\newcommand{\vq}{\mvec{q}}

\newcommand{\vr}{\mvec{r}}

\newcommand{\vt}{\mvec{t}}
\newcommand{\vu}{\mvec{u}}
\newcommand{\vv}{\mvec{v}}
\newcommand{\vw}{\mvec{w}}
\newcommand{\vx}{\mvec{x}}
\newcommand{\vy}{\mvec{y}}
\newcommand{\vz}{\mvec{z}}
\newcommand{\vha}{\hat{\mvec{a}}}

\newcommand{\vhq}{\hat{\mvec{q}}}
\newcommand{\vhr}{\hat{\mvec{r}}}
\newcommand{\vht}{\hat{\mvec{t}}}
\newcommand{\vhu}{\hat{\mvec{u}}}
\newcommand{\vhv}{\hat{\mvec{v}}}
\newcommand{\vhw}{\hat{\mvec{w}}}
\newcommand{\vhx}{\hat{\mvec{x}}}
\newcommand{\vhy}{\hat{\mvec{y}}}
\newcommand{\vhz}{\hat{\mvec{z}}}

\newcommand{\vty}{\tilde{\mvec{y}}}
\newcommand{\vtz}{\tilde{\mvec{z}}}
\newcommand{\vM}{\mvec{M}}
\newcommand{\vA}{\mvec{A}}
\newcommand{\vB}{\mvec{B}}

\newcommand{\vD}{\mvec{D}}

\newcommand{\vF}{\mvec{F}}

\newcommand{\vI}{\mvec{I}}
\newcommand{\vJ}{\mvec{J}}

\newcommand{\vP}{\mvec{P}}

\newcommand{\vX}{\mvec{X}}

\newcommand{\vzero}{\mvec{0}}

\newcommand{\calA}{\mathcal A}

\newcommand{\eat}[1]{}

\newcommand{\poly}{\mathrm{poly}}

\newcommand{\eps}{\epsilon}

\newcommand{\R}{\mathbb{R}}
\newcommand{\C}{\mathbb{C}}
\newcommand{\Z}{\mathbb{Z}}
\newcommand{\So}{\mathbf{S}^1}

\newcommand{\true}{\textnormal{\sc true}}
\newcommand{\false}{\textnormal{\sc false}}

\newcommand{\prob}{\mathop{\textnormal{Prob}}}

\newcommand{\be}{\begin{enumerate}}
\newcommand{\ee}{\end{enumerate}}
\newcommand{\bi}{\begin{itemize}}
\newcommand{\ei}{\end{itemize}}
\newcommand{\beq}{\begin{equation}}
\newcommand{\eeq}{\end{equation}}

\newcommand{\bp}{\begin{proof}}
\newcommand{\ep}{\end{proof}}
\newcommand{\bcor}{\begin{cor}}
\newcommand{\ecor}{\end{cor}}
\newcommand{\bthm}{\begin{thm}}
\newcommand{\ethm}{\end{thm}}
\newcommand{\blmm}{\begin{lmm}}
\newcommand{\elmm}{\end{lmm}}
\newcommand{\bdefn}{\begin{defn}}
\newcommand{\edefn}{\end{defn}}
\newcommand{\bprop}{\begin{prop}}
\newcommand{\eprop}{\end{prop}}
\newcommand{\bconj}{\begin{conj}}
\newcommand{\econj}{\end{conj}}
\newcommand{\bopm}{\begin{opm}}
\newcommand{\eopm}{\end{opm}}
\newcommand{\brmk}{\begin{rmk}}
\newcommand{\ermk}{\end{rmk}}

\newcommand{\inner}[1]{\langle #1 \rangle}

\theoremstyle{plain}                   
\newtheorem{thm}{Theorem}[section]
\newtheorem{lmm}[thm]{Lemma}
\newtheorem{prop}[thm]{Proposition}
\newtheorem{cor}[thm]{Corollary}

\theoremstyle{definition}              

\newtheorem{opm}[thm]{Open Problem}
\newtheorem{conj}[thm]{Conjecture}

\newtheorem{defn}[thm]{Definition}

\newtheorem{rmk}[thm]{Remark}
\newtheorem{claim}[thm]{Claim}

\newtheorem{observation}[thm]{Observation}

\newcommand{\bbox}{
\begin{center}
\begin{tabular}{|c|}
\hline
}
\newcommand{\ebox}{
\\
\hline
\end{tabular}
\end{center}
}

\newlength{\toppush}
\newcommand{\avg}{\mathbb{E}}
\newcommand{\Avg}[1]{\avg\left[{#1}\right]}

\algrenewcommand\algorithmicrequire{\textbf{Input:}}
\algrenewcommand\algorithmicensure{\textbf{Output:}}
\algrenewcommand\algorithmicwhile{\textbf{While}}
\algrenewcommand\algorithmicfor{\textbf{For}}
\algrenewcommand\algorithmicreturn{\textbf{Return}}
\algrenewcommand\algorithmicif{\textbf{If}}

\newcommand{\lst}{\mathcal{L}}

\newcommand{\abs}[1]{\left|#1\right|}

\definecolor{ao}{rgb}{0,0.5,0}

\newcommand{\eqdef}{\stackrel{\text{def}}{=}}

\newcommand{\Prob}[1]{\prob\left[{#1}\right]}

\newcommand{\brackets}[1]{\left[{#1}\right]}
\newcommand{\parens}[1]{\left({#1}\right)}

\newcommand{\ceils}[1]{\left\lceil{#1}\right\rceil}

\newcommand{\supp}{\mathrm{supp}}

\newcommand{\jac}[3]{\mathcal{J}^{({#1},{#2})}_{{#3}}}
\newcommand{\bessel}[2]{J_{{#1}}\parens{{#2}}}

\newcommand{\alphajac}{\alpha}
\newcommand{\betajac}{\beta}
\newcommand{\addjac}{\frac{\alphajac+\betajac+1}{2}}

\newcommand{\qprune}[1]{q_{\textsc{Prune}}\parens{{#1}}}
\newcommand{\tprune}[1]{T_{\textsc{Prune}}\parens{{#1}}}

\newcommand{\prunealg}{\textsc{Prune}}

\newcommand{\approxcos}{\textsc{ApproxArcCos}}
\newcommand{\jacforcos}{\textsc{QueryJacobiForCos}}
\newcommand{\prunespread}{\textsc{PruneNonSpread}}
\newcommand{\verify}{\textsc{Verify}}
\newcommand{\chk}{\textsc{Check}}
\newcommand{\refine}{\textsc{Refine}}

\newcommand{\trunc}[2]{\textsc{trunc}_{{#1}}\parens{{#2}}}
\newcommand{\sign}[1]{\textsc{sgn}\parens{{#1}}}

\newcommand{\assert}[1]{\textsf{Assertion {#1}}}

\newcommand{\tl}{\tilde{\ell}}

\newcommand{\cnst}{C}

\newcommand{\median}{\mathrm{Median}}

\newcommand{\evalpts}{\lambda}

\newcommand{\smll}{\textsc{large}}

\newcommand{\g}[2]{g_{{#1}}\parens{#2}}
\newcommand{\kap}[3]{\kappa_{#1,#2}\parens{#3}}

\usepackage{arydshln}
\usepackage{comment}
\usepackage[colorlinks, linkcolor=blue, citecolor=red]{hyperref}
\usepackage{multirow}
\usepackage{tablefootnote}
\usepackage{tabularx}
\usepackage{booktabs}

\begin{document}
\title{\textbf{Sparse Recovery for Orthogonal Polynomial Transforms}}
\author{\textsc{Anna Gilbert}\footnotemark[1] \and \textsc{Albert Gu}\footnotemark[2] \and \textsc{Christopher R\'{e}}\footnotemark[2] \and \textsc{Atri Rudra}\footnotemark[3] \and \textsc{Mary Wootters}\footnotemark[2]}
\date{\footnotemark[1]~~Department of Mathematics\\
University of Michigan\\
\texttt{annacg@umich.edu}\\
\vspace*{2mm}
\footnotemark[2]~~Department of Computer Science\\
Stanford University\\
\texttt{\{albertgu,chrismre,marykw\}@stanford.edu}\\
\vspace*{2mm}
\footnotemark[3]~~Department of Computer Science and Engineering\\
University at Buffalo\\
\texttt{atri@buffalo.edu}
}

\maketitle

\setcounter{page}{0}
\thispagestyle{empty}

\begin{abstract}

In this paper we consider the following \em sparse recovery \em problem. We have query access to a vector $\vx \in \R^N$ such that $\vhx = \vF \vx$ is $k$-sparse (or nearly $k$-sparse) for some orthogonal transform $\vF$.  The goal is to output an approximation (in an $\ell_2$ sense) to $\vhx$ in sublinear time.  This problem has been well-studied in the special case that $\vF$ is the Discrete Fourier Transform (DFT), and a long line of work has resulted in sparse Fast Fourier Transforms that run in time $O(k \cdot \mathrm{polylog} N)$.  However, for transforms $\vF$ other than the DFT (or closely related transforms like the Discrete Cosine Transform), the question is much less settled.

In this paper we give sublinear-time algorithms---running in time $\poly(k \log(N))$---for solving the sparse recovery problem for orthogonal transforms $\vF$ that arise from \em orthogonal polynomials. \em  More precisely, our algorithm works for any $\vF$ that is an orthogonal polynomial transform derived from  \em Jacobi polynomials\em.  The Jacobi polynomials are a large class of classical orthogonal polynomials (and include {\em Chebyshev} and {\em Legendre} polynomials as special cases), and show up extensively in applications like numerical analysis and signal processing.  One caveat of our work is that we require an assumption on the sparsity structure of the sparse vector, although we note that vectors with random support have this property with high probability. 

Our approach is to give a very general reduction from the $k$-sparse sparse recovery problem to the $1$-sparse sparse recovery problem that holds for any flat orthogonal polynomial transform; then we solve this one-sparse recovery problem for transforms derived from Jacobi polynomials.  Frequently, sparse FFT algorithms are described as implementing such a reduction; however, the technical details of such works are quite specific to the Fourier transform and moreover the actual implementations of these algorithms do not use the $1$-sparse algorithm as a black box.  In this work we give a reduction that works for a broad class of orthogonal polynomial families, and which uses any $1$-sparse recovery algorithm as a black box.

\end{abstract}

\newpage

\section{Introduction}

In this paper, we consider the following \em sparse recovery \em problem.  Suppose that we have query access to a vector $\vx \in \mathbb{R}^N$, which has the property that for a fixed \em orthogonal transform matrix \em $\vF$, $\vhx = \vF \vx$ is $k$-sparse (or approximately $k$-sparse, in the sense that $\vhx$ is close in $\ell_2$ distance to a $k$-sparse vector).  The goal is to recover an approximation $\vhz$ to $\vhx$, so that $\|\vhx - \vhz\|_2$ is small with high probability, as quickly as possible.

Variants of this problem have been studied extensively over several decades---we refer the reader to the book~\cite{Foucart2013} for many examples and references.
One particularly well-studied example is the \em sparse Fast Fourier Transform \em (sFFT)---see the survey~\cite{FourierSurvey2014} and the references therein.  In this case, the matrix $\vF$ is taken to be the Discrete Fourier Transform (DFT) and a long line of work has produced near-optimal results: algorithms with running time $O(k\polylog(N))$ and sample complexity $O(k\log N)$~\cite{HIKP12, IKP14, IK14}.

We study the sparse recovery problem for a more general class of transforms $\vF$ called \em orthogonal polynomial transforms, \em and in particular those that arise from \em Jacobi polynomials, \em a broad class of orthogonal polynomials (OPs).  Jacobi polynomials include as special cases many familiar families of OPs, including Gegenbauer and in particular Chebyshev, Legendre, and Zernike\footnote{To be more precise the {\em radial} component of a Zernike polynomial is a Gegenbauer and hence, a Jacobi polynomial.} polynomials, and the corresponding OP transforms appear throughout numerical analysis and signal processing.

Despite the progress on the sFFT described above, much remains unknown for general orthogonal polynomial transforms.
As discussed more in Section~\ref{sec:rel-work} below, 
the \em sample complexity \em of the sparse recovery problem is well understood, and the `correct' answer is known to be $\Theta(k \polylog(N))$ queries to $\vx$.  However, the algorithmic results that go along with these sample complexity bounds result in $\poly(N)$ time algorithms.  Our goal in this work will be \em sublinear \em time algorithms as well as sublinear sample complexity.  There are sublinear-time algorithms
available for the special cases of Chebyshev and Legendre polynomials that work by essentially reducing to the Fourier case.  For general Jacobi polynomials, such reductions are not available.  We elaborate in Appendix~\ref{app:hard} why reducing general Jacobi polynomials to the Fourier case does not seem easy.
There are also algorithms based on Prony's method, some of which work for quite general families of OPs~\cite{prony-general}.  However these general results require exact sparsity; to the best of our knowledge versions of Prony's method that are provably robust to noise are restricted to classes of OPs similar to the Fourier transform. 

\paragraph{Results.}
In this work, we give the first (to the best of our knowledge) sublinear-time algorithms with provable guarantees for the (approximately-)sparse recovery problem for general orthogonal transforms derived from Jacobi polynomials.  We discuss our results in more detail in Section~\ref{sec:results} and briefly summarize them here.  Our algorithms run in time $\poly(k \log(N))$ and given query access to $\vv = \vF^{-1} \vhv$, can find approximations to $\vhv$ when $\vhv$ is approximately $k$-sparse of an appropriate form.  More precisely,  we can handle vectors $\vhv = \vhx + \vhw$ where $\vhx$ is $k$-sparse with a `spread-out' support (made precise in Definition~\ref{def:sparse-separated}), and $\vhw$ is an adversarial noise vector with sufficiently small $\ell_2$ norm.  We obtain guarantees of the following flavor: for any such vector $\vv$, we can find $\vhz$ such that $\|\vhz - \vhx\|_2 \leq 0.01 \|\vhx\|_2$ with high probability.  

We note that these results are weaker than the results for the sFFT: our sample complexity and running time are polynomially larger, and we need stronger assumptions on the sparse signals.  However, we also note that the decade or so of work on the sFFT culminating in the results above began with similar results (see~\cite{GGIMS}, for example, in which the dependence on $k$ is an unspecified polynomial) and we hope that this work will similarly be a first step towards near-optimal algorithms for general orthogonal polynomial transforms. 

\paragraph{Techniques.}
Our techniques follow the outline of existing algorithms for the sFFT, although as we elaborate on in Section~\ref{sec:overview}, the situation for general Jacobi polynomials is substantially more complicated.  More precisely, we first give a very general reduction, which reduces the $k$-sparse case to the $1$-sparse case.  The idea of such a reduction was implicit in the sFFT literature, but previous work has relied heavily on the structure of the DFT.  Our reduction applies to a broad class of OPs including Jacobi polynomials.
Next, we show how to solve the $1$-sparse recovery problem for general Jacobi polynomials. The basic idea is to use known approximations of Jacobi polynomial evaluations by certain cosine evaluations~\cite{szego} in order to iteratively narrow down the support of the unknown $1$-sparse vector.  We give a more detailed overview of our techniques in Section~\ref{sec:overview}.

\paragraph{Organization.}
For the rest of the introduction, we briefly introduce orthogonal polynomial transforms, discuss previous work, and give a high-level overview of our approach.  After that we introduce the formal notation and definitions we need in Section~\ref{sec:prelims}, after which we state our results more formally in Section~\ref{sec:results}.  Then we prove our main results: the reduction from $k$ to $1$-sparse recovery is proved in Section~\ref{sec:redux}, the $1$-sparse recovery algorithm for Jacobi polynomials is presented in Section~\ref{sec:jacobi-1sps}, and the resulting $k$-sparse recovery algorithm for Jacobi polynomials is presented in Section~\ref{sec:k-sps-jacobi}.

\subsection{Orthogonal Polynomial Transforms}

Orthogonal polynomials (OPs) play an important role in classical applied mathematics, mathematical physics, and the numerical analysis necessary to simulate solutions to such problems. 
We give more precise definitions in Section~\ref{sec:prelims} but briefly a family of orthogonal polynomials $p_0(X), p_1(X), \ldots$ is a collection of polynomials defined on an interval ${\cal D}$ of $\R$, that are pairwise orthogonal with respect to a (non-negative) weight function~$w$.

In this work we study \em Jacobi polynomials \em (defined formally in Section~\ref{sec:prelims}), which are a very general class of orthogonal polynomials.  These include Chebyshev polynomials, Legendre polynomials, Zernike polynomials and more generally Gegenbauer polynomials.  These OP families show up in many places.
For example, Zernike polynomials are a family of orthogonal polynomials on the unit disk that permit an analytic expression of the 2D Fourier transform on the disk. They are used in optics and interferometry~\cite{Tango1977}. They can be utilized to extract features from images that describe the shape characteristics of an object and were recently used for improved cancer imaging~\cite{yu2016}. Different families of orthogonal polynomials give rise to different quadrature rules for numerical integration~\cite{R81,quadrature}. Specifically, Chebyshev polynomials are used for numerical stability (see e.g.\ the ChebFun package~\cite{chebfun}) as well as approximation theory (see e.g.\ Chebyshev approximation~\cite{cheb-approx}). Chebyshev polynomials also have certain optimal {\em extremal} properties, which has resulted in many uses in theoretical computer science, including in learning theory, quantum complexity theory, linear systems solvers, eigenvector computation, and more~\cite{musco}. Further, Jacobi polynomials form solutions of certain differential equations~\cite{jacobi}. More recent applications include Dao et al.'s~\cite{DaoDeSaRe2017} use of orthogonal polynomials to derive quadrature rules for computing kernel features in machine learning. 

Orthogonal polynomials naturally give rise to (discrete) orthogonal polynomial transforms.  
Suppose that $\vF$ is an $N \times N$ matrix, with each column corresponding to an orthogonal polynomial $p_0, \ldots, p_{N-1}$  and each row an evaluation point $\lambda_0, \ldots, \lambda_{N-1}$ in a suitable domain and suitably normalized so that it is an orthogonal matrix (Definition~\ref{def:OP}).
A familiar example might be the DFT: in this language, the DFT matrix is defined by the polynomials $1, X, X^2, \ldots, X^{N-1}$, evaluated at points $\lambda_j = \omega^j$ where $\omega$ is the $N$th root of unity.\footnote{We note that in this work we consider a setting slightly different than this example, where $\mathcal{D} = [-1,1]$ rather than $\So$.}
Like the Fourier Transform, it is known that all OP transforms admit `fast' versions, allowing matrix-vector multiplication in time $O(N\log^2(N))$~\cite{driscoll}.\footnote{We note that even though the work of~\cite{driscoll} has in some sense solved the problem of computing any OP transform in near-linear time, many practical issues still remain to be resolved and the  problem of computing OP transforms in near-linear time  has seen a lot of research activity recently. We just mention two recent works~\cite{jacobi-1,jacobi-2} that present near-linear time algorithms for the Jacobi polynomial transforms (and indeed their notion of {\em uniform Jacobi transform} corresponds exactly to the Jacobi polynomial transform that we study in this paper). However, these algorithms inherently seem to require at least linear-time and it is not clear how to convert them into sub-linear algorithms, which is the focus of our work.}
Thus, our problem of sparse recovery for OP transforms is a natural extension of the sFFT problem, with applications to the areas mentioned above.

\subsection{Related Work}
\label{sec:rel-work}

As previously described, there has been a great deal of work on the sFFT; we refer the reader to the survey \cite{FourierSurvey2014} for an overview.  There has also been work on non-Fourier OP transforms.  We break up our discussion below into discussion on the \em sample complexity \em (which as mentioned above is largely settled) and the \em algorithmic complexity \em (which remains largely open).

\paragraph{Sample complexity.}
The sample complexity of OP transforms $\vF$ has been largely pinned down by the \em compressed sensing \em literature.
For example, suppose that 
 $\vF \in \mathbb{R}^{N \times N}$ is any orthogonal and sufficiently flat matrix, in the sense that none of the entries of $\vF$ are too large.   Then a result of Rudelson and Vershynin (and a sharpening of their result by Bourgain) shows that $m = O(k\log{k}\log^2{N})$ samples suffice to establish that the matrix $\Phi \in \mathbb{R}^{m \times N}$ (which is made up of $m$ sampled rows from $\vF^T$) has the \em Restricted Isometry Property \em (RIP)~\cite{rip-B14,RV}. 
Finding $\vhx = \vF \vx$ from samples of $\vF$ of corresponds to the problem of finding an (approximately) $k$-sparse vector $\vhx$ from the linear measurements $\Phi \vhx$, which is precisely the  compressed sensing problem.  It is known that if $\Phi$ satisfies the RIP, then this can be solved (for example with $\ell_1$ minimization) in time $N^{O(1)}$.
We note that very recently a result due to B{\l}asiok et al. show that this is essentially tight, in that $O(k\log^2{N})$ queries (for a certain range of $k$) to $\vx$ are not enough to compute a $k$-sparse approximation of $\vF\vx$~\cite{Blasiok2019}.
Bounds specific to the DFT over finite fields can be found in~\cite{Rao2019}.

Foucart and Rauhut~\cite{Foucart2013} show that if the orthogonal polynomials satisfy a Bounded Orthogonal System (BOS) which are suitably flat, then if the $m$ evaluation points $\lambda_j$ are chosen uniformly at random proportional to the weight function $w$, then the $m \times N$ matrix $\Phi$ defined by normalizing $\vP_N[i,j] = p_j(\lambda_i)$ appropriately satisfies the RIP with high probability provided that $m$ has an appropriate dependence on $N,k,\epsilon,$ and the flatness of the matrix, and this again gives an $N^{O(1)}$-time algorithm to solve the sparse recovery problem.

Rauhut and Ward~\cite{RW12} show that for Jacobi polynomial transforms if the evaluation points were picked according to the {\em Chebyshev measure}, then with $O(k \, \mathrm{polylog}{N})$ random measurements, the corresponding matrix has the RIP (note that the Foucart and Rauhut sample the evaluation points according to the measure of orthogonality for the Jacobi polynomials, which in general is {\em not} the Chebyshev measure). This result again does not give a sub-linear time algorithm but was used in the result of~\cite{legendre-sparse} which we describe below.

While these approaches can give near-optimal sample complexity, they do not give sublinear-time algorithms.  In fact, it is faster to compute $\vhx$ exactly by computing $\vF \vx$, if we care only about the running time and not about sample complexity~\cite{driscoll}.  Thus, we turn our attention to sublinear-time algorithms.

\paragraph{Sublinear-time algorithms for OP transforms.}

There have been several works generalizing and building on the sFFT results mentioned above.  One direction is to the multi-dimensional DFT (for example in~\cite{IK14,KVZ19}). 
Another direction is to apply the sFFT framework to orthogonal polynomials with similar structure.
One example is Chebyshev polynomials and the Discrete Cosine Transform (DCT).  
It was observed in~\cite{legendre-sparse} (also see Appendix~\ref{sec:dct-reduction}) that this can be reduced to sFFT in a black box manner, solving the sparse recovery problem for Chebyshev polynomials and the DCT.
A second example of OP transforms which can essentially be reduced to the sFFT is \em Legendre polynomials. \em
Hu et al.~\cite{legendre-sparse} seek to recover 
an unknown $k$-term Legendre polynomial (with highest magnitude degree limited to be $N/2$), defined on $[-,1,1]$, from samples.  
They give a sublinear two-phase algorithm: in the first phase, they 
reduce $k$-sparse-Legendre to sFFT to identify a set of candidate Legendre polynomials. The second phase uses the RIP result for BOS to produce a matrix that is used to estimate the coefficients of the candidate Legendre polynomials. 
We note that in this work the setting is naturally continuous, while ours is discrete.  

Choi et al.~\cite{Choi2018} study higher dimensions and obtain sublinear-time algorithms for more general harmonic expansions in multiple dimensions. 
The results of \cite{Choi2018} complement our work.  More precisely, that work shows how to use any algorithm for a univariate polynomial transform to design an algorithm for a multi-variate polynomial transform where the multi-variate polynomials are products of univariate polynomials in the individual variables.  Thus our improvements for univariate polynomial transforms can be used with \cite{Choi2018}.

Finally, there are sparse OP transforms based on Prony's method.  
The work~\cite{prony-general} extends Prony's method to a very general setting, including Jacobi polynomials, and gives an algorithm that requires only $O(k)$ queries to recover exactly $k$-sparse polynomials.  However, these general results work only for exact sparsity and are in general not robust to noise.  There has been work extending and modifying these techniques to settings with noise (for example,~\cite{HuaSarkar1990,Potts2010}), but to the best of our knowledge the only provable results for noise are for either the sFFT or closely related polynomial families.  We note that \cite{Potts2016} presents a Prony-like algorithm for Legendre and Gegenbauer polynomials and demonstrates empirically that this algorithm is robust to noise, although they do not address the question theoretically.

\subsection{Technical overview}
\label{sec:overview}

Our technical results have two main parts.  First, inspired by existing approaches to the sFFT, we provide a general reduction from the $k$-sparse recovery problem to the $1$-sparse recovery algorithm, which works for any family of OPs that is sufficiently `flat:' that is, no entry of the matrix $\vF$ is too large.  Second, we provide a $1$-sparse recovery algorithm for Jacobi polynomials.  We give an overview of both parts below.

For what follows, let $\vF$ be an orthogonal matrix. For simplicity in this overview we will assume that there is no noise.  That is, we want to compute the exactly $k$-sparse $\vhx=\vF\vx$ given query access to $\vx$.  However, we note that our final results do work for approximately $k$-sparse vectors $\vhv = \vhx + \vhw$ provided that $\|\vhw\|_2$ is sufficiently small.

\subsubsection{Reduction to one-sparse recovery}
We give a general reduction from the $k$-sparse recovery problem to the one-sparse recovery problem, which works for a broad class of OP families defined on a finite interval.\footnote{\label{fn:hermite}We note that our results do not (yet) work for the case when the orthogonality is defined over an infinite interval. In particular, our reduction does not work for the Hermite and Laguerre polynomials.}   At a high level, the idea is as follows.  Suppose that $\vhx = \vF\vx$ is $k$-sparse and $\vb \in \mathbb{R}^N$ is a `filter': at this stage it is helpful to think of it like a boxcar filter, so $\vb$ is $1$ on some interval $I$ and zero outside of that interval.  If we choose this interval randomly, we might hope to isolate a single `spike' of $\vhx$ with $\vb$: that is, we might hope that $\vD_{\vb} \vhx$ is one-sparse, where $\vD_{\vb}$ is the diagonal matrix with $\vb$ on the diagonal.  Suppose that this occurs.  In order to take advantage of this one-sparse vector with a black-box solution to the one-sparse recovery problem, we would need query access to the vector $\vF^{-1} \vD_{\vb} \vhx = \vF^{-1} \vD_{\vb} \vF \vx$, while what we have is query access to $\vx$.  Thus, we would like to design $\vb$ so that $\vF^{-1} \vD_{\vb} \vF$ is \em row-sparse. \em  This would allow us to query a position of $\vF^{-1} \vD_{\vb} \vhx$ using only a few queries from $\vx$.

One of our main technical contributions is showing how to design such a vector $\vb$, so that $\vb$ approximates a boxcar filter and so that $\vF^{-1} \vD_{\vb} \vF$ is row-sparse for \em any \em OP transform $\vF$.  

Then, given this filter, we can iteratively identify and subtract off `spikes' in $\vhx$ until we have recovered the whole thing.
Of course, the actual details are much more complicated than the sketch above.  First, the one-sparse solver might have a bit of error, which will get propagated through the algorithm.  Second, in our analysis the vector $\vhx$ need not be exactly $k$-sparse.  Third, $\vb$ will only approximate a boxcar filter, and this is an additional source of error that needs to be dealt with.  We will see how to overcome these challenges in Section~\ref{sec:redux}.

For the reader familiar with the sFFT, this approach might look familiar: most sFFT algorithms work by using some sort of filter to isolate single spikes in an approximately sparse signal.  Below, we highlight some of the challenges in extending this idea beyond the Fourier transform.  Some of these challenges we have overcome, and one we have not (yet) overcome.  We mention this last open challenge both because it explains the assumption we have to make on the sparsity structure of $\vhx$, and also because we hope it will inspire future work.

\paragraph{Challenge 1: Choice of filter.}  One key difficulty in extending sFFT algorithms to general orthogonal polynomials is that the filters used in the sFFT approach are very specific to the Fourier transform.  Indeed, much of the progress that has been made on that problem has been due to identifying better and better choices of filter specialized to the Fourier transform.  In order to find filters that work for \em any \em OP family, we take a different approach and construct a filter out of low-degree Chebyshev polynomials.  Then we use the orthogonality properties of the OP family to guarantee that $\vF^{-1} \vD_{\vb} \vF$ has the desired sparsity properties.

\paragraph{Challenge 2: Explicit black-box reduction.}  Because our goal is generality (to as broad a class of OPs as possible), we give an explicit reduction that uses a $1$-sparse solution as a black box.  To the best of our knowledge, existing work on the sFFT does not explicitly do this: a reduction of this flavor is certainly implicit in many of these works, and even explicitly given as intuition, but we are not aware of an sFFT algorithm which actually uses a $1$-sparse recovery algorithm as a black box.

\paragraph{Challenge 3: Equi-spaced evaluation points.} The evaluations points in DFT and the DCT are equispaced (in the angular space). This fact is crucially exploited in sFFT algorithms (as well as the reduction of DCT to DFT---see Appendix~\ref{sec:dct-reduction} for more details on the reduction). Unfortunately, the roots of Jacobi polynomials are no longer equally spaced---see Appendix~\ref{app:hard} for why this is a barrier to reducing $k$-sparse recovery of Jacobi polynomials directly to $k$-sparse recovery for DCT or DFT. However, it is known that the roots of Jacobi polynomials are `spread out,' (in a sense made below precise in Definition~\ref{def:roots-spread}), and we show that this property is enough for our reduction. In fact, our reduction from $k$-sparse recovery to $1$-sparse recovery works generally for any `flat' OP family with `spread out' roots. 

\paragraph{(Open) Challenge 4: Permuting the coordinates of $\vhx$.}  In the approach described above, we hoped that an interval $I$ would `isolate' a single spike.  In the sFFT setting, this can be achieved through a permutation of the coordinates of $\vhx$.  In our language, in the sFFT setting it is possible to define a random (enough) permutation matrix $\vP$ so that $\vP \vhx$ has permuted coordinates, and so that $\vF^{-1} \vD_{\vb} \vP \vF$ is row-sparse---this argument crucially exploits the fact that the roots of unity are equispaced in the angle space.  This means that not only can we sample from the one-sparse vector $\vD_{\vb} \vhx$, but also we can sample from $\vD_{\vb} \vP \vhx$, and then there is some decent probability that any given spike in $\vhx$ is isolated by $\vb$.  However, we have not been able to come up with (an approximation to) such a $\vP$ that works in the general OP setting.  This explains why we require the assumption that the support of $\vhx$ be reasonably `spread out,' so that we can hope to isolate the spikes by $\vb$.  This assumption is made precise in Definition~\ref{def:sparse-separated}.  We note that if such a $\vP$ were found in future work, this would immediately lead to an improved $k$-sparse recovery result for Jacobi polynomials, which would work for arbitrary sparse signals $\vhx$.

\subsubsection{A one-sparse recovery algorithm for Jacobi polynomials}
With the reduction complete, to obtain a $k$-sparse recovery algorithm for general Jacobi polynomials we need to solve the one-sparse case.
We give an overview of the basic idea here.  First, we note that via well-known approximations of Jacobi polynomials~\cite{szego}, one can approximate the evaluation of any Jacobi polynomial at a point in $(-1,1)$ by evaluating the cosine function at an appropriate angle. 
Using some standard local error-correcting techniques (for example, computing $\cos(A)$ via $\cos{A} =\frac{\cos(A+B)+\cos(A-B)}{2\cos{B}}$ for a random $B$), we reduce the $1$-sparse recovery problem to computing $\theta$ from noisy values of $\cos(w\theta)$ for some integers $w\ge 1$. Since the reduction is approximate, some care has to be taken to handle some corner cases where the approximation does not hold. In particular, we have to figure out for which real numbers $y\in [0,N)$, does its orbit $\inner{xy}$ for $x\in\Z_N$ have small order. We give a result to handle this, which to the best of our knowledge (and somewhat surprisingly) seems to be new.\footnote{We thank Stefan Steinerberger for showing us a much simpler proof than our original more complicated proof, which also gave worse parameters.}
With this out of the way, our algorithm to compute the value of $\theta$ from the evaluations $\cos(w\theta)$ is based on the following idea.  Assuming we already know $\cos(\theta)$ up to $\pm \eps$, we get a noise estimate of $\theta$ (which lives in the range $\arccos(\cos(\theta)\pm \eps)$)  and then use the evaluations at $w>1$ to `dilate' the region of $[0,\pi]$ where we know $\theta$ lies, reducing $\eps$.  We proceed iteratively until the region of uncertainty is small enough that there are only  
$O(1)$ possibilities remaining, which we then prune out using the fact that $\vF$ is orthogonal and flat, in the sense that none of its entries are too large. (We note that proving $\vF$ is flat needs a bit of care. 
In particular, we need a sharper bound on Jacobi polynomials (than the cosine approximation mentioned above) in terms of Bessel functions to prove that {\em all} entries of $\vF$ are small.)
Similar ideas have been used for $1$-sparse recovery for the DFT (for example, in~\cite{HIKP12}),
although our situation is more complicated than the DFT because working with cosines instead of complex exponentials means that we lose sign information about $\theta$ along the way.

\section{Background and Preliminaries}
\label{sec:prelims}

\subsection{Notation}

We use bold lower-case letters ($\vx, \vy$) for vectors and bold upper-case letters ($\vP, \vF$) for matrices. Non-bold notation $x,y,U$ is used for scalars in $\mathbb{R}$.
In general, if there is a given transform $\vF$ we are considering, then the notation $\vhx \in \R^N$ indicates $\vF \cdot \vx$. 
We use the notation $\vx[i]$ or $\vX[i,j]$ to index into a vector or matrix, respectively.
All of our vectors and matrices are $0$-indexed, i.e. the entries of a vector $\vx \in \R^N$ are $\vx[0], \dots, \vx[N-1]$.
We use $[N]$ to denote the set $\{0, \ldots, N-1\}$. Given a subset $S\subset [N]$, we will denote the complement set (i.e. $[N]\setminus S$) by $S^c$.

Given a vector $\vx\in\R^N$ and an integer $1\le s\le N$, we define $\smll(s,\vx)$ to be the magnitude of the $s$th largest value in $\vx$ (by absolute value).

For any vector $\vu\in\R^N$, we define $\vD_{\vu} \in \R^{N \times N}$ as the diagonal matrix with $\vu$ on its diagonal. For a diagonal matrix $\vD$, and any real $\alpha$ we denote $\vD^\alpha$ to denote the diagonal matrix with the $(i,i)$ entry being $\parens{\vD[i,i]}^\alpha$.
Given a vector $\vx \in \R^N$ and set $S \subseteq [N]$, $\vx_S$ denotes the vector $\vx$ where all entries out of $S$ are masked to $0$.
For $\vx \in \R^N$, $\supp(\vx) \subseteq [N]$ denotes the support (i.e. the set of non-zero positions) of $\vx$.

We use $x \pm h$ to refer to either the interval $[x-h, x+h]$ or a point in this interval, whichever is clear from context.
Similarly, if $S$ is an interval $[a,b]$ then $S \pm h$ is the interval $[a-h, b+h]$.

When stating algorithms, we use superscript notation to denote query access.  That is $\mathcal{A}^{(\vx)}(\vz)$ takes input $\vz$ and has query access to $\vx$.

We use the notation $f(n) \lesssim g(n)$ to mean that there is some constant $C$ so that, for sufficiently large $n \geq n_0$, $f(n) \leq C g(n)$.

\subsection{Orthogonal Polynomials}

For the remainder of this paper, we consider polynomials $p_0(X),p_1(X),\dots$ that form a normalized orthogonal polynomial family with respect to any compactly supported measure $w(X)$.
By suitably scaling and translating $X$, we can ensure that the orthogonality is on $[-1,1]$.\footnote{See footnote \ref{fn:hermite}.}   
In particular $\deg(p_i) = i$ and
for any $i,j\ge 0$,
\begin{equation}
\label{eq:OP-cond}
\int_{-1}^1 p_i(X)p_j(X)w(X)dX=\delta_{i,j},
\end{equation}
where $\delta_{i,j}=1$ if $i=j$ and $0$ otherwise.

Then for given $N$ evaluation points $\evalpts_0,\dots\evalpts_{N-1}$, define the orthogonal polynomial transform $\vP_N$ as follows. 
For any $0\le i,j<N$, we have
\[\vP_N[i,j]=p_j(\evalpts_i).\]
In other words, the rows of $\vP_N$ are indexed by the evaluation points and the columns are indexed by the polynomials.

For the rest of the paper, assume $\evalpts_0\le \evalpts_1\le \cdots\le \evalpts_{N-1}$ are the roots of $p_N(X)$. Then it is well-known (see e.g.~\cite{szego}) that
\begin{itemize}
  \item The roots lie in the support of the measure (i.e.\ $\evalpts_i \in [-1, 1]$) and are distinct (i.e. $\evalpts_0<\evalpts_1<\cdots<\evalpts_{N-1}$).
  \item There exists weights Gaussian quadrature weights
$w_{\ell}=\frac 1{\sum_{j=0}^{N-1}p_j\parens{\evalpts_{\ell}}^2}, i = 0, \dots, N-1$
such that for any polynomial $f(X)$ of degree at most $2N-1$,
\begin{equation}
\label{eq:quadrature}
\int_{-1}^1 f(X)w(X)dX=
\sum_{\ell=0}^{N-1} f(\evalpts_{\ell})\cdot w_{\ell}.
\end{equation}
\end{itemize}

We are now ready to define the orthogonal matrix corresponding to $\vP_N$ that we deal with in this paper:

\begin{defn}\label{def:OP}
  Let $p_0(X), \dots, p_{N-1}(X),\dots$ be an orthogonal polynomial family, $\evalpts_0, \dots, \evalpts_{N-1}$ be the roots of $p_N(X)$, and $w_0, \dots, w_{N-1}$ be the Gaussian quadrature weights.
  Define $\vD_{\vw}$ to be the diagonal matrix with $w_0,\dots,w_{N-1}$ on its diagonal, and
  \[
    \vF_N = \vD_\vw^{\frac{1}{2}} \vP_N.
  \]
\end{defn}
Note that by~\eqref{eq:OP-cond} and~\eqref{eq:quadrature},
\[
  \vF_N^T\vF_N = \vP_N^T \vD_\vw \vP_N = \vI_N,
\]
so $\vF_N$ is an orthogonal matrix.
In particular, 
\[ \vP_N^T\vD_\vw\vP_N[i,j] = \sum_{k=0}^{N-1} p_i(\lambda_k) w_k p_j(\lambda_k) = \int_{-1}^1 p_i(X)p_j(X)w(X)\,dX = \delta_{i,j}. \]

Note that since $\vF_N$ is orthogonal, by definition we have 
\[\vF_N^{-1}=\vF_N^T.\]

\subsubsection{Jacobi Polynomials and Special Cases}
In this section we define Jacobi Polynomials, our main object of interest, and point out a few special cases.
We note that families of named orthogonal polynomials $\{p_i(X)\}$ are sometimes defined through different means, hence are normalized differently up to constants.
The corresponding discrete orthogonal polynomial transform (e.g. Discrete Legendre Transform) frequently refers to multiplication by $\vP$ instead of $\vF$.
In these cases, the transform satisfies $\vP_N^T \vD_\vw \vP_N = \vD$ for a diagonal matrix $\vD$ corresponding to the normalization.
The transform $\vF_N = \vD_\vw^{\frac{1}{2}} \vP \vD^{-\frac{1}{2}}$ we consider (note that this matrix is indeed orthogonal) is thus equivalent up to diagonal multiplication.

\paragraph{Jacobi polynomials.}

Jacobi polynomials are indexed by two parameters $\alphajac,\betajac>-1$ and these are polynomials $\set{P_j^{(\alphajac,\betajac)}}_{j\ge 0}$ that are orthogonal with respect to the measure
\[w^{(\alphajac,\betajac)}(X)=(1-X)^{\alphajac}\cdot(1+X)^{\betajac}\]
in the range $[-1,1]$. This definition is not normalized, in the sense that we have  $\vP_N^T \vD_\vw \vP_N = \vD$, where 
\[ \vD[j,j]=\frac{2^{\alphajac+\betajac+1}}{2j+\alphajac+\betajac+1}\cdot \frac{\Gamma(j+\alphajac+1)\Gamma(j+\betajac+1)}{\Gamma(j+1)\Gamma(j+\alphajac+\betajac+1)}\]
 (see~\cite[Pg. 68, (4.3.3)]{szego}). We will come back to the normalization in Section~\ref{sec:jacobi-1sps} (cf.\ Corollary~\ref{cor:jac-U-bound}).

We record three well-known special cases: Chebyshev polynomials (of the first kind) are special case of $\alphajac=\betajac=-\frac 12$ and Legendre polynomials are the special case of $\alphajac=\betajac=0$ (up to potentially a multiplicative factor that could depend on the degree $j$).
Another notable special case of Jacobi polynomials are the {\em Gegenbauer} or {\em ultraspherical polynomials} ($\alphajac=\betajac$).

\paragraph{Chebyshev polynomials of the 1st kind.}
The Chebyshev polynomials of the 1st kind are orthogonal with respect to the weight measure $w(X) = (1-X^2)^{-\frac{1}{2}}$.

The normalized transform $\vF_N$ has the closed form
\[
  \vF_N[i,j] =
  \begin{cases}
    \sqrt{\frac{1}{N}} & j = 0 \\
    \sqrt{\frac 2 N} \cdot \cos\left[ \frac{\pi}{N} j \left( i+\frac{1}{2} \right) \right] & j = 1, \dots, N-1.
  \end{cases}
\]
This is a variant of the Discrete Cosine Transform (DCT-III, or the inverse DCT). It is well-known that the DCT-III can be `embedded' into a DFT of twice the dimension, and we work out some of the details of how to use the sparse FFT to compute a sparse DCT in Appendix~\ref{sec:dct-reduction}.

\paragraph{Legendre polynomials.}

Legendre polynomials are orthogonal with respect to the uniform measure i.e. $w(X)=1$ and play a critical role in multipole expansions of potential functions (whether electrical or gravitational) in spherical coordinates. They are also important for solving Laplace's equation in spherical coordinates.

\subsubsection{Roots of Orthogonal Polynomials}

Since $\lambda_i \in [-1,1]$ for all $i$, there is a unique $\theta_i \in [0, \pi]$ such that $\lambda_i = \cos \theta_i$.
Our reduction holds for orthogonal polynomials that have roots that are `well-separated' in this angle space:

\begin{defn}
\label{def:roots-spread}
Let $0 < C_0 < C_1$.
A family of orthogonal polynomials $p_0(X), p_1(X), \ldots$ is $(C_0,C_1, \gamma_0)$-{\em dense} if for all large enough $d$, the following holds.

Let $\evalpts_0,\ldots, \evalpts_{d-1}$ be the roots of $p_d$, and $\theta_i = \arccos \evalpts_i$.
Then for any $i \in [d]$, for any $\gamma \geq \gamma_0/d$:
\[ C_0 \gamma d \leq \abs{ \set{ \theta_0, \ldots, \theta_{d-1}} \cap \brackets{ \theta_i - \frac{\gamma}{2} , \theta_i + \frac{\gamma}{2} } } \leq C_1 \gamma d. \]
\end{defn}

It turns out that any family of Jacobi polynomials has the required property: their roots are spaced out such that $\theta_{\ell}$ is close to $\ell\pi/N$ (Theorem~\ref{thm:jacobi-roots}).

\subsection{Sparse Recovery Problem}
\label{sec:sr-def}

We will consider approximately $k$-sparse vectors $\vhv = \vhx + \vhw$, where $\vhx$ is $k$-sparse and $\|\vhw\|_2$ is sufficiently small.  We will require that $\vhx$ has a `spread out' support, defined as follows.

\begin{defn}
\label{def:sparse-separated}
Let $k \in [N]$ and $0\le \sigma < 1$.
We say that a vector $\vx\in\R^N$ is $(k,\sigma)$-{\em sparsely separated} if there are $k$ non-zero locations in $\vx$ and any two non-zero locations are more than $\sigma N$ indices apart.
\end{defn}

It is not hard to see that a vector $\vx$ with random support of size $k$ is, with constant probability, $\parens{k,\Omega\parens{\frac{1}{k^2}}}$-sparsely separated.

In our reduction, we will reduce the $k$-sparse recovery problem to the special case of $k=1$. Next, we define some notation for the $1$-sparse case.

\begin{defn}
\label{def:1-sparse} We say that the matrix $\vF_N$ has an $(N,\eps,\delta,\mu)$ one-sparse recovery algorithm with query complexity $Q(N,\eps,\delta,\mu)$ and time complexity $T(N,\eps,\delta,\mu)$
if there exists an algorithm  $\calA$ with the properties below:

For all $\vy$ so that $\vhy=\vF_N \vy$ can be decomposed as
\[\vhy=\vty+\vw,\]
where ${\vty}=v\cdot \ve_h$ is $1$-sparse and
\[\norm{2}{\vw}\le \eps\abs{v},\]
we have:
\begin{enumerate}
\item $\calA$ makes at most $Q(N,\eps,\delta,\mu)$ queries into $\vy=\vF^{-1}\parens{v\cdot\ve_h+\vw}$.
\item With probability at least $1-\mu$, $\calA$ outputs $\tilde{v}\cdot \ve_h$ with $|v-\tilde{v}|\le \delta\abs{v}$ in time $T(N,\eps,\delta,\mu)$.
\end{enumerate}
\end{defn}
\paragraph{Pre-processing time.}
Our algorithm requires some pre-processing of $\vF_N$.  Our pre-processing step is given in Algorithm~\ref{alg:preprocessing} and involves computing the roots $\lambda_1, \ldots, \lambda_N$ of $p_N$ and storing them in an appropriate data structure, and additionally forming and storing some matrices that we will use in our algorithm.
Finding the roots and creating the data structure can be done in time $\poly(N)$, and 
the rest of the pre-processing step also takes time $\poly(N)$.  We note that this is an up-front cost that needs to be only paid once.

\paragraph{Precision.} We note that we need to make certain assumptions on size of the entries in $\vhv$ since otherwise we would not even be able to read coefficients that are either too large or too small and need $\omega(\log{N})$ bits to represent. Towards this end we will make the standard assumption that $\norm{2}{\vhv}=1$. In particular, this allows us to ignore any coefficients that are smaller than say $\frac 1N$ since their contribution to $\norm{2}{\vhv}$ is at most $\frac 1{\sqrt{N}}$, which will be too small for our purposes.\footnote{More generally, we can ignore smaller coefficients as long as they are polynomially large.} In particular, this implies that we only have to deal with numbers that need $O(\log{N})$ bits and as is standard in the RAM model, basic arithmetic operations on such numbers can be done in $O(1)$ time. We will implicitly assume this for the rest of the paper (except in the proof of Lemma~\ref{lem:jac-to-cos-correct}, where we will explicitly make use of this assumption).

\section{Results}
\label{sec:results}

In this section we state our main results.  These results follow from more detailed versions which are stated with the proofs of these results.

We start off with our main result for Jacobi polynomials.  We state an informal version here, and refer the reader to Corollary~\ref{cor:jacobi-k-sps} for the formal result. 
\bthm[General Sparse Recovery for Jacobi Polynomial Transform, Informal]
\label{thm:jacobi-intro}
Fix arbitrary parameters $\alphajac,\betajac>-1$ for Jacobi polynomials and let $\vJ^{(\alphajac,\betajac)}_N$ be the $N \times N$ orthogonal matrix that arises from it as in Definition~\ref{def:OP}.
Then there is an algorithm \textsc{Recover} that does the following.
Let $\vv = \vx + \vw$ where $\vhx = \vJ^{(\alphajac,\betajac)}_N \vx$ is $(k, C_1/k^2)$-sparsely separated, and suppose that
$ \|\vhw\|_2 \lesssim \delta \min_{h \in \supp(\vhx)} |\vhx[h]|. $
Then with probability at least $0.99$, \textsc{Recover} outputs $\vhz$ such that
\[ \|\vhx - \vhz\|_2 \lesssim \delta \|\vhx\|_2, \]
with $\poly\parens{\frac{k\log{N}}\delta}$ queries and running time $\poly\parens{\frac{k\log{N}}\delta}$.
\ethm

\begin{rmk}
The requirement on the noise term might be bad if one entry of $\vhx$ is extremely small compared to the rest.  However in this case we can decrease $k$ and add the very small entries of $\vhx$ to the noise term $\vhw$ resulting in a potentially better guarantee.  We note that our algorithm iteratively finds the large components of $\vhx$ and in fact has a mechanism for stopping early when all of the `large-enough' entries have been found.
\end{rmk}

To prove the above result, we first reduce the $k$-sparse recovery problem to $1$-sparse recovery problem, in the presence of a small amount of noise.
Next, we present an informal statement of our reduction.
See Theorem~\ref{thm:redux-main} for the formal result.

\bthm[Main Reduction, Informal]
\label{thm:redux-main-intro}
Let $p_1, \ldots, p_N$ be a $(C_0, C_1, \gamma_0)$-dense orthogonal polynomial family, and let $\vF_N$ be the $N \times N$ orthogonal matrix that arises from it as in Definition~\ref{def:OP}.
Suppose that $|\vF^{-1}_N[i,j]| \lesssim 1/\sqrt{N}$ for all $i,j \in [N]$.
Suppose that for some sufficiently small $\delta > 0$, 
$\vF_N$ has a $\parens{N,O(\delta),\delta,O(C_0/k^2)}$ one-sparse recovery algorithm with query complexity $Q$ and running time $T$. 
 
Then there is an algorithm \textsc{Recover} that does the following.
Let $\vv = \vx + \vw$ where $\vhx = \vF_N \vx$ is $(k, C_1/k^2)$-sparsely separated, and suppose that
$ \|\vhw\|_2 \lesssim \delta \min_{h \in \supp(\vhx)} |\vhx[h]|. $
Then with probability at least $0.99$, \textsc{Recover} outputs $\vhz$ so that
\[ \|\vhx - \vhz\|_2 \lesssim \delta \|\vhx\|_2, \]
with $\poly(k/\delta C_0)Q$ queries and running time $\poly(k/\delta C_0)T$.
\ethm

The final algorithmic piece missing from the result above is the algorithm for $1$-sparse recovery. We provide this missing piece for Jacobi polynomials (see Theorem~\ref{thm:jacobi-1-sps} for the formal statement):
\begin{thm}[$1$-Sparse Recovery for Jacobi Transform, Informal]
\label{thm:jacobi-1-sps-results}
There exists a universal constant $C$ such that the following holds.
Consider the Jacobi transform for any fixed parameters $\alphajac,\betajac>-1$.
There exists an $(N,\eps,C\cdot\eps,\gamma)$ $1$-sparse recovery algorithm for the Jacobi transform that makes $\poly\parens{\log\parens{\frac N{\gamma}}\cdot \frac 1\eps}$ queries and takes time $\poly\parens{\log\parens{\frac N{\gamma}}\cdot \frac 1\eps}$.
\end{thm}

\subsection{Open Questions}
Before we dive into the proofs of the results above, we list a few questions left open by our work. 

\begin{enumerate}
\item First, it is a natural to try and improve our $k$-sparse recovery algorithm to work for arbitrary $k$-sparse support, rather than `well-separated' supports.  One natural way to do this is to address the fourth (open) challenge in Section~\ref{sec:overview} for a general class of OPs. 
\item Second, we could hope to handle a more general class of noise $\vhw$ than we currently do.  One could hope to handle \em any \em vector $\vv$, with an error guarantee that degrades smoothly with the $\ell_2$ norm of the `tail' of $\vv$.
\item Third, we would like to extend our results to hold for OPs defined over infinite intervals (e.g. Hermite and Laguerre polynomials).
\item Fourth, we would like to solve the sparse recovery for $\vF^T$ (where $\vF$ is as in Definition~\ref{def:OP}): i.e. given query access to $\vx$ figure out a good $k$-sparse approximation to $\vF^T\vx$ (recall that $\vF^{-1}=\vF^T$). (Note that this problem can be equivalently stated as follows: given query access to $\vF\vy$, compute a good $k$-sparse approximation to $\vy$.) Currently our results do not solve this problem since we cannot show that the existence of a filter $\vb$ such that $\vF\vD_{\vb}\vF^T$ is row-sparse. Note that this is not an issue for DFT since it is symmetric.
\item Finally, we would like to reduce the exponent on $k$ in our final runtime. In particular, for the case of random $k$-sparse support, the dependence on $k$ in the runtime for Jacobi transform is $k^8$. We note that we have not tried too hard to optimize the constants though we believe even getting a quadratic dependence on $k$ with our framework would be challenging. 
We would like to stress that the majority of the work in the sFFT literature has been to make the dependence on $k$ be linear and for such results, it seems very unlikely that a generic reduction from $k$-sparse recovery to $1$-sparse recovery would work. In other words, using the knowledge about the $1$-sparse recovery algorithm for DFT seems necessary to get a overall $k$-sparse FFT with running time $k\poly(\log{n})$.
\end{enumerate}

\section{Reduction to $1$-sparse case}
\label{sec:redux}

In this section, we will prove Theorem~\ref{thm:redux-main-intro}, which shows how we can reduce the $k$-sparse recovery problem (when $\vx$ is also `well-separated') to the $1$-sparse case.
Theorem~\ref{thm:redux-main-intro} follows from the following theorem, which is the main theorem in this section.
(Recall that $\smll(s,\vx)$ is the $s$'th largest value in $\vx$ by absolute value.)

\bthm
\label{thm:redux-main}
There are some constants $\delta_0, \cnst' > 0$ so that the following holds.
Let $p_1, \ldots, p_N$ be a $(C_0, C_1, \gamma_0)$-dense orthogonal polynomial family, and let $\vF_N$ be the $N \times N$ orthogonal matrix that arises from it as in Definition~\ref{def:OP}.
Let
\begin{equation}
\label{eq:U}
 U = \max_{i,j \in [N]^2} |\vF^{-1}_N[i,j]|. 
\end{equation}

Then there is an algorithm \textsc{Recover} that does the following.

Consider any $k \in [N]$, $\delta < \delta_0$, $0 < \mu < 1$, and $\gamma \ge \gamma_0/N$.
Let
\begin{equation}
  \label{eq:mu_0}
  \mu_0 = \frac{\mu^2 C_0^2 \gamma^2 }{k^2}.
\end{equation}
Suppose that
$\vF_N$ has a $\parens{N,\frac{6\delta}{\cnst},\delta,\frac{\mu_0}2}$ one-sparse recovery algorithm ${\cal A}$ for some $C \ge C'$ with query complexity $Q\parens{N,\frac{6\delta}{\cnst},\delta,\frac{\mu_0}2}$ and time complexity $T\parens{N,\frac{6\delta}{\cnst},\delta,\frac{\mu_0}2}$.  

Let $\vhv = \vhx + \vhw$ so that $\vhx = \vF_N \vx$ is $(k, C_1\gamma)$-sparsely separated, and so that
$\norm{2}{\vhw} \le \frac{\delta}{2\cnst}\smll(k, \vhx)$.

Then
with probability at least $1-\mu$, 
$\textsc{Recover}^{(\vv)}(k,\delta,\mu,\gamma)$ outputs $\vhz$ so that
\begin{equation}
  \label{eq:redux-main-guarantee-1}
  \|\vhx - \vhz\|_2 \leq 3 \delta \|\vx\|_2
\end{equation}
and so that
\begin{equation}
  \label{eq:redux-main-guarantee-2}
  \|\vhv - \vhz \|_2 \leq 3\delta \| \vx \|_2 + \|\vw\|_2.
\end{equation}
Further,
$\textsc{Recover}$ makes at most
\[ \poly\left(\frac{k\log(1/\mu)}{\gamma \delta C_0}\right) \cdot \parens{NU^2 + Q\parens{N, \frac{6\delta}{C}, \delta, \frac{\mu_0}{2}}} \]
queries to $\vv$,
and has running time
\[ \poly\left(\frac{k\log(1/\mu)}{\gamma \delta C_0}\right) \cdot \parens{NU^2 + T\parens{N, \frac{6\delta}{C}, \delta, \frac{\mu_0}{2}}} \]
\ethm

\begin{algorithm}
\caption{$\textsc{Recover}^{(\vv)}(k,\delta,\mu,\gamma)$}\label{alg:recover}
\begin{algorithmic}
\Require Query access to $\vv = \vx + \vw$ where $\vhx=\vF\vx$ is $(k,C_1\gamma)$-sparsely separated and $\vw$ is as in Theorem~\ref{thm:redux-main}, as well as parameters $k, \delta,\mu$.
\Ensure $\vhz$ (an approximation to $\vhx$) 
\Statex
\State $\vhz\gets\vzero$
\For{$i=1,\dots,k$}
    \State $\vtz, \texttt{stop} \gets\textsc{Peeler}^{(\vv)}\parens{\vhz,k,\delta/\cnst, \mu,\gamma}$ 
    \If{\texttt{stop}}
        \State Break
    \EndIf
    \State $\vhz \gets \vtz$
\EndFor
\State \Return{$\vhz$}
\end{algorithmic}
\end{algorithm}

The basic idea of \textsc{Recover} (Algorithm~\ref{alg:recover}) is as follows.  We will define an algorithm \textsc{Peeler} (Algorithm~\ref{alg:peel}) which will
iteratively `peel' off the heavy hitters and store them in the approximation $\vhz$ until the variable \texttt{stop} is set to $\true$.  We will show that if we stop, then every value in the residual $\vhv - \vhz$ will be very small, at which point we will be done.  The idea for the \textsc{Peeler} algorithm is illustrated in Figure~\ref{fig:peel}: we first use a filter to hopefully isolate a single spike of $\vhx$, and then we use the one-sparse recovery algorithm to estimate this spike and subtract it off.

In Section~\ref{sec:proof-redux-main} below, we prove Theorem~\ref{thm:redux-main} assuming that a suitable \textsc{Peeler} algorithm exists.  More precisely, we will formalize what we need from \textsc{Peeler} in Lemma~\ref{lem:redux-work-horse}, and we will prove Theorem~\ref{thm:redux-main} assuming Lemma~\ref{lem:redux-work-horse}.  Next, in Sections~\ref{sec:peeler-prelims} and \ref{sec:peelerpf}, we will state the \textsc{Peeler} algorithm (Algorithm~\ref{alg:peel}) and prove that it works.  Section~\ref{sec:peeler-prelims} contains some useful preliminaries, and Section~\ref{sec:peelerpf} contains the statement of \textsc{Peeler} and the proof of Lemma~\ref{lem:redux-work-horse}.

\subsection{Proof of Theorem~\ref{thm:redux-main}}\label{sec:proof-redux-main}
In this section we prove Theorem~\ref{thm:redux-main} which implies Theorem~\ref{thm:redux-main-intro}.  
For the rest of this section, given a vector $\vx\in\R^N$, we will denote its transform $\vF_N\cdot \vx$ by $\vhx$.

The algorithm \textsc{Recover} (Algorithm~\ref{alg:recover}) uses an algorithm called \textsc{Peeler} (Algorithm~\ref{alg:peel}) multiple times.  
This algorithm will satisfy the following guarantee.
\blmm
\label{lem:redux-work-horse}
Consider an orthogonal polynomial family and its corresponding matrix $\vF_N$ that satisfies the properties in Theorem~\ref{thm:redux-main}.

There is an algorithm \textsc{Peeler} with the following guarantee.

Let $k \in [N]$, $\eps$ be sufficiently small, $0 < \mu < 1$, $\gamma \ge \gamma_0/N$, and $\mu_0$ as in~\eqref{eq:mu_0}.
Suppose that $\vF_N$ has a $(N, 6\eps, \cnst\eps, \mu_0/2)$ one-sparse recovery algorithm ${\cal A}$ for some sufficiently large $C$.

Suppose \textsc{Peeler} has query access to $\vv=\vx+\vw$ such that $\vhx$ is $\parens{k,C_1\gamma}$-sparsely separated.
Consider an input $\vhz$ with sparsity $\|\vhz\|_0 < k$ such that $\|\vhw\|_2 \le \epsilon\|\vhv - \vhz \|_\infty$ and the following holds for every $i\in \supp(\vhz)$:
\begin{equation}
  \label{eq:peeler-condition}
  \abs{\vhz[i]-\vhv[i]}\le \cnst\eps\abs{\vhv[i]} \text{ and } \abs{\vhv[i]} \ge (1-2\cnst\eps)\cdot\| \vhv - \vhz \|_\infty.
\end{equation}
Then, with probability at least $1-\mu$, \textsc{Peeler}$^{(\vv)}(\vhz, k, \epsilon, \mu, \gamma)$ returns $(\vtz, \text{\texttt{stop}})$, 
where $\vtz \in \mathbb{R}^N$ and \texttt{stop} is a Boolean variable, so that
\begin{itemize}
\item If \texttt{stop} is $\false$, then $\vtz=\vhz+\tilde{v}\cdot \ve_h$ for some $h\in\supp(\vhx)\setminus \supp(\vhz)$,
and~\eqref{eq:peeler-condition} is satisfied for all $i \in \supp(\vtz)$, i.e.
\[\abs{\tilde{v}-\vhv[h]}\le \cnst\eps\abs{\vhv[h]}\text{ and } \abs{\vhv[h]}\ge (1-2\cnst\eps)\| \vhv - \vtz \|_\infty.\]
\item If \texttt{stop} is $\true$, then $\supp(\vtz) = \supp(\vhz)$, \eqref{eq:peeler-condition} is satisfied for all $i \in \supp(\vtz)$, and
$\|\vhv-\vtz\|_{\infty}\le \frac{\eps}{\sqrt{k}}\|\vhv\|_{\infty}$.
\end{itemize}
Further, if ${\cal A}$ has query complexity $Q(N, 6\eps, \cnst\eps, \mu_0/2)$ and running time $T(N, 6\eps, \cnst\eps, \mu_0/2)$,
$\textsc{Peeler}$ makes at most
\[ O\left( T_0 \cdot T_2 \cdot \frac{ \sqrt{k}}{\eps \gamma} \left( T_1 + Q(N,6\eps, \cnst\eps, \mu_0/2)\right) \right) \]
queries to $\vx + \vw$, and runs in time at most
\[ O\left( T_0 \cdot T_2 \cdot \frac{ k^{3/2}}{\eps \gamma} \left( T_1 + T(N,6\eps, \cnst\eps, \mu_0/2)\right) \right) \]
where $\mu_0, T_0, T_1, T_2$ are as defined in Algorithm~\ref{alg:peel}.
\elmm

The basic idea of Lemma~\ref{lem:redux-work-horse} is that 
\textsc{Peeler} will iteratively `peel' off the heavy hitters and store them in the approximation $\vhz$ until the variable \texttt{stop} is set to $\true$.  As long as we are doing well enough so far, the lemma says that we will find a new heavy hitter.  If we stop, then the lemma implies that every value in the residual $\vhv - \vhz$ is very small.

We will present \textsc{Peeler} and prove Lemma~\ref{lem:redux-work-horse} in Sections~\ref{sec:peeler-prelims} and \ref{sec:peelerpf}.  Before that, we prove Theorem~\ref{thm:redux-main}---which says that \textsc{Recover} (Algorithm~\ref{alg:recover}) works---assuming Lemma~\ref{lem:redux-work-horse}.

\begin{proof}[Proof of Theorem~\ref{thm:redux-main}]

We will repeatedly apply Lemma~\ref{lem:redux-work-horse} with $\eps=\frac{\delta}{C}$ (and preserving the remaining parameters); for the rest of this section, we use this choice for $\eps$.
As a preliminary, we note the following useful bounds between $\vhx$ and $\vhv$ at the locations of the heavy hitters (or {\em spikes}), due to the condition on the noise $\vhw$ in the statement of the theorem.  More precisely, we have the following claim.
\begin{claim}\label{claim:convert-x-v}
  For any $i \in \supp(\vhx)$,
  \begin{equation}
    \label{eq:v-to-x}
    (1-\epsilon)\abs{\vhx[i]} \le \abs{\vhv[i]} \le (1+\epsilon)\abs{\vhx[i]}
  \end{equation}
    \begin{equation}
    \label{eq:x-to-v}
    (1-\epsilon)\abs{\vhv[i]} \le \abs{\vhx[i]} \le (1+\epsilon)\abs{\vhv[i]}.
  \end{equation}
  Furthermore if $v$ satisfies $\abs{v-\vhv[i]} \le C\epsilon\abs{\vhv[i]}$, then
  \begin{equation}
    \label{eq:x-decrease}
    |v-\vhx[i]| 
    \le 
    2\epsilon(C+1)|\abs{\vhx[i]}.
  \end{equation}

\end{claim}
\begin{proof}
  For any $i\in\supp(\vhx)$, the noise condition says 
\begin{align*}
\abs{\vhx[i]-\vhv[i]}&=\abs{\vhw[i]}\\
&\leq \| \vhw\|_2 \\
&\leq (\eps/2)\smll(k, \vhx) \\
&\le(\epsilon/2)\abs{\vhx[i]}, 
\end{align*}
which implies~\eqref{eq:v-to-x}.
  This also implies that
  \[
    (1-\epsilon/2)\abs{\vhv[i]} \le \frac{1}{1+\epsilon/2}\abs{\vhv[i]} \le \abs{\vhx[i]} \le \frac{1}{1-\epsilon/2}\abs{\vhv[i]} \le (1+\epsilon)\abs{\vhv[i]},
  \]
which in turn implies~\eqref{eq:x-to-v}.

  Consequently, if $v$ satisfies $\abs{v-\vhv[i]} \le C\epsilon\abs{\vhv[i]}$, then
  \begin{align*}
    \abs{\vhx[i]-v} &\le \abs{\vhw[i]}+\abs{\vhv[i]-v} \le \frac\epsilon 2\cdot\abs{\vhx[i]} + \abs{\vhv[i]-v}
    \\&\le \frac \epsilon 2\cdot(1+\epsilon)\abs{\vhv[i]} + C\epsilon\abs{\vhv[i]}
    \le (1+\epsilon)\parens{\frac\epsilon 2\cdot(1+\epsilon)+C\epsilon}\abs{\vhx[i]}
    \\&\le 2\epsilon(C+1)\abs{\vhx[i]}
  \end{align*}
\end{proof}

\paragraph{Conditions to Lemma~\ref{lem:redux-work-horse}.}

We first check that at every call to \textsc{Peeler} in \textsc{Recover}, the conditions to Lemma~\ref{lem:redux-work-horse} are always met with this choice of $\epsilon$.
First we note that the invariant~\eqref{eq:peeler-condition} is vacuously satisfied at the beginning of \textsc{Recover}, and the guarantees of \textsc{Peeler} imply that it is therefore always satisfied.
The only condition left to check is that the requirement on the noise $\vhw$ in Lemma~\ref{lem:redux-work-horse} is met:
that is, that $\|\vhw\|_2 \le \epsilon \|\vhv - \vhz\|_\infty$.
Let $h \in \supp(\vhx)\setminus\supp(\vhz)$, which exists since \textsc{Peeler} is called at most $k$ iterations.
By our assumption on the noise in Theorem~\ref{thm:redux-main} and~\eqref{eq:x-to-v},
\[ \|\vhw\|_2 \le \frac{\epsilon}{2}\cdot \smll(k, \vhx) \le \frac{\epsilon}{2}|\vhx[h]| \le \frac{\epsilon}{2}(1+\epsilon)|\vhv[h]| = \frac{\epsilon}{2}(1+\epsilon)|\vhv[h]-\vhz[h]| \leq \eps \| \vhv - \vhz \|_\infty. \]

\paragraph{Establishing the error guarantee.}

Next, we claim that the estimates $\vhz$ returned by \textsc{Peeler} are always good in an $\ell_2$ sense.

\begin{claim}\label{claim:boundell2}
Let $\vhz$ be the output of \textsc{Peeler} at any iteration.
Let $Z$ be the support of $\vhz$ and $X$ be the support of the largest $|Z|$ coordinates of $\vhx$.
Then
\[
\| \vhx - \vhz\|_2^2 
\leq 4\epsilon^2(C+1)^2 \|\vhx_X\|_2^2 + (1 + 5C\eps) \|\vhx_{X^c}\|_2^2,
\]
\end{claim}
\begin{proof}
We have
\begin{align}
\| \vhx - \vhz \|_2^2 &= \left(\sum_{i \in Z \setminus X} |\vhx[i] - \vhz[i]|^2 + \sum_{i \in X \setminus Z} |\vhx[i] - \vhz[i]|^2  \right) + \left( \sum_{i \in Z \cap X} |\vhx[i] - \vhz[i]|^2 + \sum_{i \in (Z\cup X)^c} |\vhx[i] - \vhz[i]|^2 \right)  \notag \\
&\leq  \left(\sum_{i \in Z \setminus X} 4\epsilon^2(C+1)^2 |\vhx[i]|^2 + \sum_{i \in X \setminus Z} |\vhx[i]|^2  \right) + \left( \sum_{i \in Z \cap X} 4\epsilon^2(C+1)^2 |\vhx[i]|^2 + \sum_{i \in (Z\cup X)^c} |\vhx[i]|^2 \right),
\label{eq:boundell2}
\end{align}
where we have used the guarantee from \textsc{Peeler} \eqref{eq:peeler-condition} as well as \eqref{eq:x-decrease}.
Now suppose that $i \in Z \setminus X$ and $j \in X \setminus Z$.
By definition of $X$, $\abs{\vhx[i]} \le \abs{\vhx[j]}$.
On the other hand, 
we have
\[ |\vhx[i]| \ge (1-\epsilon)|\vhv[i]| \geq (1-\epsilon)(1 - 2C\eps)|\vhv[j]| \ge (1-\epsilon)^2(1-2C\epsilon)|\vhx[j]|, \]
where the first and the last inequality follows from Claim~\ref{claim:convert-x-v} and the second inequality follows from~\eqref{eq:peeler-condition} and
since $\norm{\infty}{\vhv-\vhz} \ge \abs{\vhv[j]-\vhz[j]} = \abs{\vhv[j]}$.
Since $|Z \setminus X| = |X\setminus Z|$,
 the first term in~\eqref{eq:boundell2} can be bounded by
\begin{align*}
  \sum_{i \in Z \setminus X} 4\epsilon^2(C+1)^2|\vhx[i]|^2 + \sum_{j \in X \setminus Z} |\vhx[j]|^2 
   &\leq \sum_{j \in X \setminus Z} 4\epsilon^2(C+1)^2|\vhx[j]|^2 + \sum_{i \in Z \setminus X} \frac{|\vhx[i]|^2}{(1-\epsilon)^4(1-2C\epsilon)^2}
  \\&\leq 4\epsilon^2(C+1)^2 \|\vhx_{X \setminus Z} \|_2^2 + (1 + 5C\eps) \|\vhx_{Z \setminus X} \|_2^2
\end{align*}
using the fact that $C$ is a sufficiently large absolute constant, and that $\eps$ is sufficiently small compared to $C$.
Adding in the second term in \eqref{eq:boundell2}, we have
\begin{align*}
\| \vhx - \vhz\|_2^2 &\leq \left( 4\epsilon^2(C+1)^2 \|\vhx_{X \setminus Z} \|_2^2 + (1 + 5C\eps) \|\vhx_{Z \setminus X} \|_2^2 \right) + \left( 4\epsilon^2(C+1)^2 \|\vhx_{Z \cap X}\|_2^2 + \| \vhx_{(Z \cup X)^c} \|_2^2 \right) \\
&\leq 4\epsilon^2(C+1)^2 \|\vhx_X\|_2^2 + (1 + 5C\eps) \|\vhx_{X^c}\|_2^2,
\end{align*}
as desired.

\end{proof}

Suppose that \textsc{Peeler} never returns $\texttt{stop}=\true$ in \textsc{Recover}.
In this case, the final output $\vhz$ has support of size $k$ so $X=\supp(\vhx)$ in Claim~\ref{claim:boundell2} and we get
\[
  \norm{2}{\vhx-\vhz}^2 \le 4\epsilon^2(C+1)^2 \norm{2}{\vhx}^2,
\]
which immediately establishes the first error bound in the theorem.

On the other hand, suppose some call to \textsc{Peeler} returns $\texttt{stop}=\true$.
Let $\vhz$ denote the output of \textsc{Peeler} on this final step, $s = \|\vhz\|_0$, and
let $X$ denote the support of the $s$ largest entries of $\vhx$.
Let $h$ be the index of the $(s+1)$'th largest entry in $\vhx$, and let $h'$ be the largest entry in $\vhx_{\supp(\vtz)^c}$. Note that $h' \in \supp(\vhx)$ since the algorithm stopped before hitting all $k$ heavy hitters.
Now
\begin{align*}
  \|\vhx_{X^c}\|_2 &\le \sqrt{k}\|\vhx_{X^c}\|_\infty = \sqrt{k} |\vhx[h]| \le \sqrt{k} |\vhx[h']| \\
                    &\le \sqrt{k} (1+\epsilon)|\vhv[h']| \\
                    &= \sqrt{k} (1+\epsilon)|\vhv[h']-\vtz[h']| \\
                    &\leq \sqrt{k} (1+\epsilon) \| \vhv - \vtz\|_\infty \\
                    &\leq \epsilon(1+\epsilon) \| \vhv \|_\infty \\
                    &\leq \epsilon(1+\epsilon)(1+\epsilon) \| \vhx \|_\infty \leq \epsilon(1+\epsilon)(1+\epsilon) \| \vhx \|_2
\end{align*}
In the above, the second line we have used Claim~\ref{claim:convert-x-v}.
In the second to last line we have used the guarantee from Lemma~\ref{lem:redux-work-horse}, and in the last line we use the fact that $\arg\max_i |\vhv[i]| \in \supp(\vhx)$ which follows from the assumption on $\vhw$ and Claim~\ref{claim:convert-x-v}.

Thus by Claim~\ref{claim:boundell2} we have
\begin{align*}
  \| \vhx - \vhz\|_2^2 
  &\leq 4\epsilon^2(C+1)^2 \|\vhx_X\|_2^2 + (1 + 5C\eps) \|\vhx_{X^c}\|_2^2,
  \\&\leq 4\epsilon^2(C+1)^2 \|\vhx\|_2^2 + (1 + 5C\eps)\epsilon^2(1+\epsilon)^2(1+\epsilon)^2 \|\vhx\|_2^2,
  \\&\leq 9C^2\epsilon^2\|\vhx\|_2^2
\end{align*}
provided that $C$ is sufficiently large.
This establishes the first error bound~\eqref{eq:redux-main-guarantee-1} in the theorem statement.
Now the second error bound~\eqref{eq:redux-main-guarantee-2} follows from the triangle inequality.

\paragraph{Query and Time complexity of the Recover algorithm.}

Finally, we compute the query complexity and running time.  
With our choices of $T_0,T_1$ and $T_2$, it can be checked that
 \textsc{Peeler} uses query complexity 
\[ \poly\left(\frac{k \log(1/\mu) }{\gamma \delta C_0}\right) \cdot (NU^2 + Q(N, 6\delta/C, \delta, \mu C_0 \gamma^2/k^2)) \]
and has running time
\[ \poly\left(\frac{k \log(1/\mu)}{\gamma \delta C_0}\right) \cdot (NU^2 + T(N, 6\delta/C, \delta, \mu C_0 \gamma^2/k^2)). \]
Since
the work of \textsc{Recover} is dominated by running \textsc{Peeler} at most $k$ times,
the same query complexity and running time bounds hold for \textsc{Recover} (since the extra factor of $k$ can be absorbed into the $\poly(k)$ factor).
\end{proof}

\subsection{Preliminaries to the proof of Lemma~\ref{lem:redux-work-horse}}\label{sec:peeler-prelims}
In this section we present a few results which will be useful in the proof of Lemma~\ref{lem:redux-work-horse}.
We begin by showing that $(C_0,C_1,\gamma_0)$-dense orthogonal polynomial families (see Definition~\ref{def:roots-spread}) and $(k, C_1 \gamma)$-sparsely separated vectors (see Definition~\ref{def:sparse-separated}) work well together in the following sense.

\blmm
\label{lem:isolation}
Let $H$ be the support of any $(k,C_1\gamma)$-sparsely separated vector $\vhx$ where $\gamma \ge 2\gamma_0/N$.
Let $h\in H$. Pick $0\le \ell<N$ uniformly at random. Then
\begin{itemize}
\item[(i)]\label{item:isolation-1} With probability at least $C_0\gamma/2$, we have
\begin{equation}
\label{eq:heaviest}
\theta_h\in \brackets{\theta_{\ell}-\frac{\gamma}{4},\theta_{\ell}+\frac{\gamma}{4}},
\end{equation}
and
\item[(ii)]\label{item:isolation-2} Conditioned on~\eqref{eq:heaviest}, the following holds with probability $1$:
\begin{equation}
\label{eq:other-hh}
\set{\theta_i|i\in H\setminus \{h\}} \cap \brackets{\theta_{\ell}-\frac{\gamma}{2},\theta_{\ell}+\frac{\gamma}{2}}=\emptyset.
\end{equation}
\end{itemize}
\elmm
\begin{proof}
We start with part (i). 
Indeed note that for~\eqref{eq:heaviest} to occur, we must have
\[\theta_{\ell}\in \brackets{\theta_h-\frac{\gamma}{4},\theta_h+\frac{\gamma}{4}}.\]
Since the orthogonal polynomial family is $(C_0,C_1,\gamma_0)$-dense, the above happens for at least $\frac{C_0\gamma N}2$ values of $\ell$.
Since $\ell$ was chosen randomly, the above holds with probability $\frac{C_0\gamma}{2}$. Once~\eqref{eq:heaviest} holds, we show that~\eqref{eq:other-hh} follows from $\vhx$ being $(k,C_1\gamma)$-separated. Since the orthogonal polynomial family is $(C_0,C_1,\gamma_0)$-dense, the number of $\ell$ such that
\[\theta_{\ell}\in \brackets{\theta_h-\frac{\gamma}{2},\theta_h+\frac{\gamma}{2}}\]
is at most $C_1\gamma N$.
In particular, any such $\ell$ satisfies $|\ell-h| \le C_1\gamma N$.
Since $\vhx$ is $(k,C_1\gamma)$-separated none of these can be contained in $H\setminus \{h\}$, as desired.
\end{proof}

The \textsc{Peeler} algorithm relies on 
 \em boxcar polynomials, \em which we define as follows (see Figure~\ref{fig:boxcar}).
\bdefn \label{def:boxcar}
Let $d\ge 1$ be an integer and $0 < \gamma \le \frac \pi 2, \theta \in [0,\pi], \eps\ge 0$ be reals.
We say that $p(X)$ is a $(d,\eps,\theta,\gamma)$-{\em box car polynomial} if the following are true:
\begin{enumerate}
\item[(a)]\label{item:BC-a} $p(X)$ has degree $d$,
\item[(b)]\label{item:BC-b} $|p(\cos \phi)|\le \eps$ for any $\phi\in [0,\pi]\setminus[\theta-2\gamma,\theta+2\gamma]$,
\item[(c)]\label{item:BC-c} $|p(\cos \phi)-1|\le \eps$ for any $\phi\in \brackets{\theta-\gamma,\theta+\gamma}$, and
\item[(d)]\label{item:BC-d} $|p(\cos \phi)|\le 1+\eps$ for the remaining values of $\phi$.
\end{enumerate}
\edefn
\begin{figure}
\centering
\begin{tikzpicture}[yscale=3, xscale=7]

\draw (-.1,0) -- (1.1,0);
\node(a) at (0, -.3) {$0$};
\draw (a) to (0, 1.2);
\node(b) at (1, -.3) {$\pi$};
\draw (b) to (1,0);
\node (x) at (.6, -.3) {$\theta$};
\node (x1) at (.4, -.3) {$\theta - 2\gamma$};
\node (x11) at (.5, -.6) {$\theta - \gamma$};
\node (x22) at (.7, -.6) {$\theta + \gamma$};
\node (x2) at (.8, -.3) {$\theta + 2\gamma$};
\draw (x) to (.6, 0);
\draw (x1) to (.4, 0);
\draw (x11) to (.5, 0);
\draw (x2) to (.8, 0);
\draw (x22) to (.7, 0);
\draw[dashed] (x1) to (.4, 1.1);
\draw[dashed] (x11) to (.5, 1.1);
\draw[dashed] (x2) to (.8, 1.1);
\draw[dashed] (x22) to (.7, 1.1);
\node (o) at (-.1,1) {$1$};
\draw (o) to (0,1);
\draw[dashed] (o) to (1.1, 1);
\node (e) at (-.1,.9) {$1 - \eps$};
\draw[red, dashed] (e) to (1.1, .9);
\node (e2) at (-.1,1.1) {$1 + \eps$};
\draw[red, dashed] (e2) to (1.1, 1.1);
\node (ee) at (-.1, .1) {$\eps$};
\draw[red, dashed] (ee) to (1.1, .1);
\node (ee2) at (-.1, -.1) {$-\eps$};
\draw[red, dashed] (ee2) to (1.1, -.1);
\draw [cyan,ultra thick] plot [smooth, tension=.6] coordinates { (0,0) (.2,.05) (.4, 0) (.5, 1.05)  (.6,.92) (.7,.96) (.8,0) (1,0) };
\end{tikzpicture}
\caption{A (cartoon of a) $(d, \eps, \theta, \gamma)$-box car polynomial.}\label{fig:boxcar}
\end{figure}
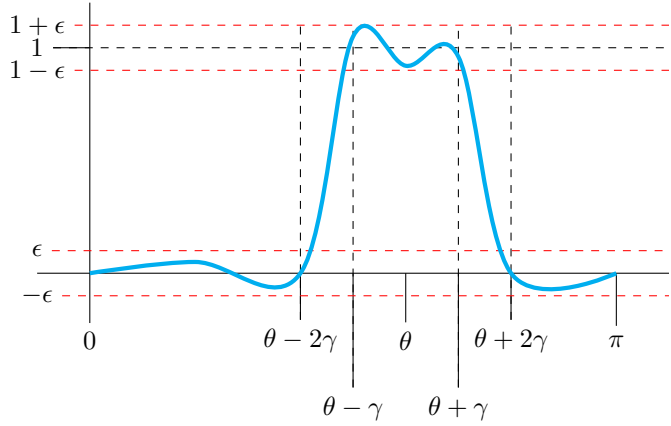
Fortunately for us, such polynomials exist with reasonably low degree:
\blmm
\label{lem:boxcar-lowdeg}
For any  $\theta \in [0, \pi]$, there exists an $\parens{O\parens{\frac{1}{\gamma\eps}},\eps,\theta,\gamma}$-boxcar polynomial.
\elmm
To prove Lemma~\ref{lem:boxcar-lowdeg},
we will need the following result from approximation theory~\cite{R81,jackson}:
\bthm[Jackson's Theorem]
\label{thm:jackson-trig}
There is a constant $C$ such that:
if $f : [0, 2\pi] \to \C$ is a periodic function with Lipschitz constant $L$, then for every $d$ there is a trigonometric polynomial $T$ of degree at most $d$ such that
\begin{equation}
  \label{eq:jackson}
  |f(x) - T(x)| \le C\frac{L}{d}.
\end{equation}
Furthermore if $f$ is even then $T$ is even.
\ethm
Recall that a trigonometric polynomial of degree $d$ has the form $a_0 + \sum_{m=1}^d a_m \cos(mx) + \sum_{m=1}^d b_m \sin(mx)$,
and it is even if all $b_m=0$.
In particular, since $\cos(mx) = T_m(\cos x)$ where $T_m$ is the $m$th Chebyshev polynomial, 
an even trigonometric polynomial $T$ of degree $d$ in $x$ is a polynomial of degree $d$ in $\cos x$,
in other words $T \circ \arccos$ is a polynomial of degree $d$.
Now we can prove Lemma~\ref{lem:boxcar-lowdeg}.

\begin{proof}[Proof of Lemma~\ref{lem:boxcar-lowdeg}]
Consider the following piecewise linear function $f$ such that
\[f(x)=\begin{cases}
\frac{x-\delta+2\gamma}{\gamma} & \text{ if }\delta-2\gamma \le x \le \delta-\gamma\\
1 & \text{ if } \delta-\gamma \le x\le \delta+\gamma\\
1-\frac{x-\delta-\gamma}{\gamma} & \text{ if }\delta+\gamma \le x \le \delta+2\gamma\\
0& \text{ otherwise}.
\end{cases}
\]
and such that $f$ on $[\pi,2\pi]$ is reflected from $[0,\pi]$ so that $f$ is even.
Clearly $f$ has Lipschitz constant $\frac{1}{\gamma}$.
Let $d$ satisfy $\frac{C}{d}\frac{1}{\gamma} \le \epsilon$ and $T$ be a trigonometric polynomial of degree $d$ satisfying Theorem~\ref{thm:jackson-trig}.
Then~\eqref{eq:jackson} directly implies conditions (b), (c) and (d) for the function $p = T \circ \arccos$, which is a polynomial of degree $d$ as noted before.
\end{proof}

\subsection{Proof of Lemma~\ref{lem:redux-work-horse}}\label{sec:peelerpf}
In this section we prove Lemma~\ref{lem:redux-work-horse}.
Fix a large enough $N$.  For the rest of this section, for notational convenience we will use $\vF=\vF_N$.

We begin by describing the algorithm \textsc{Peeler}, which requires setting up some notation.

For every $0\le \ell <N$, let $b_{\ell}(X)$ be an
$\parens{d,\frac{\eps}{\sqrt{k}},\theta_{\ell},\frac{\gamma}{4}}$-boxcar polynomial, for some $d$ which we will choose to be
\[d=O\parens{\frac{\sqrt{k}}{\eps\gamma}}.\]
Define $\vb_\ell=\parens{b_\ell(\evalpts_1),\dots,b_\ell(\evalpts_{N-1})}$.

Given any $\tau\ge 0$, we define the following `truncation' function that is defined for any $x\in\R$:
\[\trunc{\tau}{x}=\begin{cases}
\sign{x}\cdot \tau&\text{ if } \abs{x}>\tau\\
x &\text{ otherwise}.
\end{cases}
\]

For notational convenience, we will let  $\vr = \vv - \vz$ and $\vhr = \vhv - \vhz$ be the residual vector.

Given this notation,
the algorithm \textsc{Peeler} to compute $\vtz$ is presented in Algorithm~\ref{alg:peel}, and is illustrated in Figure~\ref{fig:peel}. 
Next, we go through the various components of Algorithm~\ref{alg:peel} and prove that they work as intended.

\begin{figure}
\centering
\begin{tikzpicture}[xscale=4,yscale=2]
\begin{scope}

\draw (-.1,0) -- (1.1,0);
\node(a) at (0, -.3) {$0$};
\draw (a) to (0, 1.2);
\node(b) at (1, -.3) {$1$};
\draw (b) to (1,0);
\foreach \i in {.3,.5,.8}{
	\node[draw,ultra thick, red, fill=red, circle, scale=0.3](c) at (\i,1 + .2*\i) {};
	\draw[red,ultra thick] (c) to (\i, 0);
}
\foreach \i in {.1, .2, .4, .6, .7,  .9}
{
	\pgfmathparse{round(mod(\i*10,5)) /15  }
       \let\j\pgfmathresult
	\node[draw, thick, red, circle, scale=0.3] (c) at (\i, \j) {};
	\draw[red,thick] (c) to (\i, 0);
}
\node at (-.3,.5) {$\vhx -\vhz$};
\end{scope}
\begin{scope}[yshift=-2cm]

\draw (-.1,0) -- (1.1,0);
\node(a) at (0, -.3) {$0$};
\draw (a) to (0, 1.2);
\node(b) at (1, -.3) {$1$};
\draw (b) to (1,0);
\node (x) at (.6, -.3) {$\alpha_\ell$};
\node (x1) at (.4, 0) {};
\node (x11) at (.5, 0) {};
\node (x22) at (.7, 0) {};
\node (x2) at (.8, 0) {};
\draw (x) to (.6, 0);
\draw (x1) to (.4, 0);
\draw (x11) to (.5, 0);
\draw (x2) to (.8, 0);
\draw (x22) to (.7, 0);
\draw[dashed] (x1) to (.4, 1.1);
\draw[dashed] (x11) to (.5, 1.1);
\draw[dashed] (x2) to (.8, 1.1);
\draw[dashed] (x22) to (.7, 1.1);
\node (o) at (-.1,1) {$1$};
\draw (o) to (0,1);
\draw[dashed] (o) to (1.1, 1);
\node (e) at (0,.9) {};
\draw[red, dashed] (e) to (1.1, .9);
\node (e2) at (0,1.1) {};
\draw[red, dashed] (e2) to (1.1, 1.1);
\node (ee) at (0, .1) {};
\draw[red, dashed] (ee) to (1.1, .1);
\node (ee2) at (0, -.1) {};
\draw[red, dashed] (ee2) to (1.1, -.1);
\draw [cyan,ultra thick] plot [smooth, tension=.6] coordinates { (0,0) (.2,.05) (.4, 0) (.5, 1.05)  (.6,.92) (.7,.96) (.8,0) (1,0) };
\node at (-.3,.5) {$\vb_{\ell}$};
\end{scope}
\begin{scope}[yshift=-4cm]

\draw (-.1,0) -- (1.1,0);
\node(a) at (0, -.3) {$0$};
\draw (a) to (0, 1.2);
\node(b) at (1, -.3) {$1$};
\draw (b) to (1,0);
\foreach \i in {.5}{
	\node[draw,ultra thick, violet, fill=violet, circle, scale=0.3](c) at (\i,1) {};
	\draw[violet,ultra thick] (c) to (\i, 0);
}
\node(d) at (-.1, 1) {$y_\ell$};
\draw[dashed] (d) to (c);
\foreach \i in {.4, .6}
{
	\pgfmathparse{round(mod(\i*10,4)) /15  }
       \let\j\pgfmathresult
	\node[draw, thick, violet, circle, scale=0.3] (c) at (\i, \j) {};
	\draw[violet,thick] (c) to (\i, 0);
}
\foreach \i in {.1, .2, .3, .7, .8, .9}
{
	\pgfmathparse{round(mod(\i*10,5)) /40  }
       \let\j\pgfmathresult
	\node[draw, thick, violet, circle, scale=0.3] (c) at (\i, \j) {};
	\draw[violet,thick] (c) to (\i, 0);
}

\node(d) at (0.5, -.3) {$\alpha_h$};
\draw(d) to (0.5, 0);
\node at (-.3,.5) {$\vhy_\ell = \mathbf{D}_{\vb_{\ell}} (\vhx-\vhz)$};
\end{scope}
\draw[->] (1.2,.5) to node[above]{$\vF^{-1}$} (1.8,.5);
\begin{scope}[xshift=2cm]

\draw (-.1,0) -- (1.1,0);
\node(a) at (0, -.3) {$0$};
\draw (a) to (0, 1.2);
\node(b) at (1, -.3) {$1$};
\draw (b) to (1,0);
\foreach \i in {.1, .2, .3, .4, .5, .6, .7, .8,  .9}
{
	\pgfmathparse{.3 + round(mod(\i*100,7)) /20  }
       \let\j\pgfmathresult
	\node[draw, thick, red, circle, scale=0.3] (c) at (\i, \j) {};
	\draw[red,thick] (c) to (\i, 0);
}
\node at (1,1) {$\vF^{-1}(\vhx -\vhz) = \vx - \vz$};
\end{scope}
\begin{scope}[xshift=2cm, yshift=-4cm]

\draw (-.1,0) -- (1.1,0);
\node(a) at (0, -.3) {$0$};
\draw (a) to (0, 1.2);
\node(b) at (1, -.3) {$1$};
\draw (b) to (1,0);
\foreach \i in {.1, .2, .3, .4, .5, .6, .7, .8,  .9}
{
	\pgfmathparse{.2 + round(mod(\i*67,5)) /19  }
       \let\j\pgfmathresult
	\node[draw, thick, violet, circle, scale=0.3] (c) at (\i, \j) {};
	\draw[violet,thick] (c) to (\i, 0);
}
\node at (1.3,.5) {$\vF^{-1} \vhy_\ell$};
\end{scope}
\draw[->] (1.2,-3.5) to node[above]{$\vF^{-1}$} (1.8,-3.5);
\node[draw,circle,blue](e) at (2.5,-3.6) {};
\draw[blue,->] (e) to[out=90,in=-90] (2.6, -.1);
\draw[blue,->] (e) to[out=90,in=-90] (2.1, -.1);
\draw[blue,->] (e) to[out=90,in=-90] (2.9, -.1);
\draw[blue] (2.3, -3) to[out=45,in=135] (2.7, -3);
\node[blue] at (2.73,-3.02) {$d$};
\end{tikzpicture}
\caption{The intuition behind the \textsc{Peeler} algorithm (Alg.~\ref{alg:peel}).
The boxcar polynomial $\vb_\ell$ (hopefully) isolates a single spike in $\vhx - \vhz$, resulting in $\vhy_\ell \approx y_\ell \ve_h$.
Each entry of $\vF^{-1} \vhy$ can be queried by querying $O(d)$ entries of $\vx - \vz$, using the \textsc{SimulateQueryAccess} algorithm (Alg.~\ref{alg:simulate}).  Using this access, we can use the 1-sparse recovery algorithm on $\vF^{-1} \vhy$ to recover $y_\ell$ and $\mathbf{e_h}$.  Then we add this spike to $\vhz$ and continue.}\label{fig:peel}
\end{figure}
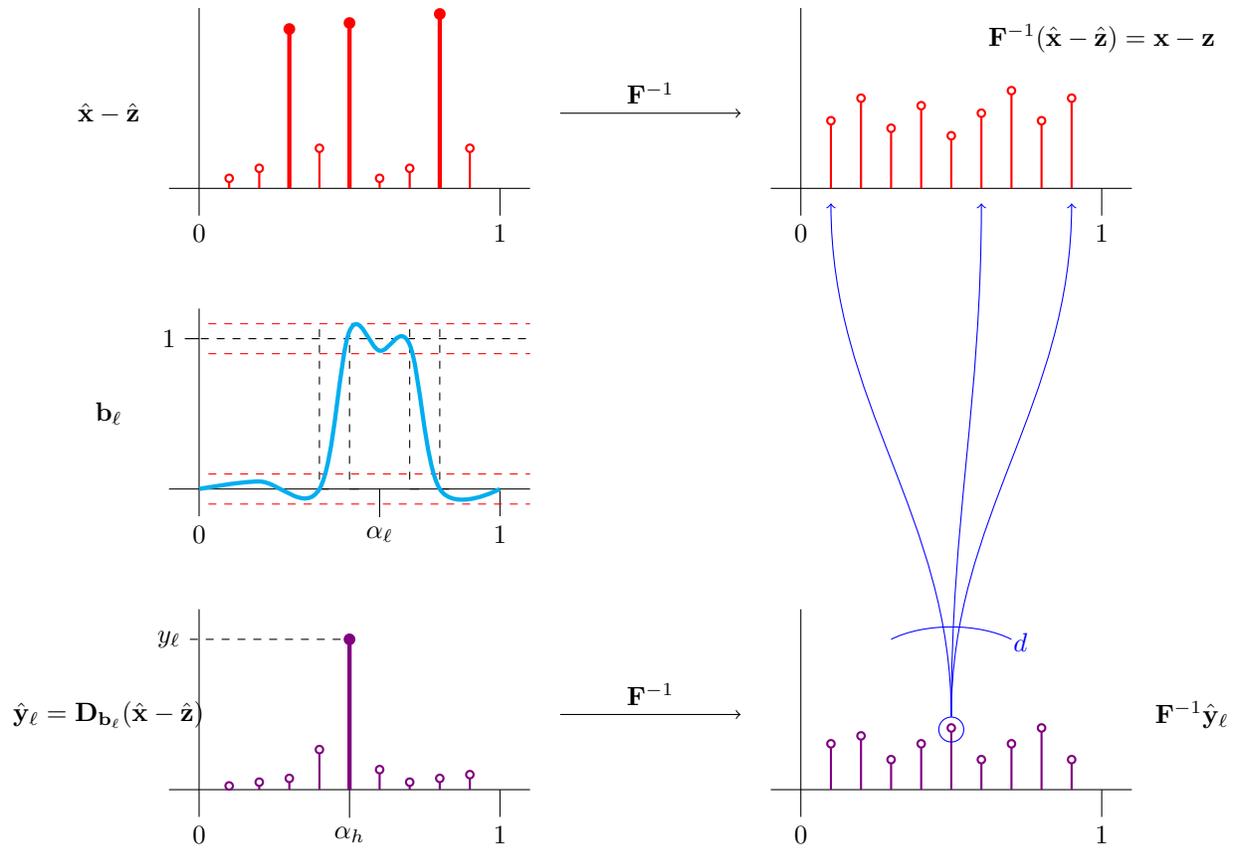

\begin{algorithm}
\caption{$\textsc{Peeler}^{(\vv)}(\vhz, k, \epsilon, \mu,\gamma)$}\label{alg:peel}
\begin{algorithmic}[1]
\Require Query access to $\vv = \vx + \vw = \vF^{-1}\vhx + \vF^{-1}\vhw$ and full access to $\vhz$ so that, for every $i\in\supp(\vhz)$,~\eqref{eq:peeler-condition} is satisfied. Additionally takes as input parameters $k$, $\eps$, $\mu$ and $\gamma$.
\Ensure A vector $\vtz$ and a Boolean value \texttt{stop}.  
\Statex

\State Choose parameters as follows (the constants in the $\Theta(\cdot)$ bounds will be implicitly defined in the proofs).

\State $\mu_0 \gets \Theta\parens{ \frac{ \mu^2 C_0^2 \gamma^2 }{k^2} }$
\State $T_0 \gets \Theta\parens{\frac{\log(1/\mu_0)}{C_0 \gamma}}$
\State $T_1 \gets \Theta\parens{\frac{NU^2 \log(2/\mu_0)}{\eps^2}}$
\Comment{In fact $T_1$ will get set again (to be the same thing) in \textsc{Verify}, but we list it here as well for the reader's convenience.}
\State $T_2 \gets \Theta( k \log_{1/\eps}(k/\eps) )$

\Comment{Notice that $\mu_0 \leq \frac{\mu}{2T_0T_2}$.}

\State $\mathcal{L}\gets \emptyset$
\For{$i=1\dots,T_2$} \label{alg:peeler:outer}
\For{$t=1\dots,T_0$} \label{alg:peeler:inner}
    \State Pick $0\le \ell <N$ uniformly at random. \label{alg:peeler:ell}
    \State Define $\vhy_{\ell}=\vD_{\vb_{\ell}}(\vhx + \vhw -\vhz)$ where $\vb_\ell[i] = b_\ell\parens{\evalpts_i}$.

\Comment{Recall that $b_\ell$ is the boxcar polynomial centered at $\theta_\ell$}

    \State Use \textsc{SimulateQueryAccess} (Alg.~\ref{alg:simulate}) to simulate query access to $\vy_{\ell}=\vF^{-1} \cdot \vhy_\ell$
    \State Let $(y_{\ell},h)$ be the output of the $\parens{N,6\eps,\cnst\eps,\frac {\mu_0}{2}}$ $1$-sparse recovery algorithm $\mathcal{A}$ with query access to $\vy_{\ell}=\vF^{-1}\cdot \vhy_{\ell}$  \label{alg:peeler:1-sparse} 

\Comment{Notice that this defines $h \in [N]$}

    \If{$\theta_h\in \brackets{\theta_{\ell}-\frac\gamma4,\theta_{\ell}+\frac\gamma4}$} \label{alg:peeler:range}
	\If {$\verify^{(\vy_{\ell})}(y_{\ell},h,\mu_0/2)$} \label{alg:peeler:verify} \Comment{Use Alg.~\ref{alg:simulate} to simulate query access to $\vy_{\ell}$}
            \State Add $\parens{y_{\ell},h}$ to  $\mathcal{L}$. \label{alg:peeler:list}
        \EndIf
    \EndIf
\EndFor
\State Let $(\bar{v},\bar{h})\in\mathcal{L}$ be the pair in $\mathcal{L}$ with the largest absolute value of the `$v$' component. \label{alg:peeler:final-vh}
\If{$\bar{h} \not\in\supp(\vhz)$}
    \State \Return{$\vhz+\bar{v}\cdot\ve_{\bar{h}}$, $\texttt{stop}=\false$ }
\Else
\State $\vhz \gets \vhz + \bar{v} \cdot \ve_{\bar{h}}$ \label{alg:peeler:update-z}
\EndIf
\EndFor
\State \Return{ $\vhz$, $\texttt{stop}=\true$ }
\end{algorithmic}
\end{algorithm}

\begin{algorithm}
\caption{\textsc{Preprocessing} (Done once before \textsc{Recover} (Alg.~\ref{alg:recover}) is ever run, and used in \textsc{SimulateQueryAccess} (Alg.~\ref{alg:simulate})).}\label{alg:preprocessing}
\begin{algorithmic}
\Require{A parameter $s$ and a description of an orthogonal polynomial family $p_1, \ldots, p_N$}

\Comment{We should choose $s = d$, where $d = O(\sqrt{k}/\gamma \eps)$ is as defined in the proof of Lemma~\ref{lem:redux-work-horse}.}

\Ensure{Matrices $\vM_0, \ldots, \vM_s$ which will be stored}

\State{Compute the roots $\lambda_1,\ldots, \lambda_N$ of $p_N$ and store them in a data structure as described in Corollary~\ref{cor:jac-roots}.}

\For {$r = 0, \ldots, s+1$}
	\State Let $\vM_r = \vP_N^T \vD_\vw \parens{\vD_{ \mathbf{\lambda}}}^r \vP_N$. 

	\Comment{Here, $\vD_{\mathbf{\lambda}}$ is the diagonal matrix with the evaluation points $\lambda_1, \ldots, \lambda_N$ on the diagonal.}

	\Comment{This can be done in time $O(N^2 \log(N))$ using a fast multiplication algorithm for orthogonal polynomial transforms~\cite{driscoll}.}

\EndFor
\end{algorithmic}
\end{algorithm}
\begin{algorithm}
\caption{$\textsc{SimulateQueryAccess}^{(\vx + \vw)}(\vhz, b_\ell, j)$}\label{alg:simulate}
\begin{algorithmic}
\Require Query access to $\vF^{-1}(\vhx + \vhw) $, the polynomial $b_\ell(X) = \sum_{r=0}^d b_{\ell, r} X^t$, and an index $j$
\Ensure $(\vF^{-1}\vhy_\ell)[j]$, where $\vhy_\ell$ is as in Algorithm~\ref{alg:peel}
\Statex

\For { $i$ so that $|j-i| \leq d$ }
	\State Compute $\nu_i = \sum_{r=0}^d b_{\ell, r} \vM_r[j,i]$
	\Comment{ $\vM_s$ was computed in \textsc{Preprocessing} (Alg.~\ref{alg:preprocessing}) }

	\State Compute $\vz[i] = (\vF^{-1} \vhz)[i]$
	\Comment{ Takes time $O(\|\vhz\|_0) = O(k)$ }
\EndFor

\State \Return{ $\sum_{i : |j-i| \leq d} \nu_i ( \vx[i] + \vw[i] - \vz[i] )$ }
\end{algorithmic}
\end{algorithm}

\begin{algorithm}
\caption{$\verify^{(\vy)}(v,h,\mu)$}\label{alg:verify}
\begin{algorithmic}
\Require $v\in\R$, $h\in [N]$, and query access to $\vy=\vF^{-1}\vhy$
\Ensure $\true$ if $v\cdot \ve_h$ is the only `spike' in $\vhy$ with failure probability $\mu$.

\Statex
\State $T_1 \gets \Theta\left( \frac{NU^2\log(1/\mu)}{\epsilon^2} \right)$
\State Choose $\Omega\subseteq [N]$ of size $T_1$ by sampling elements of $[N]$ uniformly at random with replacement.
\State $g\gets \frac N{T_1} \sum_{j\in\Omega}\parens{\trunc{100|v|U}{\vF^{-1}\parens{\vhy - v\ve_h}[j]}}^2$
\Comment{Estimate $\|\vhy-v\cdot \ve_h\|_2^2$}
\If{$g>\frac {v^2}{1000}$}
\State \Return{\false}
\Else \State \Return{\true}
\EndIf
\end{algorithmic}
\end{algorithm}

\subsubsection{Simulating Query Access}
We begin with the analysis of \textsc{SimulateQueryAccess} (Alg.~\ref{alg:simulate}), which allows \textsc{Peeler} (Alg.~\ref{alg:peel}) to simulate query access to $\vF^{-1} \vhy_\ell$. In \textsc{Peeler}, this is needed both to run the one-sparse recovery algorithm, and to run the $\verify$ algorithm.

We first prove a general property about low-degree polynomials:
\blmm
\label{lem:row-sparse}
Let $b(X)$ be any polynomial of degree $d$.
Then for all $0 \leq j < N$, the only values $0 \leq i < N$ so that
\begin{equation}
\label{eq:orth-poly-int-new}
\sum_{t=0}^{N-1} p_i(\lambda_{t}) p_j(\lambda_{t}) b(\lambda_{t})w_{t}\neq0
\end{equation}
are those so that $|i - j |\leq d$.
In particular, if we define $\vb=\parens{b(\lambda_0),\dots,b(\lambda_{N-1})}$, then each row $j$ of the matrix $\vF^T\vD_{\vb}\vF$ has at most $O(d)$ non-zero values in it, at positions $i$ so that $|i - j| \leq d$.
\elmm
\begin{proof}
Since the orthogonal polynomials $\{p_0, p_1, p_2, \ldots, p_r\}$ form a basis for polynomials of degree at most $r$ (this follows since $\deg(p_i)=i$), it follows from orthogonality conditions that for any $r$ and for any polynomial $f(X)$ of degree strictly less than $r$, we have
\[ \int_{-1}^1 p_{r}(X)f(X)w(X)dX=0.\]
Then the above implies that
\[\int_{-1}^1 p_i(X)p_j(X)b(X)w(X)dX\]
is zero whenever $|i-j|>d$. Further,~\eqref{eq:quadrature} implies that we have
\[\sum_{t=0}^{N-1} p_i(\lambda_{t}) p_j(\lambda_{t}) b(\lambda_{t})w_{t} = \int_{-1}^1 p_i(X)p_j(X)b(X)w(X)dX\]
whenever $i+j+d\le 2N-1$.
Thus, \eqref{eq:orth-poly-int-new} only holds if either $|i - j| \leq d$ or if $i + j + d \geq 2N$.

Because $i,j \in \{0,\ldots,N-1\}$, the only way that $i + j + d \geq 2N$ is if $|i - j| \leq d - 2$, which is already covered by the above.
\end{proof}

With this fact out of the way, we observe that \textsc{SimulateQueryAccess} works as intended.
\begin{prop} Each call to \textsc{SimulateQueryAccess} (Alg.~\ref{alg:simulate}) uses $O(d)$ queries to $\vx+\vw = \vF^{-1}(\vhx+\vhw)$, runs in time $O(kd)$, and returns $\vhy_\ell[j]$.
\end{prop}
\begin{proof}
As in Algorithm~\ref{alg:peel}, define
\[ \vhy_\ell = \vD_{\vb_\ell} (\vhx + \vhw - \vhz). \]
Thus,
\begin{align}
\label{eq:s1-new}
\vF^{-1}\vhy_{\ell}
& = \vF^{T}\cdot \vD_{\vb_{\ell}}\cdot \vF(\vx+\vw - \vz),
\end{align}
where $\vz = \vF^{-1} \vhz$ and in above we have used the definition of $\vhy_\ell$ and the fact that $\vF^{-1} = \vF^T$.

Now we have
\begin{equation}\label{eq:eta}
 (\vF^T \vD_{\vb_\ell} \vF)[j,i] =
\sum_{t=0}^{N-1} p_i(\lambda_{t}) p_j(\lambda_{t}) b_\ell(\lambda_{t})w_{t}
\end{equation}
which by
Lemma~\ref{lem:row-sparse} is only nonzero if $|i -j | \leq d$, since $d$ is the degree of $b_\ell$.
Further, we can expand $b_\ell(\lambda_t) = \sum_{r=0}^d b_{\ell,r} \lambda_t^r$ and observe that when \eqref{eq:eta} is nonzero, then it is equal to $\sum_{r=0}^d b_{r,j} \vM_r[j,i]$.  Thus, we have
\begin{align*}
(\vF^T \vD_{\vb_\ell} \vF)[j,i] =
\begin{cases}
0 & |i-j| > d  \\
\nu_i=\sum_{r=0}^d b_{\ell, r} M_r[j,i] & |i-j| \le d
\end{cases} \\
\end{align*}
since this is exactly how we have set $\nu_i$ in Algorithm~\ref{alg:simulate}.
Therefore by \eqref{eq:s1-new} we have
\[ \vF^{-1}\vhy_\ell[j] = \sum_{i=0}^{N-1} (\vF^T \vD_{\vb_\ell} \vF)[j,i] ( \vx[i] + \vw[i] - \vz[i] ) = \sum_{i=0}^{N-1} \nu_i( \vx[i] + \vw[i] - \vz[i]),\]
which is what is returned in \textsc{SimulateQueryAccess} (Alg.~\ref{alg:simulate}).

In order to compute $\sum_i \nu_i (\vx[i] + \vw[i] - \vz[i])$ we need $O(d)$ queries to $\vx$.  We also need $O(d)$ queries to $\vz$ which we can compute directly from $\vhz$ in time $O(k)$ per query (using the fact that $\vhz$ is $k$-sparse, and that $\vz = \vF^{-1} \vhz$).  This proves the proposition.
\end{proof}

\subsubsection{Correctness of the $\textsc{Peeler}$ algorithm}
In the rest of the proof of Lemma~\ref{lem:redux-work-horse}, we argue that $\vtz=\textsc{Peeler}^{(\vv)}(\vhz,k, \eps,\mu, \gamma)$ satisfies the required properties.

Recalling that $\vhr = \vhv - \vhz$, let $h^*$ be the location of the largest magnitude entry in $\vhr$. 
To prove the correctness of the \textsc{Peeler} algorithm, we will need the following two lemmas (which we will prove later in Sections~\ref{sec:good-spike-yes} and~\ref{sec:bad-spike-no}) that reason about any iteration of the inner loop on line~\ref{alg:peeler:inner} in Algorithm~\ref{alg:peel}.
We first argue that a `good spike' (that is, a value of $h$ so that $|\vhr[h]|$ is large) will pass the $\verify$ check and hence the largest heavy hitter in $\vhr=\vhv-\vhz$ will be included in the set $\lst$:
\begin{lmm}[Good spikes get noticed]
\label{lem:good-spike-yes}
Let $\eps$ be sufficiently small compared to $\cnst$.
Let $h\in [N]$.
Suppose that $\vhr[h]$ satisfies $\abs{\vhr[h]}\ge \frac 12\cdot\|\vhr\|_{\infty}$ and
that $\vhv = \vhx + \vhw$ so that
$\norm{2}{\vhw} \le \eps\norm{\infty}{\vhr}$.
Also suppose that
$\theta_h\in \brackets{\theta_{\ell}-\frac\gamma4,\theta_{\ell}+\frac\gamma4}$ (line~\ref{alg:peeler:range}).
Then in \textsc{Peeler}, with probability $1-\mu_0$ (over $\mathcal{A}$ (line~\ref{alg:peeler:1-sparse}) and \textsc{Verify} (line~\ref{alg:peeler:verify})), $(\tilde{v},h)$ gets added to $\lst$ such that (where $v=\vhr[h]$):
\[\abs{\tilde{v}-v}\le \cnst\eps\cdot\abs{v}.\]
\end{lmm}

Next, we will argue that `bad spikes' get caught by $\verify$ and hence they will not prevent the actual heavy hitter from being chosen:
\begin{lmm}[Bad spikes get pruned out]
\label{lem:bad-spike-no}
Let $\eps$ be sufficiently small compared to $\cnst$. 
Let $h\in [N]$.
Suppose that $\vhr[h]$ satisfies $\abs{\vhr[h]} < \frac 12\norm{\infty}{\vhr}$
and that $\vhv = \vhx + \vhw$ so that
$\norm{2}{\vhw} \le \eps\norm{\infty}{\vhr}$.
Then with probability at least $1-\mu_0$,
if $\abs{y_\ell} \ge (1-\cnst\eps)\norm{\infty}{\vhv-\vhz}$ then $(y_\ell,h)$ is not added to $\lst$.
\end{lmm}

Before we prove either of the two lemmas above, we first use them to argue the correctness of $\textsc{Peeler}$ algorithm (as claimed in Lemma~\ref{lem:redux-work-horse}).
Recall that by assumption, 
\begin{equation}\label{eq:preconditions}
 |\vhr[i]| \leq \cnst \eps |\vhv[i]| \qquad \text{and} \qquad |\vhv[i]| \geq (1 - 2\cnst \eps) \|\vhr\|_\infty
\end{equation}
for all $i \in \supp(\vhz)$.

\begin{claim}
  \label{claim:peeler-outer}
  In every iteration of the outer \textbf{For } loop (line~\ref{alg:peeler:outer} in Algorithm~\ref{alg:peel}),
  with probability at least $1-\mu/T_2$,
  \begin{equation}\label{eq:t1}
 | \bar{v} - \vhr[\bar{h}] | \leq \cnst \eps |\vhr[\bar{h}]|.
\end{equation}
and
\begin{equation}\label{eq:t2}
 |\vhr[\bar{h}]| \geq (1 - 2\cnst \eps) \norm{\infty}{\vhr}.
\end{equation}
  where $(\bar{v}, \bar{h}) \in \lst$ is the pair picked by the \textsc{Peeler} algorithm after the inner \textbf{For} loop has completed (line~\ref{alg:peeler:final-vh} in Algorithm~\ref{alg:peel}).
  (Above, $\vhr$ is the value of $\vhv-\vhz$ at the beginning of the loop.)
\end{claim}

\begin{proof}
Let us consider a single iteration of the outer \textbf{For} loop.

We first note that the conditions to Lemma~\ref{lem:good-spike-yes} and ~\ref{lem:bad-spike-no} cannot both be satisfied, and so exactly one of them is invoked in the analysis each time $\verify$ is called.
By a union bound over all $T_0$ iterations of the inner loop, we can say that with probability 
$1-T_0\mu_0$ 
(over all randomness in lines~\ref{alg:peeler:1-sparse} and~\ref{alg:peeler:verify}), then Lemmas~\ref{lem:good-spike-yes} and \ref{lem:bad-spike-no} have the favorable outcome every time they are invoked in one iteration of the outer loop.
By our choices of $T_0,T_2,\mu_0$, we have $\mu_0 \leq \frac{\mu}{2T_0T_2}$.  Indeed, plugging in our choices, we would like to show that 
 \[\mu_0 \log(1/\mu_0) \ll \frac{ \mu C_0 \gamma }{ k \log_{1/\eps} (k/\eps)} \]
for which it suffices to show
\[ \mu_0 \log(1/\mu_0) \ll \frac{ \mu C_0 \gamma }{ k \log(k) }. \]
Since $\mu_0 = \parens{ \frac{ \mu C_0 \gamma }{ k } }^2$, this is true (assuming $k$ is large enough).
Thus, the probability that Lemma~\ref{lem:good-spike-yes} and Lemma~\ref{lem:bad-spike-no} have the favorable outcome every time they are invoked in one iteration of the outer loop is at least $1 - \frac{\mu}{2T_2}$.
  
Suppose that this occurs, and now consider the
`spike' at $h^* = \mathrm{argmax}_i |\vhr[i]|$.
By Lemma~\ref{lem:isolation} and our choice of $T_0$, with probability $1-\frac{\mu}{2T_2}$ (over all randomness in line~\ref{alg:peeler:ell})
there is some iteration $t$ in the inner loop (line~\ref{alg:peeler:inner}) where
$\theta_{h^*}\in \brackets{\theta_{\ell}-\frac\gamma4,\theta_{\ell}+\frac\gamma4}$.
More precisely, we have
\begin{align*}
\mathbb{P}\left\{ \theta_{h^*} \not\in \brackets{\theta_{\ell}-\frac\gamma4,\theta_{\ell}+\frac\gamma4} \text{ for all }T_0\text{ iterations }\right\}
&\leq \parens{ 1 - \frac{C_0\gamma}{2} }^{T_0} \\
&\leq \exp\left( - \frac{ C_0 \gamma T_0 }{2} \right)\\
&\leq \poly(\mu_0) \\
&= \poly\left( \frac{ \mu^2 C_0^2 \gamma^2 }{ k^2 } \right)\\
&\leq \frac{\mu}{2 T_2}
\end{align*}
setting the constants in the definitions of $T_0$ and $\mu_0$ appropriately.
In this iteration where $\theta_{h^*}\in\brackets{\theta_{\ell}-\frac\gamma4,\theta_{\ell}+\frac\gamma4}$,
$h^*$ satisfies the conditions of Lemma~\ref{lem:good-spike-yes}, so that
$(\tilde{v},h^*)$ gets added to $\lst$ such that $\abs{\tilde{v}}\ge (1-\cnst\eps)\norm{\infty}{\vhr}$.
Thus with probability at least $1-\frac{\mu}{T_2}$, we have that $(\widetilde{v}, h^*) \in \lst$ and  Lemmas \ref{lem:good-spike-yes} and \ref{lem:bad-spike-no} have the favorable outcome if they are invoked, in every inner iteration during iteration of the outer loop. 

Now by the definition of $\bar{v}$, we must have $\bar{v} \ge \widetilde{v} \geq (1 - \cnst \eps) \norm{\infty}{\vhr}$.
Then $|\vhv[\bar{h}]| \geq \frac{1}{2} \norm{\infty}{\vhr}$ or else Lemma~\ref{lem:bad-spike-no} is contradicted.
Then Lemma~\ref{lem:good-spike-yes} implies~\eqref{eq:t1}, and combining with the previous equation establishes~\eqref{eq:t2}.
\end{proof}

By Claim~\ref{claim:peeler-outer} and a union bound, we have that with probability $1-\mu$ over the whole algorithm, \eqref{eq:t1} and~\eqref{eq:t2} hold for every iteration of the outer loop (line~\ref{alg:peeler:outer}).
It remains to establish that assuming this is true, the conclusion of Lemma~\ref{lem:redux-work-horse} holds.

We first note that if $\vhz_1, \vhz_2$ are the values of $\vhz$ before and after line~\ref{alg:peeler:update-z} respectively (with corresponding $\vhr_i=\vhv-\vhz_i$), then since $\vr_2=\vr_1-v\ve_{\bar h}$~\eqref{eq:t1} says that
\begin{equation}
  \label{eq:t3}
  \abs{\vhr_2[\bar h]} \le C\epsilon \abs{\vhr_1[\bar h]}.
\end{equation}
In other words entries of $\vhr$ only decrease in magnitude throughout the algorithm.
Thus for all $i\in\supp(\vhz)$ at the beginning of the algorithm, \eqref{eq:preconditions} continue to hold throughout the algorithm.

Now we consider the two cases: either a $\bar{h} \not\in \supp(\vhz)$ is found in some iteration and \textsc{Peeler} returns $\texttt{stop}=\false$, or one is never found and it returns $\texttt{stop}=\true$.
In the first case, we only need to establish that \eqref{eq:preconditions} holds for $\bar{h}$ as well at the time \textsc{Peeler} stops.
It is easy to see that~\eqref{eq:t1} and~\eqref{eq:t2} imply~\eqref{eq:preconditions} using the fact that $\vtz[\bar{h}]=\bar v$ and $\vhr[\bar{h}]=\vhv[\bar{h}]$. Further, by the assumption on $\norm{2}{\vhw}$ implies that for small enough $\eps$,~\eqref{eq:preconditions} implies that $\bar{h}\in\supp(\vhx)$.

Now consider the case that \textsc{Peeler} returns $\texttt{stop}=\true$.  That is, it completes the outer \textbf{For} loop and never chooses a pair $(\bar{v}, \bar{h})$ so that $\bar{h} \not\in \supp(\vhz)$.  In this case we will claim that $\vhz$ was already very close to $\vhv$ to begin with, in which case we will be done.

More precisely, given the choice of $T_2 = \Theta( k \log_{1/\eps}(k/\eps))$ and the fact that $|\supp(\vhz)| \leq k$, if \textsc{Peeler} returns $\texttt{stop}=\true$, then there is some $\bar{h} \in \supp(\vhz)$ that was chosen at least $\Theta(\log_{1/\eps}(k/\eps))$ times in the outer \textbf{For} loop.  Let $\vhz^{(0)}, \vhz^{(1)}, \ldots, \vhz^{(\log(k/\eps))}$ be the iterates of $\vhz$ during all of the times that $\bar{h}$ was chosen (at the start of the outer loop); let $\vhr^{(i)}$ be the corresponding residual vector $\vhv - \vhz^{(i)}$ and let $\bar{v}^{(i)}$ denote the corresponding value that was selected by \textsc{Peeler}. For each $i$, we have
\begin{align*}
\cnst \eps \norm{\infty}{\vhr^{(i)}} &\geq \cnst \eps | \vhr^{(i)}[\bar{h}] |\\
&\geq |\vhr^{(i+1)}[\bar{h}]| \\
&\geq (1 - 2\cnst\eps) \norm{ \infty}{ \vhr^{(i+1)} }
\end{align*}
where in the second line we used~\eqref{eq:t3} and in the third line we used~\eqref{eq:t2}.
Thus, we have
\[ \norm{ \infty }{\vhr^{(i+1)}} \leq \left( \frac{ \cnst \eps }{ 1 - 2 \cnst \eps } \right) \norm{\infty}{\vhr^{(i)}}. \]
Iterating this $\Theta(\log_{1/\eps}(k/\eps))$ times, we see that the final $\vtz$ that is returned by \textsc{Peeler} along with $\texttt{stop}=\true$ satisfies (assuming $\eps$ is sufficiently small)
\begin{equation}\label{eq:T2}
 \|\vtz - \vhv\|_\infty \leq \| \vhr^{(\Theta(\log_{1/\eps}(k/\eps)))} \|_\infty \leq \frac{\eps}{\sqrt{k}} \|\vhz - \vhv\|_\infty  \leq \frac{\eps}{\sqrt{k}} \|\vhv\|_\infty,
\end{equation}
which is what we required from \textsc{Peeler} in the case that $\texttt{stop}=\true$. 
Above, we are implicitly setting the constant inside the $\Theta$ in the definition of $T_2$ in \eqref{eq:T2}.
In the last inequality we used the fact that for $i\in\supp(\vhz)$, we have $\abs{\vhv[i]-\vhz[i]}\le C\eps\abs{\vhv[i]}\le \abs{\vhv[i]}$ for small enough $\eps$.

This completes the proof of the error guarantee of \textsc{Peeler}.

\subsubsection{Query and Time complexity of the \textsc{Peeler} algorithm}
The \textsc{Peeler} algorithm calls the 1-sparse recovery algorithm $T_0 \cdot T_2$ times, and for each of the $Q(N,6\eps, \cnst\eps, \mu_0/2)$ queries that the 1-sparse recovery algorithm uses, \textsc{SimulateQueryAccess} uses $O(d) = O(\sqrt{k}/\eps \gamma)$ queries.  Additionally, \textsc{Verify} is called $T_0 \cdot T_2$ times, and uses \textsc{SimulateQueryAccess} $T_1$ times each time it is called.  Thus, the total number of queries is
\[ O\left( T_0 \cdot T_2 \cdot \left( \frac{\sqrt{k}}{\eps \gamma} \right) \cdot (Q(N,6\eps, \cnst\eps, \mu_0/2) + T_1) \right). \]
The derivation of the running time is similar; the only overhead is that \textsc{simulateQueryAccess} has an additional factor of $k$ in its running time over its query complexity, leading to a running time of
\[ O\left( T_0 \cdot T_2 \cdot \left( \frac{k^{3/2}}{\eps \gamma} \right) \cdot (T(N,6\eps, \cnst\eps, \mu_0/2) + T_1) \right). \]

\subsubsection{Correctness of the $\verify$ algorithm}
We begin by analyzing the estimates in $\verify$ (which we will later use to prove Lemmas~\ref{lem:good-spike-yes} and~\ref{lem:bad-spike-no}):
\begin{lmm}[Estimator Lemma]
\label{lem:verify-estimate}
Let $\zeta\le 1$ be small enough. 
Let $\Omega$ be as chosen in $\verify$.
Suppose that $\vhu$ is of the form $\vhu=\vha+\vhq$, where $\norm{2}{\vhq}\le \zeta\norm{2}{\vha}$.
Consider the estimate
\[\Theta = \frac N{T_1}\sum_{j\in \Omega} \parens{\trunc{\tau}{\vu[j]} }^2,\]
where the threshold $\tau$ satisfies ($U$ is as defined in~\eqref{eq:U})
\begin{equation}
\label{eq:tau-lb}
\tau\ge \norm{1}{\vha}U.
\end{equation}
Then for any $\mu>0$, we have with probability at least $1-\mu$,
\begin{equation}
\label{eq:estimator-tight-bound}
\abs{\Theta -\norm{2}{\vu}^2} \le 9\zeta\norm{2}{\vu}^2
\end{equation}
for some choice of
\[T_1\ge \Theta\parens{\frac{N\tau^2 \log\parens{\frac 1{\mu}}}{\zeta^2\norm{2}{\vha}^2}}.\]

On the other hand for any $\vhu\in\R^N$, $\tau,\zeta>0$, the estimate $\Theta$ above satisfies $\Theta\le (1+\zeta)X^2$ with probability at least $1-\mu$ for any $X^2\ge \norm{2}{\vhu}^2$ for some choice of
\[T_1\ge \Theta\parens{\frac{N\tau^2\log\parens{\frac 1{\mu}}}{\zeta^2X^2}}.\]
\end{lmm}

We will use Bernstein's inequality in the proof of above lemma, which we recall next:
\bthm[Bernstein's Inequality]
\label{thm:bernstein}
Let $X_1,\dots,X_n$ be independent random variables with $|X_i|\le M$, then we have
\[\Prob{ \abs{\sum_{i=1}^n X_i-\sum_{i=1}^n\avg{X_i}}>t}\le 2\exp\parens{-\frac{t/2 }{\frac{1}{t}\cdot \sum_{i=1}^n \avg{X_i^2 }+M/3}}.\]
\ethm

\begin{proof}[Proof of Lemma~\ref{lem:verify-estimate}]
We start with the case when $\vhu=\vha+\vhq$, where $\norm{2}{\vhq}\le \zeta\norm{2}{\vha}$.
Then note that
\[\vu=\vF^{-1}\vhu = \vF^{-1}\vha+\vq,\]
where $\vq=\vF^{-1}\vhq$. By definition of $U$, we have that
\begin{equation}
\label{eq:a1+a2-infty}
\norm{\infty}{\vF^{-1}\vha} \le \norm{1}{\vha} \cdot U.
\end{equation}
For notational convenience, define $\vt=\trunc{\tau}{\vu}$ and $\vht$ accordingly. 
Next, we claim that by~\eqref{eq:tau-lb},  we have
\[\trunc{\tau}{\vu}= \vF^{-1}\vha+\vz\text{ where } \norm{2}{\vz}\le \norm{2}{\vq}.\]
Indeed, we will argue that component-wise the elements of $\vz$ are smaller (in absolute value) than $\vq$.
If $\abs{\vu[i]}\le \tau$, then the claim is trivially true.
If not, then $|\vu[i]| > \tau$, $|\vt[i]| = \tau$, and by~\eqref{eq:a1+a2-infty} $|\vF^{-1}\vha[i]| \le \tau$.
Since $\vu[i]$ and $\vt[i]$ have the same sign, $|\vz[i]| = |\vt[i] - \vF^{-1}\vha[i]| \le |\vu[i] - \vF^{-1}\vha[i]| = |\vq[i]|$, as required.

Thus, we have $\vF\cdot\trunc{\tau}{\vu}=\vha+\vhz$ where $\norm{2}{\vhz}\le \norm{2}{\vhq}$. 
Now by definition, we have
\begin{equation}
\label{eq:norm-w-infty-ub}
\norm{\infty}{\vt}\le \tau.
\end{equation}
To show that~\eqref{eq:estimator-tight-bound} holds, we first show that a similar bound holds for $\vu$ replaced by $\vt$ (via Bernstein's inequality) and then show that $\norm{2}{\vu}\approx\norm{2}{\vt}$. 
Recall that the set $\Omega$ in $\verify$ is chosen by including $T_1$ indices of $[N]$ uniformly at random with replacement.  Suppose these indices are $j_1, \ldots, j_{T_1}$, and let $X_i = \vt[j_i]^2$.  Thus,
\[ \sum_{i=1}^{T_1}X_i = \sum_{j \in \Omega} \vt[j]^2. \]
Further, for all $i \in [T_1]$, we have $|X_i| \leq \|\vt\|_\infty^2$ and $\mathbb{E} X_i^2 = \frac{1}{N} \sum_{i=1}^N \vt[i]^4$.  
Thus by Bernstein's inequality with $t = \frac{T_1}{N} \cdot \zeta \norm{2}{\vt}^2$, we have
\[\Prob{ \abs{\sum_{j\in \Omega} \vt[j]^2-\frac{T_1}N\cdot \sum_{i=1}^n \vt[i]^2}> \frac{T_1}N\cdot\zeta\norm{2}{\vt}^2} \le 2\exp\parens{-\frac{\frac{T_1}{2N}\cdot\zeta\norm{2}{\vt}^2 }{\frac N{T_1 \zeta \|\vt\|_2^2 }\cdot \frac{T_1}{N}\sum_{i=1}^N \vt[i]^4 +\norm{\infty}{\vt}^2/3 }}.\]
Multiplying both sides of the inequality inside the probability expression by $\frac N{T_1}$ and replacing each sum by the appropriate norm expressions, we have
{\allowdisplaybreaks
\begin{align}
\Prob{\abs{\frac N{T_1}\cdot \sum_{j\in\Omega }\vt[j]^2 - \norm{2}{\vt}^2}>\zeta\norm{2}{\vt}^2} & \le 2\cdot\exp\parens{-\frac{T_1}{2N}\cdot \frac{\zeta^2\norm{2}{\vt}^4 }{\norm{4}{\vt}^4+\zeta\norm{2}{\vt}^2\norm{\infty}{\vt}^2/3 }}\notag\\
\label{eq:l1}
& \le 2\cdot\exp\parens{-\frac{T_1}{2N}\cdot \frac{\zeta^2\norm{2}{\vt}^4 }{\norm{2}{\vt}^2\cdot \norm{\infty}{\vt}^2+\zeta\norm{2}{\vt}^2\norm{\infty}{\vt}^2/3 }}\\
\label{eq:l2}
& \le 2\cdot\exp\parens{-\frac{T_1}{2N}\cdot \frac{\zeta^2\norm{2}{\vt}^2 }{\tau^2\parens{1+\frac{\zeta}{3 }}}}\\
\label{eq:l3}
&\le \mu.
\end{align}
}
In the above~\eqref{eq:l1} follows from the fact that $\norm{4}{\vt}\le \sqrt{\norm{\infty}{\vt}\cdot\norm{2}{\vt}}$, and \eqref{eq:l2} follows from~\eqref{eq:norm-w-infty-ub}. Finally,~\eqref{eq:l3} follows by choosing
\begin{equation}
\label{eq:T1-first-choice}
T_1\ge \Theta\parens{\frac{N\tau^2\log(1/\mu)}{\zeta^2\norm{2}{\vt}^2}}.
\end{equation}
Thus, we have argued that with probability at least $1-\mu$, we have
\begin{equation}
\label{eq:estimate-bound-1}
\abs{\Theta-\norm{2}{\vt}^2}\le \zeta\norm{2}{\vt}^2.
\end{equation}
However, we wanted to prove a similar result with $\vt$ replaced by $\vu$. We do so next by showing that $\norm{2}{\vu}\approx\norm{2}{\vt}$. Indeed consider the following sequence of inequalities:
\[\norm{2}{\vha}-\norm{2}{\vq}\le \norm{2}{\vha}-\norm{2}{\vz} \le \norm{2}{\vt}\le \norm{2}{\vu} \le \norm{2}{\vha}+\norm{2}{\vq},\]
where the first inequality follows from the fact that $\norm{2}{\vz}\le \norm{2}{\vq}$, the second and the last inequality follow from the triangle inequality and the third inequality follows from the definition of the truncation function. Applying the bound $\norm{2}{\vq}\le\zeta\norm{2}{\vha}$, we get that
\[(1-\zeta)\norm{2}{\vha} \le \norm{2}{\vt}\le\norm{2}{\vu}\le (1+\zeta)\norm{2}{\vha}.\]
In other words for $\zeta\le 1$,
\[(1-2\zeta)\norm{2}{\vu}\le \frac{1-\zeta}{1+\zeta}\norm{2}{\vu}\le\norm{2}{\vt}\le \norm{2}{\vu}.\]
Applying the above in~\eqref{eq:estimate-bound-1}, we get that with probability at least $1-\mu$
\[\abs{\Theta-\norm{2}{\vu}^2} \le \abs{\Theta-\norm{2}{\vt}^2} + \abs{\norm{2}{\vu}^2-\norm{2}{\vt}^2}\le \zeta \norm{2}{\vu}^2+\abs{\norm{2}{\vu}^2-\norm{2}{\vt}^2}\le \norm{2}{\vu}^2\parens{\zeta +1-(1-2\zeta)^2}.\]
Further, using the fact that $\zeta\le 1$, implies that the RHS is upper bounded by $9\zeta\norm{2}{\vu}^2$, as desired. Finally, the bound $\norm{2}{\vt}\le (1+\zeta)\norm{2}{\vha}$ along with~\eqref{eq:T1-first-choice} implies the required bound on $T_1$ in the lemma statement.

For the second part, we consider the same random variables $X_i$ as defined above but we use $t=\frac{T_1}N\cdot \zeta X^2$. Then by applying Bernstein's inequality we get that
\[\Prob{\Theta>\norm{2}{\vt}^2+\zeta X^2}\le  2\cdot\exp\parens{-\frac{T_1}{2N}\cdot \frac{\zeta^2 X^4 }{\norm{2}{\vt}^2\cdot \norm{\infty}{\vt}^2+\zeta X^2\norm{\infty}{\vt}^2/3 }}.\]
Noting that $\norm{2}{\vt}^2\le \norm{2}{\vu}^2\le X^2$, the above (along with~\eqref{eq:norm-w-infty-ub}) implies that
\[\Prob{\Theta> (1+\zeta)X^2} \le \Prob{\Theta>\norm{2}{\vt}^2+\zeta X^2}\le 2\cdot\exp\parens{-\frac{T_1}{2N}\cdot \frac{\zeta^2 X^2 }{\tau^2\parens{1+\frac{\zeta}{3 }}}}\le \mu,\]
where the last inequality follows from the choice of $T_1$ in the second part of the lemma. The proof is complete.
\end{proof}

\subsubsection{Proof of Lemma~\ref{lem:good-spike-yes}}
\label{sec:good-spike-yes}

In this section, we use Lemma~\ref{lem:verify-estimate} to prove~\ref{lem:good-spike-yes}.

First note that $\norm{2}{\vw}\le \eps \norm{\infty}{\vhv-\vhz}< \abs{v}$ for $\eps< \frac 12$. Hence we can assume that $h\in\supp(\vhx)$.
Now suppose $\theta_h\in \brackets{\theta_{\ell}-\frac\gamma4,\theta_{\ell}+\frac\gamma4}$. Then by part (2) of Lemma~\ref{lem:isolation}, we have that no other $\theta_{h'}$ for  $h'\in\supp(\vhx)$ falls in the range
\[R=\brackets{\theta_{\ell}-\frac{\gamma}{2},\theta_{\ell}+\frac{\gamma}{2}}.\]
This implies that we have $\vhy_{\ell}$ (as defined in \textsc{Peeler}) can be expressed as
\[\vhy_{\ell} = v\cdot \ve_h+\vhq,\]
where $\vhq$ in the range $R$ only has contribution from $\vhw$ (multiplied by a factor of at most $1 + \eps/\sqrt{k}$) and outside the range $R$ it has contribution from $(\vhv-\vhz)$ (multiplied by a factor of $\frac{\eps}{\sqrt{k}}$).
(See Figure~\ref{fig:goodlemma}.)

\begin{figure}
\begin{center}
\begin{tikzpicture}[xscale=2]
\node at (4.5,1.5) {$\vhy_\ell = \mathbf{D}_{\mathbf{b}_\ell} \vhr = \textcolor{violet}{v \mathbf{e}_h} + \textcolor{orange}{\vhq}$};
\draw (0,0) to (5,0);
\node(ell) at (2.5, -.5) {$\theta_\ell$};
\draw[dashed] (ell) to (2.5, 2);
\node(mg) at (1.5, -1) {};
\draw[dashed] (mg) to (1.5, 2);
\node(pg) at (3.5, -1) {};
\draw[dashed] (pg) to (3.5, 2);
\node(mmg) at (2, -1) {};
\draw[dashed] (mmg) to (2, 2);
\node(ppg) at (3, -1) {};
\draw[dashed] (ppg) to (3, 2);
\node(h) at (2.7, -.2) {$h$};
\node[draw, thick, violet, circle, fill=violet!40, scale=.3](hval) at (2.7, 1.75) {};
\draw[thick, violet] (2.7,0) to (hval);
\draw [decoration={brace},decorate] (2,-1) to node[below]{$\gamma/4$} (1.5,-1);
\draw [decoration={brace},decorate] (2.5,-1) to node[below]{$\gamma/4$} (2,-1);
\draw [decoration={brace},decorate] (3,-1) to node[below]{$\gamma/4$} (2.5,-1);
\draw [decoration={brace},decorate] (3.5,-1) to node[below]{$\gamma/4$} (3,-1);
\draw [decoration={brace},decorate] (3.5,-1.5) to node[below]{$R$} (1.5,-1.5);

\node[violet] at (2.75,2.25) {$v \geq \frac{1}{2}$};
\draw[thick, orange] plot [smooth, tension=1] coordinates { (0,.2) (.5,.05) (1,-.1) (1.5,.2) (1.75, .5) (2,.8) (2.5, 1) (2.8,.9) (3,1) (3.2,.6) (3.5,.1) (4,.2) (4.5, 0) (5,.1)};
\node[orange] at (.5, .5) {$\vhq$};
\end{tikzpicture}
\end{center}
\caption{The set-up for the proof of Lemma~\ref{lem:good-spike-yes}.  The contribution to $\vhq$ from within $R\setminus\{h\}$ comes only from $\vhw$, and is pointwise multiplied by a boxcar polynomial $\mathbf{b_\ell}$ with value at most $1 + \eps/\sqrt{k}$.  The contribution to $\vhq$ from outside of $R$ comes from both $\vhw$ and $\vhx$, but is pointwise multiplied by a boxcar polynomial with value at most $\eps/\sqrt{k}$.}\label{fig:goodlemma}
\end{figure}
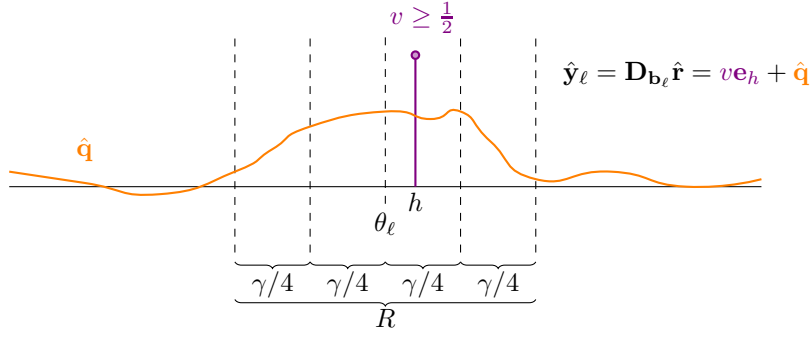
Recalling that $\vhr = \vhv - \vhz$, we can bound the noise $\vhq$ by
\begin{align*}
  \norm{2}{\vhq}^2 &= \norm{2}{\vD_{\vb_\ell}\cdot\vhr - \vhr[h]\cdot\ve_h}^2 = \parens{\vb_\ell[h]\vhr[h] - \vhr[h]}^2 + \norm{2}{\parens{\vD_{\vb_\ell}\vhr}_{R\setminus\{h\}}}^2 + \norm{2}{\parens{\vD_{\vb_\ell}\vhr}_{R^c}}^2 \notag\\
  &\le \parens{\frac{\epsilon}{\sqrt{k}}}^2\vhr[h]^2 + \parens{1+\frac{\epsilon}{\sqrt{k}}}^2\norm{2}{\vhw_{R\setminus\{h\}}}^2 + \parens{\frac{\epsilon}{\sqrt{k}}}^2\norm{2}{\vhr_{R^c}}^2 \notag\\
  &\le \parens{\frac{\epsilon}{\sqrt{k}}}^2\norm{2}{\vhr_{\supp(\vhx)}}^2 + \parens{1+\frac{\epsilon}{\sqrt{k}}}^2\norm{2}{\vhw}^2 \notag\\
  &\le \parens{\epsilon\norm{\infty}{\vhr}}^2 + 4\norm{2}{\vhw}^2 \notag\\
  &\le 5\epsilon^2\norm{\infty}{\vhr}^2. \notag\\
\end{align*}
In the above the second line uses the properties of the boxcar polynomial $\vb_\ell$ and the fact that 
\[ \| \vhr_{R\setminus \{h\}} \|_2^2 = \|(\vhw - \vhz)|_{R \setminus \{h\}} \|_2^2 = \| \vhw_{R \setminus \{h\}}\|_2^2 \]
by the fact that $\vhx_{R \setminus \{h\}}=\vzero$ and the fact that $\supp(\vhz)\subseteq \supp(\vhx)$ (which in turn is always maintained in each application of Lemma~\ref{lem:redux-work-horse}).
The third line follows from the fact that we can write 
\[\|\vhr_{R^c}\|_2^2 = \|\vhr_{R^c \cap \supp(\vhx)}\|_2^2 + \|\vhr_{R^c \setminus \supp(\vhx)}\|_2^2 \leq \| \vhr_{\supp(\vhx)}\|_2^2 + \| \vhw_{R^c}\|_2^2. \]
The fourth line follows from the fact that $|\supp(\vhx)| \le k$, and the last line uses the assumptions on $\norm{2}{\vw}$ in the lemma statement.
Finally the assumption on $\abs{v}$ in the lemma statement yields
  \[ \norm{2}{\vhq} \le 3\eps\norm{\infty}{\vhr} \le 6\eps\abs{v}. \]

The above implies that the `error' is small enough to run our $\parens{N,6\eps,\cnst\eps,\frac{\mu_0}2}$  $1$-sparse recovery algorithm. 
Hence, with probability at least $1 - \mu_0/2$, we get an estimate $(y_{\ell},h)$ from the $1$-sparse recovery algorithm such that
\begin{equation}
\label{eq:call-to-1sps}
\abs{y_{\ell}-v}\le {\cnst\cdot \eps} \abs{v}.
\end{equation}

Now we apply the second part\footnote{Note that $\tau=100\abs{v}U\ge \abs{v}U$ and hence it satisfies~\eqref{eq:tau-lb}.}
of Lemma~\ref{lem:verify-estimate} with $\vhu=\vhy-y_\ell\ve_h$, $X=\frac{\abs{y_{\ell}}}{40}$ and\footnote{We need $\norm{2}{\vhy-y_\ell\ve_h}\le (6+\cnst)\eps\abs{v}\le X$. Note that $(6+\cnst)\eps\abs{v}\le \frac{(6+\cnst)\eps\abs{y_\ell}}{(1-\cnst\eps)}$ and then the required bound follows for a sufficiently small choice of $\eps$.} $\zeta=\frac 35$ to note that the estimate $g$ in $\verify$ is, with probability at least $1-\frac{\mu_0}2$:
\begin{equation}
\label{eq:verify-cond2-no}
g\le \frac 85\cdot\frac{y_{\ell}^2}{1600}<\frac{y_{\ell}^2}{1000},
\end{equation}
as long as
\[T_1\ge \Theta\parens{ \frac{N\tau^2\log(2/\mu_0)}{X^2}},\]
which along with the fact that we  picked $\tau=100\abs{y_{\ell}}U$,~\eqref{eq:call-to-1sps} implies that we would be fine if we picked:
\[T_1\ge \Theta\parens{ N U^2\log(2/\mu_0)},\]
which we did.

Finally, note that by the union bound with probability at least $1-\mu_0$,
we have that both~\eqref{eq:call-to-1sps} and~\eqref{eq:verify-cond2-no} are satisfied. In other words, the 1-sparse recovery algorithm succeeds and $\verify$ returns $\true$, which implies that $(y_{\ell},h)$ gets added to $\lst$.
This along with~\eqref{eq:call-to-1sps} completes the proof.

\subsubsection{Proof of Lemma~\ref{lem:bad-spike-no}}
\label{sec:bad-spike-no}

Suppose that in a given iteration, the $1$-sparse recovery solver returns $(y_{\ell},h)$.
Recall that we have $\theta_h\in \brackets{\theta_{\ell}-\frac{\gamma}{4},\theta_{\ell}+\frac{\gamma}{4}}$ as otherwise $h$ will not get added to $\lst$.
If $\abs{y_{\ell}} < \parens{1-\cnst\eps}\norm{\infty}{\vhv-\vhz}$, then we have nothing to prove.
So for the rest of the proof, let us assume
\begin{equation}
\label{eq:y-l-lb}
\abs{y_{\ell}} \ge \parens{1-\cnst\eps}\norm{\infty}{\vhv-\vhz}.
\end{equation}

Let us look at the potential spikes captured besides $h$,
\[\supp(\vhx)\cap \brackets{\theta_{\ell}-\frac{\gamma}{2},\theta_{\ell}+\frac{\gamma}{2}} \setminus \set{h} \] 
Note that by part (2) of Lemma~\ref{lem:isolation} there can only be one such $h'$,
so this set has the form either $\{h'\}$ for some $h' \ne h$, or it is empty in which case for shorthand we set $h'=\bot$.

Let $\vhy_{\ell}$ be as defined in \textsc{Peeler} algorithm.
Now the vector $\vhy_{\ell}$ has two potential locations of interest at $h$ and $h'$ and the rest can be expressed as a vector $\vhq$
\[
  \vhy_\ell = v \cdot\ve_h + v' \cdot\ve_{h'} + \vhq,
\]
where $v=\vhr[h]$ and $v'=\vhr[h']$
Note in the case $h'=\bot$ the second term disappears and this says $\vhy_\ell = \vhr[h]\cdot\ve_h + \vhq$, the same as in Lemma~\ref{lem:good-spike-yes}.  (See Figure~\ref{fig:badlemma}.)
\begin{figure}
\begin{center}
\begin{tikzpicture}[xscale=2]
\node at (5,1.5) {$\vhy_\ell = \mathbf{D}_{\mathbf{b}_\ell} \vhr = \textcolor{blue}{v' \mathbf{e}_{h'}} + \textcolor{violet}{v \mathbf{e}_h} + \textcolor{orange}{\vhq}$};
\draw (0,0) to (5,0);
\node(ell) at (2.5, -.5) {$\theta_\ell$};
\draw[dashed] (ell) to (2.5, 2);
\node(mg) at (1.5, -1) {};
\draw[dashed] (mg) to (1.5, 2);
\node(pg) at (3.5, -1) {};
\draw[dashed] (pg) to (3.5, 2);
\node(mmg) at (2, -1) {};
\draw[dashed] (mmg) to (2, 2);
\node(ppg) at (3, -1) {};
\draw[dashed] (ppg) to (3, 2);
\node(h) at (2.7, -.2) {$h$};
\node(h) at (2.2, -.2) {$h'$};
\node[draw, thick, violet, circle, fill=violet!40, scale=.3](hval) at (2.7, 1.2) {};
\node[draw, thick, blue, circle, fill=blue!40, scale=.3](hprimeval) at (2.2, 2) {};
\draw[thick, violet] (2.7,0) to (hval);
\draw[thick, blue] (2.2,0) to (hprimeval);
\draw [decoration={brace},decorate] (2,-1) to node[below]{$\gamma/4$} (1.5,-1);
\draw [decoration={brace},decorate] (2.5,-1) to node[below]{$\gamma/4$} (2,-1);
\draw [decoration={brace},decorate] (3,-1) to node[below]{$\gamma/4$} (2.5,-1);
\draw [decoration={brace},decorate] (3.5,-1) to node[below]{$\gamma/4$} (3,-1);
\draw [decoration={brace},decorate] (3.5,-1.5) to node[below]{$R$} (1.5,-1.5);

\draw[thick, orange] plot [smooth, tension=1] coordinates { (0,.2) (.5,.05) (1,-.1) (1.5,.2) (1.75, .5) (2,.8) (2.5, 1) (2.8,.9) (3,1) (3.2,.6) (3.5,.1) (4,.2) (4.5, 0) (5,.1)};
\node[orange] at (.5, .5) {$\vhq$};
\end{tikzpicture}
\end{center}
\caption{The set-up for the proof of Lemma~\ref{lem:bad-spike-no}.  The contribution to $\vhq$ from within $R\setminus\{h,h'\}$ comes only from $\vhw$, and is pointwise multiplied by a boxcar polynomial $\mathbf{b_\ell}$ with value at most $1 + \eps/\sqrt{k}$.  The contribution to $\vhq$ from outside of $R$ comes from both $\vhw$ and $\vhx$, but is pointwise multiplied by a boxcar polynomial with value at most $\eps/\sqrt{k}$.  Notice that $h'$ might be $\bot$, in which case the picture looks similar to Figure~\ref{fig:goodlemma}.}\label{fig:badlemma}
\end{figure}

First, we note that the same derivation as in proof of Lemma~\ref{lem:good-spike-yes} works now to get an upper bound on $\|\vhq\|_2$:
\begin{align}
  \norm{2}{\vhq}^2 &= \norm{2}{\vD_{\vb_\ell}\cdot\vhr - \vhr[h]\cdot\ve_h - \vhr[h'] \cdot \ve_{h'}}^2 \notag\\
&= \parens{\vb_\ell[h]\vhr[h] - \vhr[h]}^2 + \parens{\vb_\ell[h']\vhr[h'] - \vhr[h']}^2 + \norm{2}{\parens{\vD_{\vb_\ell}\vhr}_{R\setminus\{h,h'\}}}^2 + \norm{2}{\parens{\vD_{\vb_\ell}\vhr}_{R^c}}^2 \notag\\
  &\le \parens{\frac{\epsilon}{\sqrt{k}}}^2(\vhr[h]^2 + \vhr[h']^2) + \parens{1+\frac{\epsilon}{\sqrt{k}}}^2\norm{2}{\vhw_{R\setminus\{h,h'\}}}^2 + \parens{\frac{\epsilon}{\sqrt{k}}}^2\norm{2}{\vhr_{R^c}}^2 \notag\\
  &\le \parens{\frac{\epsilon}{\sqrt{k}}}^2\norm{2}{\vhr_{\supp(\vhx)}}^2 + \parens{1+\frac{\epsilon}{\sqrt{k}}}^2\norm{2}{\vhw}^2 \notag\\
  &\le \parens{\epsilon\norm{\infty}{\vhr}}^2 + 4\norm{2}{\vhw}^2 \notag\\
  &\le 5\epsilon^2\norm{\infty}{\vhr}^2. \label{eq:badqsmall}
\end{align}

Next we prove a lower bound on the magnitude of the vector that \verify{} estimates:
\begin{align}
  \norm{2}{\vhy_{\ell}-y_{\ell}\ve_h} &\ge \abs{\vhy_\ell[h] - y_\ell} 
= \abs{\vhr[h]-y_\ell + \vhq[h]} 
\ge \abs{\vhr[h]-y_\ell} - \abs{\vhq[h]} 
                                        \ge \abs{y_\ell}-\abs{\vhr[h]} - \norm{2}{\vhq} 
                                        \notag
  \\&\ge \abs{y_\ell}-\frac{1}{2}\norm{\infty}{\vhr} - 3\epsilon\norm{\infty}{\vhr}
  \ge \parens{1-\frac{1/2+3\eps}{1-C\epsilon}}\abs{y_\ell}
  \notag
  \\&\ge \parens{1-\parens{\frac 12+3\eps} \parens{1+2C\epsilon}}\abs{y_\ell}
  \ge \parens{\frac 12 - (C+4)\epsilon}\abs{y_\ell}
  \label{eq:y-l-bad-spike-lb}
  \\&\ge \parens{1-C\epsilon}\parens{\frac 12 - (C+4)\epsilon}\norm{\infty}{\vhr}
  \label{eq:y-l-bad-spike-lb-r}
\end{align}
The first line follows from the triangle inequality.   The second line follows from our assumption $|\vhr[h]| \leq \frac{1}{2} \| \vhr \|_\infty$ and from \eqref{eq:badqsmall}.

Now we apply the first part of Lemma~\ref{lem:verify-estimate} with $\vhu=\vhy_\ell-y_\ell\ve_h = \left[ \parens{\vhr[h]-y_\ell}\ve_h+\vhr[h']\ve_{h'} \right] + \vhq$ and $\zeta=24\eps$.
First, we verify that for small enough $\eps$
\[
  \norm{2}{\vhq} \le 3\epsilon\norm{\infty}{\vhr} \le 24\epsilon\parens{1-C\epsilon}\parens{\frac 12 - (C+4)\epsilon}\norm{\infty}{\vhr} \le 24\eps\cdot\parens{\frac 12\cdot \abs{y_\ell}}\le 24\eps\cdot\parens{\abs{\vhr[h]-y_\ell}}\le 24\epsilon\norm{2}{\parens{\vhr[h]-y_\ell}\ve_h+\vhr[h']\ve_{h'}}
\]
and
\[
  \tau = 100U\cdot|y_\ell| \ge U\cdot\parens{1+\frac{2}{1-C\epsilon}}\abs{y_\ell} \ge U\cdot \parens{2\norm{\infty}{\vhr}+\abs{y_\ell}} \ge \norm{1}{\parens{\vhr[h]-y_\ell}\ve_h+\vhr[h']\ve_{h'}} \cdot U.
\]
Thus, Lemma~\ref{lem:verify-estimate} along with~\eqref{eq:y-l-bad-spike-lb} implies that with probability at least $1-\mu_0$,
\begin{equation}
\label{eq:g-bad-spike-lb}
g\ge (1-216\eps)\parens{\frac 12-(\cnst+4)\eps}^2\abs{y_\ell}^2 \ge (1-216\epsilon)\parens{\frac 14-2(\cnst+4)\eps}\abs{y_\ell}^2 \ge \parens{\frac 14-2(\cnst+31)\eps}\abs{y_\ell}^2 \ge \frac{\abs{y_\ell}^2}{10},
\end{equation}
as long as
\[T_1\ge \Theta\parens{\frac{N\tau^2\log(1/\mu_0)}{\eps^2y_{\ell}^2}}\ge \Theta\parens{\frac{N U^2\log(1/\mu_0)}{\eps^2}},\]
which is how we picked $T_1$.

This implies that the check in $\verify$ returns $\false$ and hence $h$ will not get added to $\lst$.

\section{One-Sparse Recovery Algorithm for Jacobi Polynomials}
\label{sec:jacobi-1sps}

In this section we will present a one-sparse recovery algorithm for Jacobi polynomials~\cite{szego}. 
In particular, we consider a fixed family of Jacobi polynomials with parameters $\alphajac,\betajac > -1$ 
and the corresponding Jacobi transform $\vF$ through Definition~\ref{def:OP}.
We will prove:

\begin{thm}
\label{thm:jacobi-1-sps}
There exists a universal constant $C$ such that the following holds. Fix any $\alphajac,\betajac>-1$.
There exists an $(N,\eps,C\cdot\eps,\mu)$ $1$-sparse recovery algorithm for the Jacobi transform (with parameters $\alphajac,\betajac$) that makes
  \[
    O\left( \frac{\log N}{\epsilon^{5/2}}\log\log N \log\parens{\frac{1}{\epsilon}}\log \parens{\frac{1}{\epsilon \mu}} \right)
  \]
queries and takes time
  \[
    O\left( \frac{\log N}{\epsilon^{3}}\log\log N \log\parens{\frac{1}{\epsilon}}\log \parens{\frac{1}{\epsilon \mu}} \right).
  \]
In the above the Big-Oh notation hides constants that depend on $\alphajac$ and $\betajac$.
\end{thm}

Recalling Definition~\ref{def:1-sparse}, in other words we seek a solution to the following specific problem:
\begin{itemize}
\item \textbf{Unknown:} Index $\ell \in [N]$ and magnitude $v$. 
\item \textbf{Query access:} To a vector
\begin{equation}
  \label{eq:y}
  \vy=\vF^{-1}(v \cdot \ve_\ell + \vhw)=v\cdot \vF^T\ve_\ell+\vw
\end{equation}
for some $\|\vw\|_2 \le \epsilon |v|$. 
\item \textbf{Output:} With probability $1-\mu$, return $\ell$ and some $\widetilde{v}$ such that $|\tilde{v}-v| \le C\epsilon|v|$. 
\end{itemize}

To do so, we first provide the required  background for Jacobi polynomials, including important bounds from~\cite{szego} (some of which are deferred to Appendix~\ref{app:jacobi}). Then, we give an overview of the entire algorithm with a comprehensive summary of the notation used. There are two basic building blocks we need: (i) $\cos(\cdot)$ evaluation from Jacobi polynomial evaluations and (ii) an approximate $\arccos$ algorithm from the noisy $\cos(\cdot)$ evaluation (whose proof we cover in Section~\ref{sec:chebyshev-1sps}). We present the basic building blocks first and then delve into the details of the entire algorithm. 

\subsection{Background: Jacobi polynomials}

Jacobi polynomials are indexed by two parameters $\alphajac,\betajac>-1$ and are orthogonal with respect to the measure
\[w^{(\alphajac,\betajac)}(X)=(1-X)^{\alphajac}\cdot(1+X)^{\betajac}\]
in the range $[-1,1]$.

Traditionally, Jacobi polynomials $P_j^{(\alphajac,\betajac)}$ are defined by the following recurrence relation~\cite[Sec. 4.5, (4.5.1)]{szego} for $j\ge 2$:\footnote{$P_0^{(\alphajac,\betajac)}(X)=1$ and $P_1^{(\alphajac,\betajac)}(X)=\frac{\alphajac+\betajac+2}2\cdot X+\frac{\alphajac-\betajac}2$.}
\begin{align}
2j(j+\alphajac+\betajac)(2j+\alphajac+\betajac-2)\cdot P_j^{(\alphajac,\betajac)}(X)=&(2j+\alphajac+\betajac-1)\set{(2j+\alphajac+\betajac)(2j+\alphajac+\betajac-2)X+\alphajac^2-\betajac^2}\cdot P_{j-1}^{(\alphajac,\betajac)}(X) \notag\\
\label{eq:jac-recur}
&-2(j+\alphajac-1)(j+\betajac-1)(2j+\alphajac+\betajac)P_{j-2}^{(\alphajac,\betajac)}.
\end{align}
However, the polynomials $P_j^{(\alphajac,\betajac)}(X)$ while being orthogonal w.r.t. the weight $w^{(\alphajac,\betajac)}(X)$, are not orthonormal w.r.t. it since for any $j$, we have (~\cite[Pg. 68, (4.3.3)]{szego}):
\[\int_{-1}^{1} \parens{P_j^{(\alphajac,\betajac)}(X)}^2\cdot w^{(\alphajac,\betajac)}(X)dX=h_j^{\alphajac,\betajac},\]
where
\[h_j^{\alphajac,\betajac}=\frac{2^{\alphajac+\betajac+1}}{2j+\alphajac+\betajac+1}\cdot \frac{\Gamma(j+\alphajac+1)\Gamma(j+\betajac+1)}{\Gamma(j+1)\Gamma(j+\alphajac+\betajac+1)},\]
with $\Gamma(z)$ being the Gamma function. 

To make the rest of the exposition simpler, we will define the Jacobi polynomial of degree $j$ with parameters $\alphajac,\betajac$ as
\begin{equation}
\label{eq:jac-orthonormal}
\jac{\alphajac}{\betajac}{j}(X)=\frac{1}{\sqrt{h_j^{\alphajac,\betajac}}}\cdot P_j^{(\alphajac,\betajac)}(X).
\end{equation}
Note that these polynomials are indeed orthonormal with respect to the measure $w^{(\alphajac,\betajac)}(X)$.

Chebyshev polynomials are special case of $\alphajac=\betajac=-\frac 12$ and Legendre polynomials are the special case of $\alphajac=\betajac=0$ (up to potentially a multiplicative factor that could depend on $j$).

We start with the following property of roots of Jacobi polynomials:
\begin{thm}[\cite{szego}, Theorem 8.9.1]
\label{thm:jacobi-roots}
Let
\[\alphajac,\betajac> -1.\]
Then for any integer $N$, let the roots of $\jac{\alphajac}{\betajac}{N}(X)$ be given by $\cos{\theta_{\ell}}$ for $\ell\in [N]$. Then, there is a universal constant $C^{\alpha,\beta}$ such that
\[\frac{\ell-C^{\alpha,\beta}}{N}\cdot \pi\le \theta_{\ell}\le \frac{\ell+C^{\alpha,\beta}}{N}\cdot \pi.\]
\end{thm}

Theorem~\ref{thm:jacobi-roots} establishes the denseness (Definition~\ref{def:roots-spread}) of Jacobi families, made explicit in Corollary~\ref{cor:jacobi-dense}.
Additionally, we will make use of the following immediate corollary:
\begin{cor}
\label{cor:jac-roots}
Let $[a,b]\subseteq [0,\pi]$ be an interval. Then the number of roots of $\jac{\alphajac}{\betajac}{N}(X)$ for the form $\cos\parens{\theta_{\ell}}$ such that $\theta_{\ell}\in [a,b]$ is upper bounded by $O\parens{\frac{(b-a)N}\pi}$. Further, the list of such roots can be computed in time $O\parens{(b-a)N}$ assuming that there is a data structure that for any $i\in [N]$, stores all the $O(C^{\alpha,\beta})$ roots in the interval $\left[ \frac{i\pi}N,\frac{(i+1)\pi}N\right)$.
\end{cor}

Further, we will also need the fact that all the $\theta_\ell$ values are bounded away from $0$ and $\pi$:
\begin{lmm}
\label{lem:jacobi-root-bounded}
Let $\alpha,\beta> - 1$. Then there exists a constant $C'>0$ such that the following holds.
Let $N\ge 1$ be large enough. Let $\evalpts_0\le \cdots\le\evalpts_{N-1}$ be the roots of the $N$th Jacobi polynomial. Then for every $0\le \ell<N$,
\[\frac {C'}N \le \theta_\ell \le \pi-\frac{C'}N.\]
\end{lmm}
\begin{proof}
Theorem 8.1.2 in~\cite{szego} states that for large enough $N$, we have
\[N\cdot \theta_0=j^{\alpha}_0+o(1),\]
where $j^{\alpha}_0$ is the first positive root of $\bessel{\alpha}{X}$. It is known that (see e.g. equation (5.3) in~\cite{bessel-root-lb}):
\[j^{\alpha}_0>4(\alpha+1).\]
Further, by Lemma~\ref{lem:jacobi-neg-x}, we have that
\[\pi-\theta_{N-1}=j^{\beta}_0\cdot(1+o(1)).\]
Thus, the lemma follows if we pick
\[C'=4\parens{\min\set{\alpha,\beta}+1},\]
which is strictly positive since $\alpha,\beta>-1$.
The proof is complete.
\end{proof}

Next, we will need an asymptotic bound on values of Jacobi polynomials\footnote{The bounds are generally stated for $P_j^{(\alphajac,\betajac)}(X)$-- we have updated it for $\jac{\alphajac}{\betajac}{j}(X)$ by adding the required factor of $\frac{1}{\sqrt{h_j^{\alphajac,\betajac}}}$ to the approximation.}:
\begin{thm}[\cite{szego}, Theorem 8.21.13]
\label{thm:jacobi-approx}
Let $\alphajac,\betajac>-1$. 
There is a constant $c>0$ such that the following holds.  
For any integer $j\ge 1$, let
\begin{equation}
\label{eq:theta-bounds}
\frac cj \le \theta \le \pi-\frac cj.
\end{equation}
Then, there exists a constant $c'$ such that
\[\jac{\alphajac}{\betajac}{j}(\cos{\theta})=\frac{1}{\sqrt{h_j^{\alphajac,\betajac}}}\cdot \frac{1}{\sqrt{j}}\cdot \kap{\alpha}{\beta}{\theta}\cdot \parens{\cos\parens{n\theta+\gamma}\pm \frac{c'}{j\sin{\theta}}},\]
where
\[n = j+\addjac,\]
\[\kap{\alpha}{\beta}{\theta}= \frac{1}{\sqrt{\pi}\cdot \parens{\sin{\frac\theta 2}}^{\alphajac+\frac 12}\parens{\cos\frac\theta 2}^{\betajac+\frac 12}},\]
and
\[\gamma= \frac{-(\alphajac+\frac 12)\pi}{2}.\]
\end{thm}

\begin{rmk}
When $\alpha$ and $\beta$ are clear from context (which will be the case for the rest of this section), we will just use $\kappa(\theta)$ instead of $\kap{\alpha}{\beta}{\theta}$.
\end{rmk}

We will need the following bound on the normalization factor $h_j^{\alphajac,\betajac}$:
\begin{lmm}
\label{lem:norm-fact-bound}
For large enough $j$, we have
\[j\cdot h_j^{\alphajac,\betajac}=\Theta_{\abs{\alphajac},\abs{\betajac}}(1),\]
where $\Theta_{x,y}(\cdot)$ hides factors that depends on $x$ and $y$.
\end{lmm}
\begin{proof}
We will use Stirling's approximation for the Gamma function: for $z\ge 1$ we have
\[\Gamma(z)=\sqrt{\frac{2\pi}z}\cdot \parens{\frac{z}{e}}^z\parens{1+O\parens{\frac 1z}}.\]
By definition of $h_j^{\alphajac,\betajac}$, we have
\[j\cdot h_j^{\alphajac,\betajac}=\frac{j\cdot2^{\alphajac+\betajac+1}}{2j+\alphajac+\betajac+1}\cdot \frac{\Gamma(j+\alphajac+1)\Gamma(j+\betajac+1)}{\Gamma(j+1)\Gamma(j+\alphajac+\betajac+1)}.\]
Note that for large enough $j$, the first fraction above is $\Theta_{\alphajac,\betajac}(1)$ so to complete the proof we will argue that the second fraction is also $\Theta_{\abs{\alphajac},\abs{\betajac}}(1)$. Indeed, we will apply Stirling's approximation to the Gamma function (where for notational convenience we will ignore the constant factors in the Stirling's approximation) and get
\begin{align*}
  \frac{\Gamma(j+\alphajac+1)\Gamma(j+\betajac+1)}{\Gamma(j+1)\Gamma(j+\alphajac+\betajac+1)}
  &= \frac{\parens{\frac{\parens{j+\alphajac+1}^{j+\alphajac+\frac 12}}{e^{j+\alphajac+1}}}\cdot \parens{\frac{\parens{j+\betajac+1}^{j+\betajac+\frac 12}}{e^{j+\betajac+1}}}   }{\parens{\frac{\parens{j+1}^{j+\frac 12}}{e^{j+1}}}\cdot \parens{\frac{\parens{j+\alphajac+\betajac+1}^{j+\alphajac+\betajac+\frac 12}}{e^{j+\alphajac+\betajac+1}}}} \cdot \frac{\sqrt{j+1}\sqrt{j+\bar{\alpha}+\bar{\beta}+1}}{\sqrt{j+\bar{\alpha}+1}\sqrt{j+\bar{\beta}+1}}\cdot (1 + o(1))\\
&=\frac{j^{j+\alphajac+\frac 12+j+\betajac+\frac 12 - j-\frac 12-j-\alphajac-\betajac-\frac 12}}{e^{-j-1-j-\alphajac-\betajac-1+j+\alphajac+1+j+\betajac+1}}\cdot \frac{\parens{1+\frac{\alphajac+1}{j}}^{j+\alphajac+\frac 12}\cdot \parens{1+\frac{\betajac+1}{j}}^{j+\betajac+\frac 12}}{\parens{1+\frac{1}{j}}^{j+\frac 12}\cdot \parens{1+\frac{\alphajac+\betajac+1}{j}}^{j+\alphajac+\betajac+\frac 12}} \cdot (1+o(1)) \\
&=\frac{\parens{1+\frac{\alphajac+1}{j}}^{j+\alphajac+\frac 12}\cdot \parens{1+\frac{\betajac+1}{j}}^{j+\betajac+\frac 12}}{\parens{1+\frac{1}{j}}^{j+\frac 12}\cdot \parens{1+\frac{\alphajac+\betajac+1}{j}}^{j+\alphajac+\betajac+\frac 12}} \cdot (1+o(1)).
\end{align*}
The proof is complete by noting that each of the remaining terms is $e^{O(\abs{\alphajac}+\abs{\betajac})}$ for large enough $j$, which is $\Theta_{\abs{\alphajac},\abs{\betajac}}(1)$, as desired,
\end{proof}

Finally, we will need the following result in our analysis.
\begin{thm}
\label{cor:jac-U-bound}
Let $\alpha,\beta> - 1$.
Let $N\ge 1$ be large enough.
Let $\evalpts_0,\dots,\evalpts_{N-1}$ be the roots of the $N$th Jacobi polynomial. Then define $\vF$ such that for every $0<\ell,j<N$, we have
\begin{equation}
  \label{eq:F-entries}
\vF[\ell,j]=\jac{\alpha}{\beta}{j}(\evalpts_{\ell})\cdot\sqrt{w_{\ell}},
\end{equation}
where
\[w_{\ell}=\frac 1{\sum_{j=0}^{N-1} \jac{\alpha}{\beta}{j}(\evalpts_{\ell})^2}.\]
Then for every $0\le \ell<N$, we have
\begin{equation}
\label{eq:w-ell-ub}
\frac 1{w_\ell} = O\parens{N\cdot\kappa^2(\theta_\ell)}
\end{equation}
and
\begin{equation}
\label{eq:jacobi-U}
\max_{0<\ell,j<N} \abs{\vF[\ell,j]}=\frac{O_{\abs{\alpha},\abs{\beta}}(1)}{\sqrt{N}}.
\end{equation}
\end{thm}
It turns out that the above result does not quite follow from Theorem~\ref{thm:jacobi-approx} and it needs some extra work to prove~\eqref{eq:jacobi-U} for the entries where equation~\eqref{eq:theta-bounds} is not satisfied. In particular, we need a more general version of Theorem~\ref{thm:jacobi-approx} (which involves the Bessel function). The proof of Theorem~\ref{cor:jac-U-bound} is deferred to Appendix~\ref{app:jac-U-bound} (with the more general version of Theorem~\ref{thm:jacobi-approx} proven in Appendix~\ref{app:jacobi-more}).

\subsection{Overview of the algorithm}

The main idea will be to use Theorem~\ref{thm:theta-from-cos} to identify $\theta_\ell$ accurately enough (Section~\ref{sec:estimate-ell}), which identifies $\ell$.
Given $\ell$, $v$ can be estimated with a simple sampling argument (Algorithm~\ref{alg:check}).

In order to estimate $\theta$ accurately enough, Theorem~\ref{thm:theta-from-cos} will require querying $\cos(w\theta_\ell)$ for $w$ up to $\Omega(N)$, say $w \le \nu N$ for some (universal, to be chosen later) constant $\nu$.

These queries in turn will be provided by Corollary~\ref{cor:jacobi-to-cos}, which computes $\cos(w\theta)$ from queries $\jac{\alphajac}{\betajac}{j}(\cos \theta)$ for $j = \Delta-w,\Delta,\Delta+w$ for certain $\Delta$ that satisfy a technical condition (equation~\eqref{eq:cos-denom-large}).

\begin{itemize}
\item In order for $\jac{\alphajac}{\betajac}{j}(\cdot)$ to be meaningful (related to $\cos$ values) via Theorem~\ref{thm:jacobi-approx}, either $j$ must be sufficiently large enough or alternatively $\theta_\ell$ cannot be too close to the boundaries of $[0,\pi]$ (equation~\eqref{eq:theta-bounds}).
By restricting the queries to $j \ge \nu N$ (enforced by only querying $\Delta \in [2\nu N, N - 2\nu N]$), there will be a constant number of possible bad values of $\ell$.

\item Because we can't determine if \eqref{eq:cos-denom-large} is true a priori, we will instead query for many random values of $\Delta$ in the allowed range.
For most values of $\ell$, there are enough good values of $\Delta$ (Definition~\ref{def:no-spread}).
There will be a small number of bad values of $\ell$, located in Lemma~\ref{lmm:not-spread}.
\end{itemize}

Finally, Algorithm~\ref{alg:check} (with bounds given in Corollary~\ref{cor:jacobi-prune}) allows us to check if $\ell$ is any specific value.
Thus the final algorithm first checks if $\ell$ is any of the possible bad values; if not, it narrows down $\ell$ using Theorem~\ref{thm:theta-from-cos} and Corollary~\ref{cor:jacobi-to-cos}.
This gives a final window for $\theta_\ell$ of size $O(1/\nu N)$, with $O(1/\nu)$ possible values of $\ell$ that can be checked with Algorithm~\ref{alg:check}.

\paragraph{Summary of notation}
In Table~\ref{table:jacobi-notation} we summarize notation for frequently used constants and parameters in this section, including a brief description and where they are defined and used.

\begin{table}[h]
  \caption{Summary of notation used in the 1-sparse Jacobi solver.}
  \centering
\begin{tabularx}{\linewidth}{@{}lX@{}}
\toprule
\multicolumn{2}{l}{\textbf{Fixed or latent parameters}} \\
  $\alphajac, \betajac$ & parameters of the fixed Jacobi family we are considering \\
  $\vF$ & Fixed orthogonal polynomial evaluation matrix (defined in Theorem~\ref{cor:jac-U-bound}). \\
  $v, \ell, \theta_\ell$ & unknowns to be determined \\
  $\vw, \vy$ & algorithm is given query access to $\vy = v \cdot \vF^T\ve_\ell + \vw$ \\
  $\epsilon$ & fixed noise level so that $\norm{2}{\vw} \le \epsilon |v|$ \\
\midrule
\multicolumn{2}{l}{\textbf{Universal constants and parameters for the fixed Jacobi family}} \\
  $C^{\alphajac,\betajac}$ & constant in the Jacobi polynomial root distribution (Theorem~\ref{thm:jacobi-roots})\\
  $c$ & Provides condition for Theorem~\ref{thm:jacobi-approx} to hold (equation~\eqref{eq:theta-bounds}) \\
  $c'$ & constant in the Jacobi evaluation approximation Theorem~\ref{thm:jacobi-approx} \\
  $n, \kappa, \gamma$ & parameters in the Jacobi evaluation approximation Theorem~\ref{thm:jacobi-approx} \\
\midrule
\multicolumn{2}{l}{\textbf{Constants we will fix}} \\
  $\nu$ & A constant given to Theorem~\ref{thm:theta-from-cos} so that it yields a final range for $\theta_\ell$ of size $O(1/\nu N)$.
  Independently, we will only query $\jac{\alphajac}{\betajac}{j}$ for $j \ge \nu N$, so that Theorem~\ref{thm:jacobi-approx} holds for any $\theta_\ell$ except for a range of size $O(1/\nu N)$. 
  It is set in Lemma~\ref{lem:jac-to-cos-correct}. \\
  $\delta_0$ & A constant determining the noise $\epsilon_0$ that Theorem~\ref{thm:theta-from-cos} sees, via equation~\eqref{eq:jacobi-to-cos-approx}.
  It is set in Lemma~\ref{lem:jac-to-cos-correct} as a function of $\epsilon$. \\
\midrule
\multicolumn{2}{l}{\textbf{Indices and other variables}} \\
  $w$ & A ``blow-up'' factor indexing the values $\cos(w \theta_\ell)$ that will be queried during Theorem~\ref{thm:theta-from-cos}. The algorithm will query $w \le \nu N$ to narrow the range for $\theta_\ell$ sufficiently. \\
  $j$ & Indexes the Jacobi polynomials $\jac{\alphajac}{\betajac}{j}(x)$. The algorithm will query $\vy[j]$ for $j \in [\nu N, N-\nu N]$ so that Theorem~\ref{thm:jacobi-approx} applies. \\
  $\Delta$ & Randomly chosen value (in $[2\nu N, N-2\nu N]$) to deduce $\cos(w \theta)$ from $J_{\Delta-w}, J_{\Delta}, J_{\Delta+w}$ (Corollary~\ref{cor:jacobi-to-cos}) \\
  $R$ & Number of rounds to perform a given subroutine that has a constant probability of failure \\
  $\rho_\epsilon$ & A constant depending on $\epsilon$ (Definition~\ref{def:rho}) \\
  $\mu$ & Error probability parameter of subroutines, as well as of the entire 1-sparse recovery algorithm \\
  \bottomrule
\end{tabularx}
  \label{table:jacobi-notation}
\end{table}

\subsection{$\cos(\cdot)$ values from Jacobi polynomial evaluation}

We use Theorem~\ref{thm:jacobi-approx} to get the following lemma, which provides us with noisy estimates of $\cos(w \theta)$ from querying values of $\jac{\alphajac}{\betajac}{\square}(\cos \theta)$:
\begin{cor}
\label{cor:jacobi-to-cos}
Consider a fixed $w \in [N]$ and $\theta \in [0, \pi]$.
Then for any integer $\Delta > w$ such that $j=\Delta-w$ satisfies~\eqref{eq:theta-bounds},
\begin{equation}
\label{eq:cos-denom-large}
\abs{\cos\parens{\parens{\Delta+\addjac}\theta+\gamma}}\ge \delta_0,
\end{equation}
and we have access to values $\widetilde{{\jac{\alphajac}{\betajac}{j}}}(\cos{\theta})$ for $j\in\set{\Delta-w,\Delta,\Delta+w}$ such that
\begin{equation}
\label{eq:approx-jac-error-bound}
\widetilde{\jac{\alphajac}{\betajac}{j}}(\cos{\theta}) = \jac{\alphajac}{\betajac}{j}(\cos{\theta}) \pm \kappa(\theta)\eps,
\end{equation}
we have
\begin{equation}
\label{eq:jacobi-to-cos-approx}
\frac{\parens{\sqrt{h_{\Delta+w}^{\alphajac,\betajac}\cdot (\Delta+w)}}\cdot \widetilde{\jac{\alphajac}{\betajac}{\Delta+w}}(\cos{\theta}) + \parens{\sqrt{h_{\Delta-w}^{\alphajac,\betajac}(\Delta-w)}}\cdot \widetilde{\jac{\alphajac}{\betajac}{\Delta-w}}(\cos{\theta})  } { 2\sqrt{h_{\Delta}^{\alphajac,\betajac}\cdot \Delta}\cdot \widetilde{\jac{\alphajac}{\betajac}{\Delta}}(\cos{\theta})  }
=
\frac{ \cos(w\theta) \pm \frac{ O(1/\delta_0)}{(\Delta-w)\sin{\theta}}\pm O(\eps/\delta_0)   }{ 1\pm \frac{ O(1/\delta_0)}{\Delta\sin{\theta}} \pm O(\eps/\delta_0)  }.
\end{equation}
\end{cor}
\begin{proof}
This follows by applying Theorem~\ref{thm:jacobi-approx} and the following elementary identity in a straightforward way:
\begin{equation}
  \label{eq:cosA+B}
  \cos{A}\cos{B}=\frac 12\cdot \parens{\cos(A+B)+\cos(A-B)}.
\end{equation}

For notational convenience define
\[A = \parens{\Delta+\addjac}\theta+\gamma,\]
and
\[B= w\theta.\]
Now consider the following relations:
\begin{align}
&\frac{\parens{\sqrt{h_{\Delta+w}^{\alphajac,\betajac}\cdot (\Delta+w)}}\cdot \widetilde{\jac{\alphajac}{\betajac}{\Delta+w}}(\cos{\theta}) + \parens{\sqrt{h_{\Delta-w}^{\alphajac,\betajac}(\Delta-w)}}\cdot \widetilde{\jac{\alphajac}{\betajac}{\Delta-w}}(\cos{\theta})  } { 2\sqrt{h_{\Delta}^{\alphajac,\betajac}\cdot \Delta}\cdot \widetilde{\jac{\alphajac}{\betajac}{\Delta}}(\cos{\theta})  }\notag\\
\label{eq:jac-to-cos-s1}
=&\frac{\parens{\sqrt{h_{\Delta+w}^{\alphajac,\betajac}\cdot (\Delta+w)}}\cdot \jac{\alphajac}{\betajac}{\Delta+w}(\cos{\theta}) + \parens{\sqrt{h_{\Delta-w}^{\alphajac,\betajac}(\Delta-w)}}\cdot \jac{\alphajac}{\betajac}{\Delta-w}(\cos{\theta}) \pm O(\kappa(\theta)\epsilon) } { 2\sqrt{h_{\Delta}^{\alphajac,\betajac}\cdot \Delta}\cdot \jac{\alphajac}{\betajac}{\Delta}(\cos{\theta}) \pm O(\kappa(\theta)\epsilon) }\notag\\
&~~~~~~~~~~=
\frac{ \kappa(\theta)\parens{\cos(A+B)+\cos(A-B)\pm \frac{c'}{(\Delta-w)\sin{\theta}}\pm \frac{c'}{(\Delta+w)\sin{\theta}} \pm O(\epsilon) } }{2\kappa(\theta)\parens{\cos{A}\pm \frac{c'}{\Delta\sin{\theta}} \pm O(\epsilon) } }\\
\label{eq:jac-to-cos-s2}
&~~~~~~~~~~=
\frac{ \frac{\cos(A+B)+\cos(A-B)}{2\cos{A}} \pm \frac{O(1/\delta_0)}{(\Delta-w)\sin{\theta}} \pm O(\epsilon/\delta_0)}{ 1\pm \frac{O(1/\delta_0)}{\Delta\sin{\theta}} \pm O(\epsilon/\delta_0) }\\
\label{eq:jac-to-cos-s3}
&~~~~~~~~~~=
\frac{ \cos(w\theta) \pm \frac{ O(1/\delta_0)}{(\Delta-w)\sin{\theta}} \pm O(\epsilon/\delta_0)  }{ 1\pm \frac{ O(1/\delta_0)}{\Delta\sin{\theta}}  \pm O(\epsilon/\delta_0) },
\end{align}
as desired.
In the above the first equation follows from Lemma~\ref{lem:norm-fact-bound},~\eqref{eq:jac-to-cos-s1} follows from Theorem~\ref{thm:jacobi-approx},~\eqref{eq:jac-to-cos-s2} follows from~\eqref{eq:cos-denom-large}, and~\eqref{eq:jac-to-cos-s3} follows from~\eqref{eq:cosA+B}. (Note that $\cos{A}\ne 0$ by~\eqref{eq:cos-denom-large}.)

\end{proof}

\subsection{Computing $\theta$ from noisy $\cos(\cdot)$ evaluation}

We being with a definition:
\begin{defn}\label{def:rho}
  Let $\rho_{\epsilon_0}$ be such that for any $\theta \in [0,\pi]$,
  \[
    \arccos\left( \frac{\cos \theta \pm \epsilon_0}{1 \pm \epsilon_0} \right) = \theta \pm \rho_{\epsilon_0}.
  \]
\end{defn}

In Section~\ref{sec:chebyshev-1sps}, we will prove the following result:
\begin{thm}
\label{thm:theta-from-cos} There is an algorithm $\approxcos$ such that the following holds.
Let $\theta\in (0,\pi)$. Then for any integer $\tau\ge 1$ and small enough $\epsilon_0$ so that $0 < \rho_{\epsilon_0} < \pi/22$, given access to evaluations
\begin{equation}
\label{eq:theta-from-cos-query}
\frac{\cos(w\theta)\pm \eps_0}{1\pm \eps_0},
\end{equation}
where $w\in\set{x_t\cdot 2^{t-1}|\text{ for every } 1\le t\le \tau\text{ with }x_t\in[1,3/2]}$,  then $\approxcos(\tau,\eps_0)$ (for $\tau\le \floor{\log_2(2N/3)}$) returns a range $S_{\tau}\subseteq [0,\pi]$ such that
\begin{itemize}
\item $\theta\in S_{\tau}$; and
\item $\abs{S_{\tau}}\le \frac{2\rho_{\eps_0}}{2^{\tau}}$.
\end{itemize}
 Finally,  $\approxcos(\tau,\eps_0)$ runs in $O(\tau)$ time and makes $O(\tau)$ queries. 
\end{thm}

We would like to point out that we do not handle $\theta\in\set{0,\pi}$ in the result above. Our proof of the result above is a bit cleaner with this assumption and since we will only use the above result for $\theta=\theta_\ell$ where $\cos\theta_\ell$ is a root of the $N$th degree Jacobi polynomial this is not an issue (due to Lemma~\ref{lem:jacobi-root-bounded}).

\subsection{Details of the one-sparse recovery algorithm for Jacobi polynomials}

\subsubsection{A pruning step}

Before we present the final $1$-sparse recovery algorithm, we need one final result that allows us to check if $\ell$ lies in a given set of candidate values:

\begin{lmm}
\label{lem:jacobi-prune}
Let $0<\eps<\frac 1{200}$ be small enough.
There exists an algorithm that given a subset $S\subseteq [N]$ and query access to $\vy=v\cdot\vF^{-1}\ve_{\ell}+\vw$ (such that $\norm{2}{\vw}\le \eps\abs{v}$) does the following with probability at least $1-\mu$: if $\ell\in S$, it outputs $\ell$ and an estimate $\tilde{v}$ such that $\abs{\tilde{v}-v}\le 13\eps\abs{v}$; otherwise it outputs $\textsc{fail}$. Further it makes $\qprune{|S|,N,\mu,\eps}$ queries and takes  $\tprune{|S|,N,\mu,\eps}$ time, where 
\[\qprune{s,N,\mu,\eps}=O\parens{U^2N\frac{\log{N}}{\eps^2}\log\parens{\frac{\log{N}}{\eps^2}}\log\parens{\frac{s}{\mu}}},\]
and
\[\tprune{s,N,\mu,\eps}=O\parens{s\cdot U^2N\frac{\log{N}}{\eps^2}\log\parens{\frac{\log{N}}{\eps^2}}\log\parens{\frac{s}{\mu}}},\]
\end{lmm}
\begin{proof}
We will solve the problem for $|S|=1$ and then just repeat the algorithm for each $\ell'\in S$. So as long as we can solve the problem for $|S|=1$ with probability at least $1-\frac\mu {\abs{S}}$, we would be fine.

Thus, the problem we want to solve is that given a target $\ell'$, we have to decide if $\ell'=\ell$ or not (with query access to $\vy$). Algorithm~\ref{alg:check} has the details on how to solve this.
\begin{algorithm}
\caption{$\chk^{(\vy)}(\ell',\mu,\eps)$}\label{alg:check}
\begin{algorithmic}
\Require Query access to $\vy=\vF^{-1}v\cdot \ve_\ell+\vw$ such that $\|\vw\|_2\le \eps\abs{v}$, a guess $\ell'$, failure probability $\mu$
\Ensure $(b,\tilde{v})$ where $b\in\set{\true,\false}$ and $\tilde{v}\in\R$

\Statex
\State $s\gets \Theta\left(U^2N\cdot \frac{\log{N}}{\eps^2}\cdot \log\left(\frac{\log{N}}{\eps^2}\right)\right)$ \Comment{$s$ is chosen large enough so that~\eqref{eq:s-lb} is satisfied for $\delta=\eps$}
\State $R\gets \Theta\parens{\log\parens{\frac 1\mu}}$.
\For{$i=1\ldots R$}
        \State  Choose $\Gamma\subseteq [N]$ of size $s$ by sampling elements of $[N]$ uniformly at random with replacement.
	\State $u_i\gets \sqrt{\frac Ns\sum_{j\in \Gamma}\vy[j]^2}$ \Comment{This is an estimate for $\norm{2}{\vy}$}
	\State $\tilde{v}_i\gets \frac Ns\sum_{j\in \Gamma}\vy[j]\cdot \vF^T[j,\ell']$ \Comment{This is an estimate for $\inner{\vy,\vF[:,\ell']}$}
\EndFor
\State $u\gets\median(u_1,\ldots,u_R)$ and $\tilde{v}\gets \median\parens{\tilde{v}_1,\ldots,\tilde{v}_R}$.
\If{$\abs{\tilde{v}} \le \frac {\abs{u}}{10}$}
	\State\Return $(\false,\tilde{v})$
\Else
	\State\Return $(\true,\tilde{v})$
\EndIf
\end{algorithmic}
\end{algorithm}

Next we argue that Algorithm~\ref{alg:check} works as it is supposed to:
\begin{lmm}
\label{lem:check-correct}
Let $(b,\tilde{v})$ be the output of $\chk^{(\vy)}\parens{\ell',\mu,\eps}$. Then with probability at least $1-\mu$ the following is true for small enough $\eps$.
If $\ell=\ell'$ then $b=\true$ and $\abs{\tilde{v}-v}\le 13\eps\cdot\abs{v}$. Otherwise, $b=\false$. Further, it makes $O\parens{U^2N\cdot \frac{\log{N}}{\eps^2}\cdot \log\left(\frac{\log{N}}{\eps^2}\right)\log\parens{\frac 1\mu}}$ queries and runs in time linear in the number of queries.
\end{lmm}
\begin{proof}
Fix any $i\in [R]$ and let $\Gamma$ be the subset of indices chosen in the $i$th iteration. We will argue that each of the following two events hold with probability at least $\frac 45$:
\begin{enumerate}
\item Let $\vw_{\Gamma}$ be the vector $\vw$ projected to coordinates in $\Gamma$. Then
\begin{equation}
\label{eq:norm-w-not-large}
\sqrt{\frac Ns}\cdot \norm{2}{\vw_\Gamma}\le \sqrt{5}\eps\abs{v}.
\end{equation}
\item The following is true for any $h,h'\in [N]$ and $v\in\R$:
\begin{equation}
\label{eq:ip-rip}
\abs{\frac Ns\cdot v\cdot\inner{\vF^T[\Gamma,h],\vF^T[\Gamma,h']}-v\cdot \delta_{h,h'}}\le 10\eps \abs{v}.
\end{equation}
\end{enumerate}
Assuming the above are true, we first complete the proof. First let us consider $u_i$:
\begin{equation}
\label{eq:u-i-bnd}
u_i=\sqrt{\frac Ns}\norm{2}{\vy_\Gamma}\in \sqrt{\frac Ns}\cdot \parens{\norm{2}{v\cdot \vF^T[\Gamma,\ell]}\pm\norm{2}{\vw_\Gamma}}\subseteq \abs{v}\cdot \brackets{\sqrt{1\pm 10\eps}\pm \sqrt{5}\cdot\eps}\subseteq \abs{v}\cdot \brackets{1\pm 13\eps}.
\end{equation}
In the above the first containment follows from the definition of $\vy$ and the triangle inequality while the second containment follows from~\eqref{eq:ip-rip} with $h=h'=\ell$ and~\eqref{eq:norm-w-not-large}. Next we consider $\tilde{v}_i$:
\begin{align}
\tilde{v}_i &= \frac Ns\cdot v\inner{\vF^T[\Gamma,\ell'],\vF^T[\Gamma,\ell]}+\frac Ns\inner{\vF^T[\Gamma,\ell'],\vw_{\Gamma}}\notag\\
&\in \brackets{v\delta_{\ell,\ell'}\pm 10\eps\abs{v}\pm \frac Ns \norm{2}{\vF^T[\Gamma,\ell]}\cdot\norm{2}{\vw_{\Gamma}} }\notag\\
&\in \brackets{v\delta_{\ell,\ell'}\pm 10\eps\abs{v}\pm (1+10\eps)\sqrt{5}\cdot\eps\abs{v} }\notag\\
\label{eq:tilde-v-i-bnd}
&\in \brackets{v\delta_{\ell,\ell'}\pm 13\eps\abs{v}}.
\end{align}
In the above the equality follows from the definition of $\tilde{v}_i$ and $\vy$. The first containment follows from~\eqref{eq:ip-rip} (with $h=\ell$ and $h'=\ell'$) and Cauchy-Schwartz while the second containment follows from~\eqref{eq:norm-w-not-large} and~\eqref{eq:ip-rip} (with $v=1$ in the latter).

Thus, by the union bound we have that with probability at least $\frac 35$ both~\eqref{eq:u-i-bnd} and~\eqref{eq:tilde-v-i-bnd} hold. Then an application of Chernoff bound implies that with probability at least $1-\mu$, we have
\[u\in \abs{v}\cdot \brackets{1\pm 13\eps} \text{ and } \tilde{v}\in  \brackets{v\delta_{\ell,\ell'}\pm 13\eps\abs{v}}.\]
Note that if $\ell=\ell'$ then $\tilde{v}$ does have the required estimate. So we just need to argue that $\chk$ outputs $\true$ if $\ell=\ell'$ and $\false$ otherwise. To see this note that the above bounds on $u$ and $\tilde{v}$ imply that
\[\frac{\abs{\tilde{v}}}{\abs{u}}\in \frac {\abs{v}\parens{\delta_{\ell,\ell'}\pm 13\eps}}{\abs{v}\parens{1\pm 13\eps}}=\frac{\delta_{\ell,\ell'}\pm 13\eps}{1\pm 13\eps}.\]
Now when $\ell\ne\ell'$, then the ratio above is at most $\frac{13\eps}{1-13\eps} \le \frac 1{10}$ by our choice of $\eps\le \frac 1{200}$. On the other hand, when $\ell=\ell'$, the above ratio is at least $\frac{1-13\eps}{1+13\eps}>\frac 1{10}$ by our choice of $\eps$. Thus, $\chk$ returns $\true$ or $\false$ based on if $\ell=\ell'$ or not, as desired.

To complete the proof we argue~\eqref{eq:norm-w-not-large} and~\eqref{eq:ip-rip}. To begin with~\eqref{eq:norm-w-not-large}. Note that by definition we have
\[\avg_{\Gamma} \norm{2}{\vw_{\Gamma}}^2 =\frac sN\norm{2}{\vw}^2.\]
Then Markov's inequality implies that with probability at least $\frac 45$, 
\[\sqrt{\frac Ns}\cdot \norm{2}{\vw_{\Gamma}}\le \sqrt{5}\norm{2}{\vw}\le \sqrt{5}\eps\abs{v},\]
where the second inequality follows from the assumed upper bound on $\norm{2}{\vw}$. This proves~\eqref{eq:norm-w-not-large}. 

Finally, we tackle~\eqref{eq:ip-rip}: this basically follows by noting that $\vB$ (as defined in Corollary~\ref{cor:ip-preserve}) is exactly $\sqrt{\frac Ns}\vF^T[\Gamma,:]$ (with $\vA=\vF^T$). Finally, note that~\eqref{eq:ip-rip} is~\eqref{eq:ip-discrepancy} (scaled by a factor of $v$). This completes the proof.
\end{proof}
Finally, we use the $\chk$ algorithm repeatedly to get our final claimed result. In particular, for every $\ell'\in S$, we run $\chk^{(\vy)}\parens{\ell',\frac \mu{\abs{S}},\eps}$. If the call for $\ell'$ returns $\parens{\true,\tilde{v}}$, then we return $\ell'$ and $\tilde{v}$. Otherwise, we return $\textsc{fail}$. The correctness of this algorithm and the time complexity follows from Lemma~\ref{lem:check-correct} (and the union bound). Note that they query complexity does not have the $\abs{S}$ factor it-- this is because Corollary~\ref{cor:ip-preserve} holds for all pairs of indices $(\ell',\ell)$ simultaneously. Hence we can re-use the same $\Gamma$ across all $\ell'\in S$ and this proves the query complexity, as desired.
\end{proof}

In fact, we will use Corollary~\ref{cor:jac-roots} along with the above to obtain the following corollary, which will be useful in our final algorithm:
\begin{cor}
\label{cor:jacobi-prune}
There is an algorithm $\prunealg$ with the following property.
If $\prunealg$ has query access to $\vy=v\cdot\vF^{-1}\ve_{\ell}+\vw$ (such that $\norm{2}{\vw}\le \eps\abs{v}$),
and $\theta_\ell\in [a,b]$ where $[a,b]\subseteq [0,\pi]$,
then $\prunealg^{(\vv)}(a,b,\mu,\eps)$ outputs $\ell$ with probability at least $1-\mu$ with $\qprune{\frac{(b-a)N}{\pi},N,\mu,\eps}$ queries and $\tprune{\frac{(b-a)N}{\pi},N,\mu,\eps}$.
If $\theta_{\ell}\not\in [a,b]$, then with probability at least $1-\mu$, $\prunealg^{(\vv)}(a,b,\mu,\eps)$ outputs $\textsc{fail}$.
\end{cor}

\subsubsection{Bad potential values of $\ell$}
We will also use $\prunealg$ to `prune' out some other `bad' potential value of $\ell$, namely those for which Corollary~\ref{cor:jacobi-to-cos} is difficult to apply because not enough values of $\Delta$ satisfy~\eqref{eq:cos-denom-large}.
For this we define the notion of `bad' value:
\begin{defn}
\label{def:no-spread}
For any $0\le p\le 1$, $0\le \nu\le \frac 12$ and $0\le \delta\le 1$, we say an index $\ell\in [N]$ is $(p,\nu,\delta)$-{\em spread} if for a random $\Delta\in [\nu N, N-\nu N]$ it is the case that with probability at least $1-p$, we have:
\[\abs{\cos\parens{\parens{\Delta+\addjac}\theta_{\ell}-\parens{\frac\alphajac 2 +\frac 14}\pi}}\ge \delta.\]
\end{defn}

We will show that there are not that many indices that are not spread:
\begin{lmm}
  \label{lmm:not-spread}
The number of indices that are {\em not} $\parens{\frac{4\rho_{\delta}}{\pi(1-2\nu)},\nu,\delta}$-spread (for any $\nu< 1/2$) is $O\parens{\parens{\frac 1{\rho_{\delta}^2}}}$.
\end{lmm}
\begin{proof}
  In this proof we will let $x \mod{y}$ for reals $x,y$ denote the unique value $x'$ in $[0,y)$ such that $(x-x')/y \in \Z$. For example, $x\mod{1}$ denotes its fractional part $x-\floor{x}$. 

Let $0\le \tl <N$ be a real number such that
\[\theta_{\ell}=\frac{\tl\pi}{N}.\]
Recall we want to avoid the case for $\Delta\in [\nu N, N-\nu N]$ that we have
\[\abs{\cos\parens{\parens{\Delta+\addjac}\theta_{\ell}-\parens{\frac\alphajac 2 +\frac 14}\pi}}\le \delta.\]
By Lemma~\ref{lmm:cos-noise} it is enough to avoid the case that
\[\parens{\Delta+\addjac}\theta_{\ell}-\parens{\frac\alphajac 2 +\frac 14}\pi \in \frac\pi 2 + z\pi \pm \rho_{\delta}\]
for some $z \in \Z$, which is the same as
\[\Delta\frac{\tl\pi}{N} \in -\parens{\addjac}\theta_{\ell}+\parens{\frac\alphajac 2 +\frac 14}\pi + \frac\pi 2 + z\pi \pm \rho_{\delta}.\]
Note that $\ell$ is fixed while we vary $\Delta$, which means for a fixed $\ell$, we want to figure out the probability (for a random $\Delta\in[\nu N, N-\nu N]$) that $\Delta\frac{\tl\pi}{N} \mod{\pi}$ is in a range $[L,R]$ of size $2\rho_{\delta}$. Note here we have $L=-\parens{\addjac}\theta_{\ell}+\parens{\frac\alphajac 2 +\frac 14}\pi + \frac\pi 2 - \rho_{\delta}$ and $R=-\parens{\addjac}\theta_{\ell}+\parens{\frac\alphajac 2 +\frac 14}\pi + \frac\pi 2+ \rho_{\delta}$. However, for the rest of the argument the only thing we will use about $L$ and $R$ is that $R-L=2\rho_{\delta}$.

We will bound the largest probability that for a random $\Delta\in [\nu N, N-\nu N]$, the angle $\Delta\cdot \frac{\tl\pi}{N} \mod{\pi}$ falls into any range inside $[0,\pi]$ of size $2\rho_{\delta}$. Further since the $\cos$ value only changes sign when then angle is increased by an integer multiple of $\pi$, we need to upper bound the size of the set:
\[\set{ \Delta\in [\nu N, N-\nu N]~\middle\vert~\Delta\frac{\tl}{N}\mod{1}\in \brackets{\frac L\pi,\frac R\pi}}.\]
We bound the size of the set above by bounding:
\begin{equation}
\label{eq:set-size}
\abs{\set{ x\in\Z_N ~\middle\vert~ x\frac{\tl}{N}\mod{1}\in \brackets{\frac L\pi,\frac R\pi}}}.
\end{equation}
Or more precisely we want to figure out for how many roots $\cos\parens{\theta_{\ell}}$, the size of the set above is at most $2\cdot\frac{R-L}{\pi} \cdot N$. We solve this problem in Appendix~\ref{sec:mod-1}.
As per the terminology in there, we consider the sequence $\frac{\tilde{\ell}}{N}$ (where $\ell\in \Z_N$ indexes the roots) and it suffices to bound the number of $\parens{\frac{2\rho_{\delta}}{\pi}}$-bad elements (as defined in Definition~\ref{def:eps-bad}) in this sequence.
By Theorem~\ref{thm:jacobi-roots}, this sequence is $O(1)$-scattered (as per Definition~\ref{def:scatter}). Then Lemma~\ref{lem:bad-in-spread-sequence} implies that the number of $\tilde{\ell}$ such that the size in~\eqref{eq:set-size} is at least $\parens{\frac{4\rho_{\delta}}{\pi}}N$ is at most $O\parens{\frac 1{\rho_{\delta}^2}}$ (call such an element {\em bad}). Note that if $\tilde{\ell}$ is not bad, then for a random $\Delta\in [\nu N, N-\nu N]$, the probability that $\Delta\frac{\tl\pi}{N}\in [L,R]$ is at most $\frac{1}{1-2\nu}\cdot \frac{4\rho_{\delta}}{\pi}$. This completes the proof.
\end{proof}

The above along with Lemma~\ref{lem:jacobi-prune} (and Corollary~\ref{cor:bad-y-intervals}) implies the following result:
\begin{cor}
\label{cor:no-spread-prune}
There exists an algorithm $\prunespread(\delta,\nu,\mu,\eps)$  that given query access to $\vy=v\cdot\vF^{-1}\ve_{\ell}+\vw$ (such that $\norm{2}{\vw}\le \eps\abs{v}$),  with probability $1-O\parens{\frac\mu{\rho_{\delta}^2}}$ outputs $\ell$ if $\ell$ is not  $\parens{\frac{4\rho_{\delta}}{\pi(1-2\nu)},\nu,\delta}$-spread. Otherwise it outputs $\textsc{fail}$. In either case, it uses $O\parens{\frac 1{\rho_{\delta}^2}\cdot \qprune{O(1),N,\mu,\eps}}$ queries and $O\parens{\frac 1{\rho_{\delta}^2}\cdot \tprune{O(1),N,\mu,\eps}}$ time.
\end{cor}

\subsubsection{Computing $\ell$ and $v$}
\label{sec:estimate-ell}

The overall algorithm to compute $\ell$ and $v$ is presented in Algorithm~\ref{alg:jacobi}.
\begin{algorithm}
\caption{1-sparse Jacobi solver}\label{alg:jacobi}
\begin{algorithmic}
\Require Noisy query access to $\jac{\alphajac}{\betajac}{\square}(\lambda_{\ell})\cdot\sqrt{w_{\ell}}$, a noise parameter $\eps$, a desired error probability $\mu$
\Ensure $(\ell,v)$

\State $\ell\gets \prunealg\parens{0,\frac {C'}{\nu\sqrt{\eps}\delta_0 N},\frac \mu 6,\epsilon}$ \Comment{For constant $C', \nu$, and $\delta_0=\delta_0(\epsilon)$ to be determined.}
\If{$\ell\ne \textsc{fail}$}
    \State\Return{$\ell$}
\EndIf
\State \assert{1:} $\frac{ O(1/\delta_0)}{\nu N \sin{\theta_{\ell}}} \le \sqrt{\eps}$ and $\theta_\ell \ge \frac{c}{\nu N}$  
\Statex

\State $\ell\gets \prunealg\parens{\pi-\frac{c}{\nu N},\pi,\frac\mu 6,\epsilon}$
\If{$\ell\ne \textsc{fail}$}
    \State\Return{$\ell$}
\EndIf
\State \assert{2:} $\frac c{\nu N} \le \theta_{\ell} \le \pi-\frac{c}{\nu N}$ \Comment{$\theta_{\ell}$ now satisfies~\eqref{eq:theta-bounds} for $j\ge \nu N$.}
\Statex

\State $\ell\gets \prunespread\parens{\delta_0,2\nu,O\parens{\rho_{\delta_0}^2\cdot \mu},\eps}$
\If{$\ell\ne \textsc{fail}$}
    \State\Return{$\ell$}
\EndIf
\State \assert{3:} $\ell$ is $\parens{\frac{4\rho_{\delta_0}}{\pi(1-2\nu)},2\nu,\delta_0}$-spread. 
\Statex

\State $[a,b]\gets \approxcos\parens{\log(\nu N)-1,\sqrt{\eps}}$ \Comment{Whenever the algorithm needs access to $\cos(j\theta_{\ell})$ call $\jacforcos\parens{j,2\nu N,O\parens{\log\log{N}+\log\parens{\frac 1\mu}}}$}
\State \assert{5:} $\ell\in [a,b]$ and $|b-a|\le O\parens{\frac{\sqrt[4]{\eps}}{\nu N}}$.
\Statex

\State $\ell \gets \prunealg\parens{a,b,\frac\mu 6,\epsilon}$
\State \assert{6:} $\ell \neq \textsc{fail}$ is not returned.
\State $(\_, v) \gets \textsc{Check}(\ell, \frac{\mu}{6}, \epsilon)$
\State \Return{$(\ell,v)$}

\end{algorithmic}
\end{algorithm}

Algorithm~\ref{alg:jacobi} needs access to Algorithm~\ref{alg:jac-to-cos}, which is based on Corollary~\ref{cor:jacobi-to-cos}.
\begin{algorithm}
\caption{$\jacforcos(w,N',R)$}\label{alg:jac-to-cos}
\begin{algorithmic}
\Require Noisy query access to $\jac{\alphajac}{\betajac}{\square}(\lambda_{\ell})$, parameters $w\le N'$ and $R$ 
\Ensure An estimate $\frac{\cos(w\theta_{\ell})\pm O(\sqrt{\eps})}{1\pm O(\sqrt{\eps})}$.

\For{$i=1\dots R$}
    \State Pick random $\Delta\in [N', N-N']$
    \State $D \gets \text{sign}(\vy[\Delta]) \cdot \max(|\vy[\Delta]|, \frac{\epsilon}{N^{3/2}})$
    \Comment Ensure that we do not divide by something too small
    \State $v_i\gets \frac{\sqrt{h_{\Delta-w}(\Delta-w)}\vy[\Delta-w] + \sqrt{h_{\Delta-w}(\Delta+w)}\vy[\Delta+w]}{2\sqrt{h_{\Delta}(\Delta)}D}$
\EndFor
\State\Return{Median of $v_1,\dots, v_R$}

\end{algorithmic}
\end{algorithm}

\paragraph{Proof of correctness of  Algorithm~\ref{alg:jacobi}.}

We first argue that \assert{1-6} in Algorithm~\ref{alg:jacobi} and~\ref{alg:jac-to-cos} hold with high (enough) probability.

\begin{lmm}
\label{lem:1-jacobi-assert-1}
With probability at least $1-\frac\mu 6$, \assert{1} holds.
\end{lmm}
\begin{proof}
We need to have the following
\[\frac{ O(1/\delta_0)}{\nu N \sin{\theta_{\ell}}} \le \sqrt{\eps}.\]
        In other words we need (for large enough $N$),
        \begin{equation}
        \label{eq:alpha-l-lb-1}
        \theta_{\ell} \ge \Omega\parens{\frac{1}{\nu\sqrt{\eps}\delta_0N}},
        \end{equation}
where in the above we have used that fact that $\sin(x)\approx x$ for small enough $x$.
But the above is handled by the first call to $\prunealg$ by making sure $C'$ is large enough. \assert{1} fails only if the call to $\prunealg$ fails, which happens with probability at most $\frac\mu 6$ (as per Corollary~\ref{cor:jacobi-prune}).

In order to further ensure that $\theta_\ell \ge \frac{c}{\nu N}$, it suffices to also require that $C' \ge c$ (and hence $C' \ge c \sqrt{\epsilon} \delta_0$).
\end{proof}

\begin{lmm}
\label{lem:1-jacobi-assert-2}
Conditioned on \assert{1} being true, with probability at least $1-\frac\mu 6$, \assert{2} holds.
\end{lmm}
\begin{proof}
By the second call to $\prunealg$ we have that $\theta_{\ell}\le \pi-\frac c{\nu N}$ with probability $1-\frac \mu 6$ (as per Corollary~\ref{cor:jacobi-prune}).
\end{proof}

The following follows directly from Corollary~\ref{cor:no-spread-prune}:
\begin{lmm}
\label{lem:1-jacobi-assert-3}
With probability at least $1-\frac\mu 6$, \assert{3} holds.
\end{lmm}

Before arguing about \assert{5}, we first argue about correctness of Algorithm~\ref{alg:jac-to-cos}.
\begin{lmm}
\label{lem:jac-to-cos-correct}
Assume \assert{1-3} are true and
let $w\le \nu N$. Finally let $\delta_0=\frac{\sqrt{\eps}}{5000}$ and $\nu =\frac 1 8$. Then a call to $\jacforcos(w,2\nu N,R)$ returns $\frac{\cos(w\theta_{\ell})\pm O(\sqrt{\eps})}{1\pm O(\sqrt{\eps})}$ with probability $\exp\parens{-\Omega(R)}$. Further this call makes $3R$ queries and has time complexity $O(R)$.
\end{lmm}
\begin{proof} The claim on the query complexity follows from the statement of Algorithm~\ref{alg:jac-to-cos}. The rest of argument follows from showing that one can apply Corollary~\ref{cor:jacobi-to-cos}. 
Next, we show that for any $i\in [R]$, $v_i$ has the correct estimate with probability at least $\frac 23$. Then an application of the Chernoff bound would prove the claimed result.

Fix any round $i\in [R]$. We would like to show that~\eqref{eq:jacobi-to-cos-approx} holds with high probability.
Note that the value of $v_i$ in Algorithm~\ref{alg:jac-to-cos} does not change if multiply the numerator and denominator by $\frac 1{v\sqrt{w_\ell}}$. I.e. WLOG we will assume that we have query access to $\frac {\vy[w]}{v\sqrt{w_\ell}}$ and the $\frac \eps{N^{3/2}}$ term in definition of $D$ would be replaced by $\frac \eps{v\sqrt{w_\ell}N^{3/2}}$.
It suffices to show that all the conditions are met, where the queries $\widetilde{{\jac{\alphajac}{\betajac}{w}}}(\cos{\theta_{\ell}})$ in equation~\eqref{eq:approx-jac-error-bound} will be provided by $\vy[w]/(v\sqrt{w_\ell})$ (due to the normalization~\eqref{eq:F-entries}.) In particular, we will show that the following conditions are satisfied:

\begin{itemize}
  \item[\eqref{eq:theta-bounds}:] This follows from \assert{2} and the fact that all of $\Delta-w,\Delta,\Delta+w$ are at least $\nu N$.
  This is achieved by making sure $w\le \nu N$ (which is guaranteed by the call to $\approxcos$ having $\log(\nu N)-1$ as its first parameter) and making sure $\Delta\in [2\nu M, N-2\nu N]$ (guaranteed by the call $\jacforcos(w,2\nu N,R)$).

  \item[\eqref{eq:cos-denom-large}:] \assert{3} implies that this condition is met with probability at least 
  \[1 - \frac{4\rho_{\delta_0}}{\pi(1-2\nu)} \ge 1-\frac{16\sqrt{5\delta_0}}{\pi} \ge \frac {5}{6}\] 
  for random $\Delta$. (In the above, the first inequality follows from Lemma~\ref{lmm:cos-noise} and our choice of $\nu$ while the last inequality follows from our choice of $\delta_0$ and the fact that $\eps\le 1$.)

  \item[\eqref{eq:approx-jac-error-bound}:]
  We analyze the error $\widetilde{{\jac{\alphajac}{\betajac}{w}}}(\cos{\theta_{\ell}}) - \jac{\alphajac}{\betajac}{w}(\cos{\theta_\ell}) = \vy[w]/(v\sqrt{w_\ell}) - \jac{\alphajac}{\betajac}{w}(\cos{\theta_\ell}) = \vw[w]/(v\sqrt{w_\ell})$ (by equations~\eqref{eq:y} and \eqref{eq:F-entries}).
  The expected squared error $\vw[\Delta]^2$ for a random position $\Delta$ is at most $\frac{\epsilon^2|v|^2}{N(1-4\nu)}$.
  Thus, by Markov with probability at least $\frac {17}{18}$, the error $\vw[\Delta]\le \frac{\sqrt{18}\eps\abs{v}}{\sqrt{N(1-4\nu)}}\le \frac{6\eps\abs{v}}{\sqrt{N}}$.
  The error $\vw[\Delta]/(v\sqrt{w_\ell})$ is therefore $\frac{6\eps}{\sqrt{Nw_{\ell}}}$, which by Theorem~\ref{cor:jac-U-bound} (more specifically~\eqref{eq:w-ell-ub}) 
is $O\parens{\kappa(\theta_{\ell})\cdot \eps}$. 
  Similarly, for a randomly chosen $\Delta$ (and fixed $w$), the errors $\vw[\Delta-w]$ and $\vw[\Delta+w]$ satisfy the same bounds.
  By the union bound,~\eqref{eq:approx-jac-error-bound} is satisfied for all three with probability at least $\frac 5 6$. 
  Finally, we note that modifying the denominator changes $\vy[\Delta]/(v\sqrt{w_\ell})$ by by at most another $\frac{\eps}{|v|\cdot N\cdot \sqrt{Nw_{\ell}}}\le \frac \eps{\sqrt{Nw_\ell}} = O(\kappa(\theta_\ell)\cdot\epsilon)$ (where in the inequality we used our assumption that $\abs{v}\ge \frac 1N$-- recall the discussion on precision in Section~\ref{sec:sr-def}).
\end{itemize}
Thus, we can apply~\eqref{eq:jacobi-to-cos-approx} to conclude that with probability at least $\frac 23$:
\[v_i= \frac{ \cos(w\theta_{\ell}) \pm \frac{ O(1/\delta_0)}{(\Delta-w)\sin{\theta_{\ell}}}\pm O(\eps/\delta_0)   }{ 1\pm \frac{ O(1/\delta_0)}{\Delta\sin{\theta_{\ell}}} \pm O(\eps/\delta_0)}.\]
This proof is complete by noting that all the error terms are $O(\sqrt{\eps})$ due to \assert{1} (and the fact that $\Delta - w \ge \nu N$ by construction) and our choice of $\delta_0$.
\end{proof}

We now argue that
\begin{lmm}
\label{lem:1-jacobi-assert-5}
Conditioned on \assert{1-3} being true,
with probability at least $1-\frac\mu 6$, \assert{5} holds.
\end{lmm}
\begin{proof}
There are $O(\log{N})$ calls to  $\jacforcos(w,2\nu N,R)$ for $R=O(\log\log{N} + \log(1/\mu))$. Thus, by Lemma~\ref{lem:jac-to-cos-correct} we know that each call gives back the desired output with probability at least $1-O(\mu/\log{N})$.
Thus, by the union bound, all the calls work as intended with probability at least $1-\frac\mu 6$. Then \assert{5} follows from Theorem~\ref{thm:theta-from-cos} (and noting that by Lemma~\ref{lem:jacobi-root-bounded} we have $\theta_\ell\in (0,\pi)$ as needed by Theorem~\ref{thm:theta-from-cos}).
\end{proof}

\begin{lmm}
\label{lem:1-jacobi-assert-6}
Conditioned on \assert{5} being true,
with probability at least $1-\frac\mu 6$, \assert{6} holds.
\end{lmm}
\begin{proof}
    This follows from Lemma~\ref{lem:jac-to-cos-correct} and Corollary~\ref{cor:jacobi-prune}. 
\end{proof}

The following is a direct consequence of Lemma~\ref{lem:check-correct}:
\begin{lmm}
  \label{lem:1-jacobi-v}
  With probability at least $1-\frac{\mu}{6}$, the value $\tilde{v}$ returned by Algorithm~\ref{alg:jacobi} satisfies $\abs{\tilde{v}-v} \le 13\epsilon|v|$.
\end{lmm}

Lemma~\ref{lem:1-jacobi-assert-1},~\ref{lem:1-jacobi-assert-2},~\ref{lem:1-jacobi-assert-3},~\ref{lem:1-jacobi-assert-5},~\ref{lem:1-jacobi-assert-6}, and~\ref{lem:1-jacobi-v} implies the following result:
\begin{lmm}
\label{lem:jacobi-id-correct}
With probability at least $1-\mu$, Algorithm~\ref{alg:jacobi} correctly outputs $\ell$ and a $\tilde{v}$ such that $\abs{\tilde{v}-v} \le 13\epsilon|v|$.
\end{lmm}
Note that the above establishes the correctness of Algorithm~\ref{alg:jacobi}.

\paragraph{Time and Query Complexity.}

\begin{lmm}
\label{lem:jacobi-id-coomplexity}
Algorithm~\ref{alg:jacobi} makes
  \[
    O_{\abs{\alphajac},\abs{\betajac}}(1)
    \cdot O\left( \frac{\log N}{\epsilon^{5/2}}\log\log N \log\parens{\frac 1\eps}\log \parens{\frac{1}{\epsilon \mu}} \right)
  \]
queries and takes time
  \[
    O_{\abs{\alphajac},\abs{\betajac}}(1)
    \cdot O\left( \frac{\log N}{\epsilon^3}\log\log N \log\parens{\frac 1\eps}\log \parens{\frac{1}{\epsilon \mu}} \right).
  \]
\end{lmm}
\begin{proof}
We analyze the various function calls in  Algorithm~\ref{alg:jacobi}, which dominate the runtime.
There are three calls to $\prunealg$ each with query complexity $\qprune{\parens{s,N,\frac \mu 6,\eps}}$ and time complexity $\tprune{\parens{s,N,\frac \mu 6,\eps}}$.
The first call has $s=O\parens{\frac{1}{\epsilon}}$ (using the fact that $c,C',\nu$ are constants and our setting of $\delta_0=\Theta\parens{\sqrt{\epsilon}}$), the second has $s=O(1)$, and the third has $s=O(\epsilon^{1/4})$.
These are dominated by the first call. 
The single call to $\prunespread$, via Corollary~\ref{cor:no-spread-prune} has query and time complexities of $O\parens{\frac{1}{\sqrt{\eps}}\cdot \qprune{O(1),N,O(\sqrt{\epsilon}\mu),\eps}}$ and  $O\parens{\frac{1}{\sqrt{\eps}}\cdot \tprune{O(1),N,O(\sqrt{\epsilon}\mu),\eps}}$ respectively.
Finally $\approxcos$ makes $O(\log{N})$ calls to $\jacforcos$ with $R=O\parens{\log\log{N}+\log\parens{\frac 1\mu}}$.
The rest of $\approxcos$ takes $O(\log{N})$ time.\footnote{The number of calls and time complexity follow from Theorem~\ref{thm:theta-from-cos}.}
Finally, by Lemma~\ref{lem:jac-to-cos-correct}, we have that each call to  $\jacforcos$ takes time and has query complexity $O(R)$. Summing everything up, we get query complexity of
  \[\qprune{\parens{O\parens{\frac 1\eps},N,\mu/6,\eps}} + \frac 1{\sqrt{\eps}}\cdot \qprune{O(1),N,O(\sqrt{\epsilon}\mu),\eps} + O\parens{\log{N}\parens{\log\log{N}+\log\parens{\frac 1\mu}}}\]
  (and time complexity with equal parameters but to $\tprune{}$ instead of $\qprune{}$).

  Plugging the value of $\qprune{}$ from Lemma~\ref{lem:jacobi-prune} yields a query complexity of
  \begin{align*}
   & O\left(\frac 1{\sqrt{\eps}}\cdot U^2N \frac{\log N}{\epsilon^2}\log\left( \frac{\log N}{\epsilon^2} \right) \log \frac{1}{\epsilon \mu} \right)
    + O\parens{\log{N}\parens{\log\log{N}+\log\parens{\frac 1\mu}}}  \\
  &\ \   = O\parens{ U^2N\frac{\log N}{\eps^{5/2}}\log\log{N}\log\parens{\frac{1}{\epsilon}}\log\parens{\frac 1{\eps\mu}}}
  \end{align*}
  and time complexity
  \begin{align*}
  &  O\left( U^2N \frac{1}{\epsilon} \frac{\log N}{\epsilon^2}\log\left( \frac{\log N}{\epsilon^2} \right) \log \frac{1}{\epsilon \mu} \right)
    + O\parens{\log{N}\parens{\log\log{N}+\log\parens{\frac 1\mu}}} \\
  &\ \   = O\parens{ U^2N\frac{\log N}{\eps^3}\log\log{N}\log\parens{\frac{1}{\epsilon}}\log\parens{\frac 1{\eps\mu}}}.
  \end{align*}
  Theorem~\ref{cor:jac-U-bound} (in particular, its implication that $U^2N=O(1)$)  yields the statement.
\end{proof}

Lemma~\ref{lem:jacobi-id-correct} and~\ref{lem:jacobi-id-coomplexity}
implies the final result, Theorem~\ref{thm:jacobi-1-sps}.

\section{A $k$-sparse Recovery Algorithm for Jacobi Polynomials}\label{sec:k-sps-jacobi}
Finally, we use our $1$-sparse recovery algorithm (i.e. Theorem~\ref{thm:jacobi-1-sps}) along with Theorem~\ref{thm:redux-main} in order to obtain a $k$-sparse recovery algorithm for Jacobi polynomials.  In order to apply Theorem~\ref{thm:redux-main}, we first need to derive some properties of Jacobi polynomials.

Theorem~\ref{thm:jacobi-roots} implies that that any family of Jacobi polynomials is dense as defined in Definition~\ref{def:roots-spread}.
\begin{cor}
  \label{cor:jacobi-dense}
  A family of Jacobi polynomials, with an associated constant $C^{\alpha,\beta}$ defined as in Theorem~\ref{thm:jacobi-roots},
  is $(\frac{1}{2\pi},\frac{3}{2\pi}, 8C^{\alpha,\beta} \pi)$-dense.
\end{cor}
\begin{proof}
Consider the roots $\cos\parens{\theta_0},\dots, \cos\parens{\theta_{d-1}}$ of degree $d$ Jacobi polynomial.
  Theorem~\ref{thm:jacobi-roots} says that for all $i \in [d]$,
  \begin{equation}
    \label{eq:jacobi-roots}
    \theta_i \in \frac{\pi}{d}\parens{i \pm C^{\alpha,\beta}}.
  \end{equation}
  First, we can bound the number of roots falling in any interval $I = [a,b]$.
  Equation~\eqref{eq:jacobi-roots} implies that any root $\theta_i \in [a,b]$ satisfies $\frac{i}{d}\pi \in \left[ a-\frac{\pi}{d}C^{\alpha,\beta}, b+\frac{\pi}{d}C^{\alpha,\beta} \right]$, and therefore there are at most
  \[
    \frac{d}{\pi}(b-a) + 2C^{\alpha,\beta}
  \]
  roots in $I$.
  Similarly, any $i$ such that $\frac{i}{d}\pi \in \left[ a+\frac{\pi}{d}C^{\alpha,\beta}, b-\frac{\pi}{d}C^{\alpha,\beta} \right]$ is a root, so there are at least
  \[
    \frac{d}{\pi}(b-a) - 2C^{\alpha,\beta}
  \]
  roots in $I$.

  Now choose $\gamma_0 = 8C^{\alpha,\beta} \pi$ in the statement of Definition~\ref{def:roots-spread}.
  Fix any $\ell \in [d]$.
 
  For any\footnote{For clarity, note that this is not the constant $\gamma$ defined in Theorem~\ref{thm:jacobi-approx}, but a fresh variable in the statement of Definition~\ref{def:roots-spread}.}
$\gamma \ge \gamma_0/d$, we are interested in bounding the number of roots falling in $\theta_\ell \pm \gamma/2$.
  Note that $\frac{\gamma}{2} \ge C^{\alpha,\beta}\frac{\pi}{d}$ by our choice of $\gamma_0$.
  Therefore the interval in question satisfies the containments
  \[
    \frac{\pi}{d}\ell \pm \left( \frac{\gamma}{2} - \frac{\pi}{d}C^{\alpha,\beta} \right) \subseteq \theta_\ell \pm \frac{\gamma}{2} \subseteq \frac{\pi}{d}\ell \pm \left( \frac{\gamma}{2} + \frac{\pi}{d}C^{\alpha,\beta} \right).
  \]
  The number of roots $r$ lying in this interval therefore satisfies
    \[\frac{d\gamma}{\pi} - 4C^{\alpha,\beta}
         \le r \le \frac{d\gamma}{\pi} + 4C^{\alpha,\beta}.\]
Now, note that since $d\gamma\ge d\gamma_0\ge \frac {8C^{\alpha,\beta}}\pi$, we have
    \[d\gamma\frac{1}{2\pi} \le d\gamma\left( \frac{1}{\pi} - \frac{1}{2\pi} \right) \le d\gamma\left( \frac{1}{\pi} - \frac{4C^{\alpha,\beta}}{d\gamma} \right)\]
and
    \[ d\gamma\left( \frac{1}{\pi} + \frac{4C^{\alpha,\beta}}{d\gamma} \right) \le d\gamma\left( \frac{1}{\pi} + \frac{1}{2\pi} \right) = d\gamma\frac{3}{2\pi},\]
  as desired.

\end{proof}

Finally, our following main result (which is the formal version of Theorem~\ref{thm:jacobi-intro}) follows by
applying Theorem~\ref{thm:redux-main} along with Theorem~\ref{cor:jac-U-bound} and Corollary~\ref{cor:jacobi-dense} (with $C_0=\frac 1{2\pi}$, $C_1=\frac 3{2\pi}$ and $\gamma_0=8C^{\alphajac,\betajac}\pi$).
\begin{cor}[$k$-sparse recovery for Jacobi Transform]
\label{cor:jacobi-k-sps}
Fix parameters $\alphajac,\betajac>-1$.
The following is true for small enough $\eps$ and some constant $C$.
Let $\vJ^{(\alphajac,\betajac)}_N$ denote the $N\times N$ Jacobi transform (as defined in~\eqref{eq:F-entries}) with associated constant $C^{\alpha,\beta}$ in Theorem~\ref{thm:jacobi-roots}.
Then there exists an algorithm with the following property.
Let $\vhx = \vJ^{(\alphajac,\betajac)}_N \vx$ be $(k, \frac{3}{2\pi}\gamma)$-sparsely separated for $\gamma \ge 8C^{\alpha,\beta} \pi / N$, and let $\vw \in \mathbb{R}^N$ with $\vhw = \vJ^{(\alphajac,\betajac)}_N \vw$ be a ``noise'' vector so that $\norm{2}{\vhw}\le \frac \delta{2C}\smll(k,\vhx)$.

Then there exists an algorithm with the following property. for any $0 < \mu < 1$, the algorithm
with probability at least $1-\mu$, 
 outputs $\vhz$ so that $\|\vhx-\vhz\|_{2} \le O(\delta)\|\vhx\|_{2}$.
Further, the algorithm makes at most
\[
  \poly\left(\frac{k\log(1/\mu)}{\gamma \delta }\right) \cdot O\left( \frac{\log N}{\delta^{5/2}} \log\log N \log\left( \frac{1}{\delta} \right)\log\left( \frac{k}{\delta\mu\gamma} \right)\right)
\]
queries to $\vx + \vw$, runs in time at most
\[
  \poly\left(\frac{k\log(1/\mu)}{\gamma \delta }\right) \cdot O\left( \frac{\log N}{\delta^3} \log\log N \log\left( \frac{1}{\delta} \right)\log\left( \frac{k}{\delta\mu\gamma} \right)\right).
\]
\end{cor}

\section{Proof of Theorem~\ref{thm:theta-from-cos}}
\label{sec:chebyshev-1sps}

In this section, we complete the final missing piece by proving Theorem~\ref{thm:theta-from-cos}.
We are looking for an unknown $\theta \in (0, \pi)$.
The queries have the form
\[T_w(\theta)=\frac{\cos(w \theta) \pm \eps}{1\pm \eps}.\]
Querying at $w=1$ gives a range (depending on the size of the noise) of possible $\theta$.

Now at beginning of every stage $t$ of the algorithm, we have a candidate interval $S_t$ which has width of order $2^{-t}$, which we know that $\theta$ lies in.
By querying at $w = \Theta(2^t)$, this interval is ``dilated'' so that the query $T_w(\theta)$ gives more information about $\theta$, and $S_t$ can be narrowed to $S_{t+1}$.

There is one obstacle in that knowing $\cos(w \theta)$ only reveals $w \theta$ up to sign.
When the two possibilities are not well separated (that is, when $\cos(w \theta)$ is close to $\pm 1$), we do not receive enough information to narrow $S_t$.
These cases can be overcome by re-querying at a different suitably chosen $w$.

\paragraph{\approxcos}

We begin with a bound on how much the noise in the query $\cos(\theta)$ affects the argument~$\theta$.
\begin{lmm}\label{lmm:cos-noise}
  If $\epsilon$ is sufficiently small, then $\rho_\epsilon \le 2\sqrt{5\epsilon}$ satisfies Definition~\ref{def:rho}.
\end{lmm}
\begin{proof}

  Note that $\frac{1}{1 \pm \epsilon} \in 1 \pm 2\epsilon$ for $\epsilon < 1 $, so
  \begin{align*}
    \frac{\cos \theta \pm \epsilon}{1 \pm \epsilon}
    &\in (\cos \theta \pm \epsilon)(1 \pm 2\epsilon)
    \\&\in (\cos \theta \pm \epsilon) \pm 2\epsilon(1 \pm \epsilon)
    \\&\in \cos \theta \pm 5\epsilon
  \end{align*}

  By Taylor expansion,
  \[
    \arccos( \cos(\theta) \pm \gamma ) = \theta \pm 2\sqrt{\gamma}.
  \]
  for any $\theta \in [0, \pi]$.
  This completes the claim.
\end{proof}

\begin{algorithm}
\caption{$\approxcos(\tau,\eps)$}\label{alg:cheb}
\begin{algorithmic}
\Require Query access to $T_{\square}(\theta)$, and parameters $\eps,\tau$
\Ensure Range $S_{\tau}$ such that $\theta\in S_{\tau}$

\State{$a_0,b_0 \gets \refine(w=1, S=[0,\pi],\eps)$}
\State{\textbf{Assertion 1:}  $a_0 \neq$\texttt{\,None} }
\State{$S_0 \gets \left[ a_0 - \rho_\epsilon  , a_0 +  \rho_\epsilon \right]$ }

\For{$t=1,\ldots,\tau$}
	\State{$a_t, b_t \gets \refine(w=2^t,S=S_{t-1},\eps)$}

		\Comment{\refine~is in Algorithm~\ref{alg:refine}.}
	\If{ Both $a_t,b_t \neq$ \texttt{None} }

		\State{\textbf{Assertion 2:} There is a unique $h \in \{1, \ldots, 2^t-1 \}$ so that $\frac{ h \pi }{2^t } \in S_{t-1} \pm 4\rho2^{-t}$}
		\State{Let $h =2^j \cdot x$ where $x$ is odd and $j < t$.}
		\State{$w \gets \frac{1}{2}\cdot 2^{t-j} \cdot (2^{j+1}+1)$}
		\State{$a_t,b_t \gets \refine(w, S_{t-1},\eps)$	}

	\EndIf
	\State{ \textbf{Assertion 3:} Exactly one of $a_t, b_t$ is not \texttt{None}.  Call that one $c_t$.}
	\State{$S_t \gets \left[ c_t - \frac{ \rho_\epsilon }{2^t} , c_t + \frac{ \rho_\epsilon }{2^t} \right]$ }
\EndFor
\State \Return{$S_\tau$}

\end{algorithmic}
\end{algorithm}
\begin{algorithm}
\caption{$\refine(w,S,\eps)$}
\label{alg:refine}
\begin{algorithmic}
	\Require A `stretch factor' $w \in \{0,\ldots,N-1\}$, query access to $T_{\square}(\theta)$, a current interval $S$ and a parameter $\eps$
	\Ensure Points $a$ and $b$, which are guesses for $\theta$. They can also be \texttt{None}.

	\State{$r \gets \arccos\parens{T_w(\theta)}$}
	\State{$\bar{S} \gets (S \pm \frac{\rho_\epsilon}{w}) \cap (0, \pi)$}
	\State{Let $a$ be any point in $\bar{S} \cap \left\{ \frac{r}{w} + z \cdot \frac{ 2\pi}{w} \,|\, z \in \mathbb{Z} \right\}$ or \texttt{None} if that intersection is empty.}
	\State{Let $b$ be any point in $\bar{S} \cap \left\{ -\frac{r}{w} + z \cdot \frac{ 2\pi}{w} \,|\, z \in \mathbb{Z} \right\}$ or \texttt{None} if that intersection is empty.}

	\Comment{The difference is the minus sign in front of $r$ in the definition of $b$.}

	\Return{ $a,b$ }
\end{algorithmic}
\end{algorithm}

We are now ready to prove Theorem~\ref{thm:theta-from-cos}.
\begin{proof}[Proof of Theorem~\ref{thm:theta-from-cos}]
We first note that the claim on the query locations follows from statement of Algorithm~\ref{alg:cheb}.

We will proceed by induction on $t$, with the inductive hypothesis that:
\begin{enumerate}
	\item[(a)] $\theta \in S_t$.
	\item[(b)] The width of $S_t$ is at most $|S_t| \leq \frac{ 2\rho_\epsilon }{2^t}$.
\end{enumerate}
Note that the above with $t=\tau$ proves the claim on correctness of Algorithm~\ref{alg:cheb}.

For notational convenience, for the rest of the proof 
\begin{defn}
We will use $\rho$ to denote $\rho_{\eps}$.
We will also extend the interval notation to let $\pm x \pm h$ denote $-x \pm h \cup x \pm h = [-x - h, -x + h] \cup [x - h, x + h]$.
\end{defn}

Let $\phi(x)$ be the value $x + z \cdot 2\pi$ for some integer $z$ so that $\phi(x) \in [-\pi, \pi]$.
We begin with an observation:
\begin{observation}\label{obs:close}
In $\refine(w,S,\eps)$,
\[ r = \pm \phi\left( w \cdot \theta \right) \pm \rho. \]
\end{observation}
\begin{proof}
Since $w\theta$ can be outside of $[-\pi,\pi]$, we use the fact that $\cos(w\theta)=\cos\parens{\phi(w\theta)}$.
If $\phi(w \cdot \theta) \in [0, \pi]$, then Lemma~\ref{lmm:cos-noise} gives
\[ r = \arccos\left( \frac{\cos \left( \phi( w \cdot \theta ) \right) \pm \eps}{1 \pm \epsilon} \right) = \phi(w \cdot \theta) \pm \rho. \]
Otherwise $\phi(w \cdot \theta) \in [-\pi, 0]$ and
\[ r = \arccos\left( \frac{\cos \left( \phi( w \cdot \theta ) \right) \pm \eps}{1 \pm \epsilon} \right) = \arccos\left( \frac{\cos \left( -\phi( w \cdot \theta ) \right) \pm \eps}{1 \pm \epsilon} \right) = - \phi(w \cdot \theta) \pm \rho. \]
\end{proof}

\paragraph{Base case.}
Now we establish the base case, for $t=0$.  Item (b) holds by construction of $S_0$, so the only things to check are that (a) holds, and also that the \textbf{Assertion 1} in the pseudocode is correct.

By Observation~\ref{obs:close},
in \refine$(1, [0,\pi],\eps)$,
we set
\[ r = \pm \phi(\theta) \pm \rho.\]
Since $\theta \in [0,\pi)$ and $r \in [0,\pi]$, the sign must be positive, and we have
\[ r = \theta \pm \rho. \]
We claim that we will set $a_0$ to be $a_0 = r$.
Certainly this value lives in $\bar{S} = [0, \pi]$, which ensures that $a_0$ will not be set to \texttt{None} (and hence \textbf{Assertion 1} holds).
Moreover, no other value $r + z \cdot 2\pi$ will be used for $a_0$, since for any $z \neq 0$, $r + z \cdot 2 \pi$ does not live in $[0, \pi]$.
Thus $a_0$ is as claimed.

Then $|a_0 - \theta| \leq \rho$, which means that the choice of
$S_0 = [ a_0 - \rho , a_0 + \rho ]$
indeed satisfies (a).

\paragraph{Establishing the inductive hypothesis.}
With the base case out of the way we proceed by induction.

\begin{claim}\label{claim:oneclose} Suppose that the inductive hypothesis holds for $t-1$, and suppose that $w \leq \frac{3}{2} \cdot 2^t$, and that $\rho$ is sufficiently small (say, $\rho < \pi/4$).  Then consider
\[ a_t, b_t \gets \refine(w, S=S_{t-1},\eps) \]
Then at least one of the following holds:
\[ |a_t - \theta| \leq \frac{\rho}{w} \qquad \text{or} \qquad |b_t - \theta| \leq \frac{\rho}{w}. \]
\end{claim}
\begin{proof}
By Observation~\ref{obs:close},
\[ r = \pm \phi( w \cdot \theta ) \pm \rho \]
which means that for some $z \in \mathbb{Z}$, and some sign $\zeta \in \{-1,1\}$,
\[ \zeta \cdot \frac{r}{w} + z \cdot \frac{2\pi}{w} = \theta \pm \frac{ \rho}{w}. \]
Suppose that $\zeta = +1$.
Then
\[ \left|\frac{r}{w} + z \cdot \frac{ 2\pi}{w} - \theta\right| \leq \frac{\rho}{w}, \]
which implies that (using the inductive hypothesis that $\theta \in S_{t-1}$) we have
\[ \frac{r}{w} + z \cdot \frac{2 \pi }{w} \in \overline{S_{t-1}}. \]
Moreover, this value of $z$ is the unique one with this property, since
\[ | \overline{S_{t-1}} | \leq \frac{4\rho}{2^t} + \frac{ 2\rho}{w} \leq \frac{ 6\rho }{w} + \frac{ 2\rho }{w} < \frac{ 2\pi }{w}  \]
using the fact that $w \leq \frac{3}{2}\cdot 2^t$ and that $\rho < \pi/4$ is sufficiently small.
Thus, for any other $z'$,
\[ \frac{r}{w} + z' \cdot \frac{2\pi}{w} \not\in \overline{S_{t-1}}. \]
This implies that, in the case where $\zeta = +1$, we will choose this $z$ in our definition of $a_t$, and hence
\[ a_t = \frac{r}{w} + z \cdot \frac{2\pi}{w} = \theta \pm \frac{ \rho }{w}. \]

Similarly, if $\zeta = -1$, we will choose
\[ b_t = \theta \pm \frac{ \rho }{w}. \]

Since $\zeta$ is either $-1$ or $+1$, one of the two cases will hold and this proves the claim.
\end{proof}

With this claim out of the way, we will establish the inductive hypotheses for the next round.  Again (b) holds by construction of $S_t$, so we focus on (a).
Consider $a_t,b_t = \refine(w=2^t, S=S_{t-1},\eps)$, the output of the first call to \refine.

\textbf{Case 0: at least one of $a_t$, $b_t$ is \texttt{None}.} If at least one of $a_t, b_t$ is \texttt{None},
the Claim~\ref{claim:oneclose} implies that one of them (say, $c_t$) is not \texttt{None} and that
\[ |c_t - \theta| \leq \frac{ \rho }{2^t}, \]
and so in particular $\theta \in S_t$, establishing the inductive hypothesis (a).

\textbf{Case 1: both $a_t$ and $b_t$ are not \texttt{None}.}
Since they are both in $\overline{S_{t-1}}$, by the inductive hypothesis $|a_t - b_t| \leq \abs{\overline{S_{t-1}}}\le  |S_{t-1}|+2\rho 2^{-t} \leq 2\rho{2^{-(t-1)}} +2\rho 2^{-t} = 6\rho{2^{-t}}$.
This means that, if $z$ is chosen in the definition of $a$ and $z'$ for $b$, then
\[ | a_t - b_t | = \left| \frac{r}{w} + z \frac{ 2\pi}{ w} - \frac{-r}{w} - z' \frac{2\pi }{w} \right| \leq \frac{ 6\rho }{ 2^t }. \]
Thus,
\[ \left| 2r + (z - z') 2\pi \right| \leq 6\rho, \]
or
\[ |r + (z-z') \pi | \leq 3\rho, \]
which implies that (since $r \in (0,\pi)$ and $(z-z')$ is an integer),
\[ r \in (0, 3\rho] \cup [ \pi - 3\rho, \pi]. \]
On the other hand, by Observation~\ref{obs:close},
\[ r = \pm \phi(2^t \cdot \theta) \pm \rho, \]
so, it follows that for some integer $h$,
\begin{equation}\label{eq:closetopi}
 \theta = \frac{ \pi \cdot h }{2^t} \pm 4 \frac{\rho}{2^t}.
\end{equation}
Thus $\frac{\pi h}{2^t} = \theta \pm 4\frac{\rho}{2^t} \in S_{t-1} \pm 4\frac{\rho}{2^t}$.
This interval has width at most $12\rho2^{-t}$, so this integer $h$ can be found from $S_{t-1}$ assuming that $\rho < \pi/12$ is sufficiently small (and this $h$ is unique).

\paragraph{Assertion 2 holds.}
Since $\theta \in (0,\pi)$, this implies that $h \in \{0,1,\ldots, 2^t\}$.
In fact, to prove \textbf{Assertion 2}, we will show that $h$ cannot be $0$ or $2^t$.

Suppose for the sake of contradiction that $h$ is zero.
Then $\theta \in [-4\rho 2^{-t}, 4\rho 2^{-t}]$, which means that $\theta \in (0, 4\rho 2^{-t}]$, since $\theta \in (0, \pi)$.
This in turn implies that 
\begin{equation}
\label{eq:S-incl}
S_{t-1} \subseteq (0,\pi) \cap (\theta \pm 4\rho2^{-t}) \subseteq [-4\rho2^{-t}, 8\rho2^{-t}]
\end{equation}
where the first inclusion follows from the fact that $|S_{t-1}|\le 4\rho2^{-t}$.

By Observation~\ref{obs:close}, $r = \phi(2^t \cdot \theta) \pm \rho \in (0, 5\rho]$ (provided $4\rho < \pi$). We note that we also used the fact that $r\in (0,\pi)$ as well as $\theta\ge 0$ to conclude that the sign in front of $\phi$ should be positive.

Note that if $14\rho < 2\pi$, then $\frac{2\pi-r}{2^t} > \frac{9\rho}{2^t}$.
In the call to {\refine}, since $\overline{S_{t-1}} \subseteq (0, 9\rho2^{-t}]$ (which in turn follows from~\eqref{eq:S-incl}), this means $b_t$ would have been set to \texttt{None}.\footnote{Note that $\frac {-r}{2^t}+z\cdot \frac{2\pi}{2^t}<0$ for any $z\le 0$ (since $r>0$) and is $> \frac{9\rho}{2^t}$ for any $z\ge 1$ (since we showed it for $z=1$ and $z>1$ can only give a larger value) so this exhausts all possibilities.} 
Thus the case $h = 0$ could not have happened.
In the case $h = 2^t$, we similarly bound
\begin{align*}
  \theta &\in [\pi-4\rho2^{-t}, \pi] \\
  S_{t-1} &\in [\pi-8\rho2^{-t}, \pi+4\rho2^{-t}] \\
  \overline{S_{t-1}} &\in [\pi-9\rho2^{-t}, \pi] \\
  r &= (\pm\phi(2^t \cdot \theta) \pm \rho) \cap (0, \pi) \in (0, 5\rho].
\end{align*}
The claim on $r$ is due to the following argument. Note that $2^t\theta\in[2^t\pi-4\rho,2^t\pi]$, which implies $\phi\parens{2^t\theta}\in [-4\rho,0]$, which in turn implies that $-\phi\parens{2^t\theta}\in (0,4\rho]$ and this is the interval that gets `picked up' when intersecting with $(0,\pi)$.

We now argue that both $\frac{r}{2^t} + 2^{t-1}\frac{2\pi}{2^t}$ and $\frac{r}{2^t} + (2^{t-1}-1)\frac{2\pi}{2^t}$ are not in $\overline{S_{t-1}}$.
Indeed  $\frac{r}{2^t} + 2^{t-1}\frac{2\pi}{2^t}>\pi$ and hence $\frac{r}{2^t} + 2^{t-1}\frac{2\pi}{2^t}\not\in \overline{S_{t-1}}$. Next we argue that 
\[\frac{r}{2^t} + (2^{t-1}-1)\frac{2\pi}{2^t} < \pi-9\rho2^{-t},\]
which implies that the LHS is not in $\overline{S_{t-1}}$ as needed. Now note that the above inequality is true if $r-2\pi <-9\rho$, which is true if $\rho <\pi/7$ (where we also used the fact that $r\le 5\rho$). The above also implies that for $z\in\Z$, $\frac r{2^t} +z\cdot \frac{2\pi}{2^t}$ is also not in $\overline{S_{t-1}}$ since the value is $>\pi$ for $z\ge 2^{t-1}$
(since it is true for $z=2^{t-1}$ as $r>0$) and the value is $<\pi-9\rho2^{-t}$ for $z\le 2^{t-1}-1$ (since it is true for $z=2^{t-1}-1$).
This implies that $a_t$ would have been \texttt{None}. 
So $h=2^t$ is also impossible.

Therefore $h \in \{1, \dots, 2^{t}-1\}$, and write $h = 2^j \cdot x$ where $x$ is odd and $j < t$.
This establishes the \textbf{Assertion 2} in the pseudocode.

\paragraph{Assertion 3 holds.}
Now, we re-choose
\begin{equation}
\label{eq:w-val}
 w = \frac{1}{2}2^{t-j}(2^{j+1}+1). 
\end{equation}
Since $t > j$, $w$ is an integer.
Further, $2^t \leq w \le \frac{3}{2} \cdot 2^t$, and so (because of our bound on $\tau$ in Theorem~\ref{thm:theta-from-cos}), $w$ is an integer between $0$ and $N-1$ inclusive and thus is a valid query.

By \eqref{eq:closetopi},
\[ \theta = \frac{ \pi \cdot 2^j \cdot x }{2^t} \pm 4 \frac{\rho}{2^t}, \]
so we have
\begin{align*}
 w \cdot \theta &= \frac{1}{2} \cdot (2^{j+1}+1) x \cdot \pi \pm 4 w \frac{\rho}{2^t} \\
&\in \frac{x'}{2} \pi\pm 6 \rho,
\end{align*}
for some odd integer $x'$,
using the fact that $w \leq \frac 32 \cdot2^t$.

Now consider the $r$ that is defined in the second call to \refine, with $w$ as in~\eqref{eq:w-val}.
By Observation~\ref{obs:close} again,
\[ r = \pm \phi( w \cdot \theta) \pm \rho
= \pm \phi\left( \frac{x'}{2} \pi \pm 6 \rho \right) \pm \rho
= \pm \frac{\pi}{2} \pm 7 \rho. \]
Now for any choice of $z, z'$,
\begin{align*}
  \abs{\parens{\frac{r}{w} + z\frac{2\pi}{w}} - \parens{-\frac{r}{w} + z'\frac{2\pi}{w}}}
  &= \abs{\frac{2r}{w} + (z-z')\frac{2\pi}{w}}
  \\&= \abs{\frac{\pm\pi \pm 14\rho}{w} + (z-z')\frac{2\pi}{w}}
  \\&= \abs{\frac{\pm 14\rho + \pi(2z''+1)}{w}}
  \\&\ge \frac{\pi-14\rho}{w}
\end{align*}
where $z'' \in \mathbb{N}$.
This implies that there cannot both be an $a_t, b_t$ in $\overline{S_{t-1}}$ as long as
\[  \frac{\pi-14\rho}{w} > \frac{ 6\rho }{w}+\frac{2\rho}{w} \ge \frac{ 4\rho }{2^t}+\frac{2\rho}{w} \ge |\overline{S_{t-1}}|, \]
which holds provided $\rho < \pi/22$.
Thus at least one of $a_t,b_t$ in the second call to \refine\ is \texttt{None}, and by Claim~\ref{claim:oneclose} the other (call it $c_t$) exists and satisfies
\[ | c_t - \theta | \leq \frac{ \rho }{w} \leq \frac{ \rho }{ 2^t}. \]
This establishes \textbf{Assertion 3}.
Thus, when we define $S_t = [c_t - \rho/2^t, c_t + \rho/2^t]$, we guarantee that $\theta \in S_t$, establishing the inductive hypothesis (a) for the next round.

\paragraph{Query and Runtime analysis.} It is easy to see that Algorithm~\ref{alg:cheb} makes at most $2\tau$ calls to Algorithm~\ref{alg:refine}, which implies that it makes at most $2\tau$ queries.
It is also easy to check that other than calls to $\refine$, the rest of $\approxcos$ runs in time $O(\tau)$. To finish the proof, we claim that each call to $\refine$ takes $O(1)$ time.\footnote{This assumes computing the $\arccos$ values takes $O(1)$ time.}
The main operation is computing the values $a$ and $b$.
To compute $a$, it suffices to find any $z\in\Z$ satisfying $L \le \frac{r}{w} + z \cdot \frac{2\pi}{w} \le R$, or any integer in $[wL/2\pi - r/2\pi, wR/2\pi - r/2\pi]$. The case for $b$ is similar.
\end{proof}

\section*{Acknowledgments}
We would like to thank Mark Iwen for useful conversations. Thanks also to Stefan Steinerberger for showing us the proof of Lemma~\ref{lem:bad-in-spread-sequence} and graciously allowing us to use it, and to  Cl\'{e}ment Canonne for helpful comments on our manuscript.  

ACG is partially funded by a Simons Foundation Fellowship. AR is partially funded by NSF grant CCF-1763481.  MW is partially funded by NSF CAREER award CCF-1844628.
AG and CR gratefully acknowledge the support of DARPA under Nos.\ FA87501720095 (D3M) and FA86501827865 (SDH), NIH under No.\ U54EB020405 (Mobilize), NSF under Nos.\ CCF1763315 (Beyond Sparsity) and CCF1563078 (Volume to Velocity), ONR under No.\ N000141712266 (Unifying Weak Supervision), the Moore Foundation, NXP, Xilinx, LETI-CEA, Intel, Google, NEC, Toshiba, TSMC, ARM, Hitachi, BASF, Accenture, Ericsson, Qualcomm, Analog Devices, the Okawa Foundation, and American Family Insurance, Google Cloud, Swiss Re, Stanford Bio-X SIG Fellowship, and members of the Stanford DAWN project: Intel, Microsoft, Teradata, Facebook, Google, Ant Financial, NEC, SAP, VMWare, and Infosys. The U.S.\ Government is authorized to reproduce and distribute reprints for Governmental purposes notwithstanding any copyright notation thereon. Any opinions, findings, and conclusions or recommendations expressed in this material are those of the authors and do not necessarily reflect the views, policies, or endorsements, either expressed or implied, of DARPA, NIH, ONR, or the U.S.\ Government.

\bibliographystyle{acm}
\bibliography{OP-sparse}

\appendix

\section{Chebyshev Reduction to sparse FFT}
\label{sec:dct-reduction}

For completeness, in this section we show how to solve the sparse approximation problem for Chebyshev polynomials via a reduction to the sFFT.  We note that a similar observation has been made in \cite{legendre-sparse}. 

In this section, everything is $0$-indexed, and negative indices count from the end.
For example, if $\vx$ has length $N$ then $\vx[-i] = \vx[N-i]$.

\subsection{Facts about DFT}
We begin with some useful facts about the Discrete Fourier Transform (DFT).
Consider the usual DFT $\hat{\vx}=\mathcal{F}\vx$ defined by
\[
  \hat{\vx}[j] = \sum_{\ell=0}^{N-1} \vx[\ell] \omega_N^{j\ell}
\]
where $\vx \in \C^N$ and $\omega_N := \exp(-2\pi i/N)$.

\begin{enumerate}
  \item\label{fact:fft-shift} Shifting $\vx$ scales $\hat{\vx}$ elementwise, and vice versa:

  Consider $\vx \in \C^N$ and $\vy$ such that $\vy[\ell] = \vx[\ell-s]$ for all $\ell \in [N]$. Then
  \[
    \hat{\vy}[j] = \sum_{\ell=0}^{N-1} \vy[\ell] \omega^{j\ell} = \sum_{\ell=0}^{N-1} \vx[\ell-s] \omega^{j(\ell-s)}\omega^{js} = \omega^{js} \hat{\vx}[j].
  \]
  If $\vy[\ell] = \vx[\ell] \omega^{\ell s}$, then
  \[
    \hat{\vy}[j] = \sum_{\ell=0}^{N-1} \vy[\ell] \omega^{j\ell} = \sum_{\ell=0}^{N-1} \vx[\ell] \omega^{(j+s)\ell} = \hat{\vx}[j+s].
  \]

  \item\label{fact:fft-support} If $\vx$ is supported only on the first half, then the even indices of $\hat{\vx}$ can be computed by a DFT of half the size, and vice versa:

  \[
    \hat{\vx}[2j] = \sum_{\ell=0}^{N/2-1} \vx[\ell] \omega_N^{2j\ell} = \sum_{\ell=0}^{N/2-1} \vx[\ell] \omega_{N/2}^{j\ell}.
  \]

  \[
    \hat{\vx}[j] = \sum_{\ell\text{ even}} \vx[\ell] \omega_N^{j\ell} = \sum_{\ell=0}^{N/2-1} \vx[2\ell] \omega_{N/2}^{j\ell}
  \]
  for $j = 0, \dots, N/2-1$.

  \item\label{fact:fft-sym} $\vx$ is symmetric ($\vx[i] = \vx[-i]$) if and only if $\hat{\vx}$ is symmetric.

  \[ \hat{\vx}[j] = \sum_{\ell=0}^{N-1} \vx[\ell]\omega^{j\ell} = \sum_{\ell=0}^{N-1} \vx[-\ell]\omega^{j\ell} = \sum_{k=0}^{N-1} \vx[\ell]\omega^{-j\ell} = \hat{\vx}[-j] \]
\end{enumerate}

\subsection{Sparse DCT Setup}

Consider the Chebyshev transform $\hat{\vc} = \mathcal{C}\vc$
\begin{equation}
  \hat{\vc}[j] = \sum_{\ell=0}^{N-1} \vc[\ell] T_\ell(\lambda_j)\label{eq:chebyshev}
\end{equation}
where $\lambda_j = \cos\left( \pi\frac{2j+1}{2n} \right)$ are the Chebyshev nodes.
Note that $T_\ell(\lambda_j) = \cos\left( \frac{\pi}{2N}\ell(2j+1) \right)$ by properties of Chebyshev polynomials.

\subsection{Reduction}

We will prove the following equivalence:
\begin{lmm}\label{lmm:dct-reduction}
  For any $\vc \in \C^N$ and its transform $\hat{\vc}=\mathcal{C}\vc$ there exists an $\vf\in\C^{2N}$ and $\hat{\vf}=\mathcal{F}\vf$ so that one can express each entry in $\vf$ (and $\hat{\vf}$) as a scalar multiple of an entry in $\vc$ (and $\hat{\vc}$ resp.).
  Further, $\vc$ is k-sparse (or $\hat{\vc}$ is k-sparse) if and only if $\vf$ is 2k-sparse (or $\hat{\vf}$ is 2k-sparse resp.).
\end{lmm}
This implies that ${\cal C}$ (and its inverse) have sparse recovery algorithms by reducing to a sparse FFT.

\begin{proof}

We will define $\vx,\vy,\vz\in\C^{4N}$ and $\vf\in\C^{2N}$ as follows.
Let $\omega := \omega_{4N}$.

Define $\vx \in \C^{4N}$ by 
\[
  \begin{cases}
    \vx[0] = 2\vc[0] & \\
    \vx[j] = \vx[4N-j] = \vc[j] & j = 1, \dots, N-1 \\
    \vx[j] = 0 & \text{otherwise}.
  \end{cases}
\]

Define $\vy[j] = \vx[j-N]$ for all $j \in [4N]$,
More specifically,
\[
  \vy =
  \begin{bmatrix}
    0 \\ \vc[N-1] \\ \vdots \\ \vc[1] \\ 2\vc[0] \\ \vc[1] \\ \vdots \\ \vc[N-1] \\ 0 \\ \vdots
  \end{bmatrix}.
\]
Note that $\vy$ is only supported on indices $0,\dots,2N-1$.
Finally, define $\vz[j] = \omega^j\vy[j]$ for all $j \in [4N]$ and $\vf[j] = \vz[j]$ for all $j \in [2N]$.
By construction, every entry of $\vf$ is a scalar multiple of an entry of $\vc$.
Furthermore for every $j\in[N]$ exactly two entries of $\vf$ depend on $\vc[j]$ (unless $j=0$, in which case one entry of $\vf$ depends on $\vc[0]$), so $\vc$ is $k$-sparse if and only if $\vf$ is $2k$-sparse.

Now we analyze $\hat{\vf}$ and show that it has the same properties.

\begin{claim}
\begin{equation}\label{eq:fourier}
  \hat{\vx}[4N-2j-1] = \hat{\vx}[2j+1] = \sum_{\ell=0}^{4N-1} \vx[\ell]\omega^{\ell(2j+1)} = 2\hat{\vc}[j]
\end{equation}
for $j = 0, \dots, N-1$.
\end{claim}
\begin{proof}
  Writing $2\cos(z \pi/2N) = \omega_{4N}^z + \omega_{4N}^{-z}$, we have
  \[
    \hat{\vc}[j] = \frac{1}{2}\sum_{\ell=0}^{N-1} \vc[\ell] \left( \omega^{\ell(2j+1)} + \omega^{-\ell(2j+1)} \right) = \frac{1}{2}\sum_{\ell=0}^{N-1} \vc[\ell] \left( \omega^{\ell(2j+1)} + \omega^{(4N-\ell)(2j+1)} \right) =\frac 12\cdot \vhx[2j+1]
  \]
  for $j = 0, \dots, N-1$. In the above, the first equality follows from the definition of DCT while the last equality follows from definition of $\vx$.

  Finally, $\vx$ is symmetric, so by fact~\ref{fact:fft-sym} $\hat{\vx}$ is as well. 
\end{proof}

Next, by fact~\ref{fact:fft-shift}, $\hat{\vy}[j] = \hat{\vx}[j] \omega^{jN}$ for all $j \in [4N]$.
Also by fact~\ref{fact:fft-shift}, $\hat{\vz}[j] = \hat{\vy}[j+1]$ for all $j \in [4N]$.
Finally, by fact~\ref{fact:fft-support}, 
\[ \hat{\vf}[j] = \hat{\vz}[2j] = \hat{\vy}[2j+1] = \omega^{(2j+1)N}\hat{\vx}[2j+1] = 2\omega^{(2j+1)N}\hat{\vc}[\min(j,2N-1-j)], \]
for every $j \in [2N]$ (note here that $\hat{\vf}$ is computed from a Fourier transform of size $2N$ instead of $4N$ since their dimensions differ). Here the last equality follows from~\eqref{eq:fourier}.

Therefore for every $j \in [N]$, two entries of $\hat{\vf}$ are a scalar multiple of $\hat{\vc}[j]$.
This concludes the proof of Lemma~\ref{lmm:dct-reduction}.

\end{proof}

\section{Why not apply Jacobi to Chebyshev reduction directly to $k$-sparse recovery?}\label{app:hard}

In this appendix we outline an approach to directly reduce $k$-sparse recovery of Jacobi to $k$-sparse recovery for DCT (and hence sFFT) and outline the main challenges in carrying out the reduction. The hope here is to give the reader a sense for why some obvious generalization of the ideas in Appendix~\ref{sec:dct-reduction} do not work for Jacobi polynomials. 

The main idea is to use Theorem~\ref{thm:jacobi-approx} to reduce evaluations of Jacobi polynomials to (approximately) evaluations of Chebyshev polynomials. One might hope then that one could directly reduce the $k$-sparse
recovery problem for Jacobi to an $O(k)$-sparse recovery for Chebyshev transform. One obvious issue is that this reduction (via  Theorem~\ref{thm:jacobi-approx}) is approximate and we need to handle the approximation error in our final analysis. However, for this discussion, let us assume this can be handled.

The major issue is that we now have to deal with Chebyshev polynomials evaluated at roots of Jacobi polynomials. There are  two natural options here:
\begin{enumerate}
\item First observe that most roots of Jacobi polynomials are close enough to the roots of the Chebyshev polynomials. In particular, by a tighter version of Theorem~\ref{thm:jacobi-roots} for the special case of $-1/2 \le \alphajac,\betajac \le 1/2$, we know that the $\ell$th root is given by $\cos{\theta_{\ell}}$, where $\theta_{\ell}=\frac{(\ell+\tilde{\ell})\pi}{N+\addjac}$, for some $0< \tilde{\ell} < 1$ (Theorem 6.3.2 in~\cite{szego}). So at least for the more restricted case of  $-1/2 \le \alphajac,\betajac \le 1/2$ (note our results hold for {\em any} fixed $\alphajac,\betajac>-1$) the $\ell$'th Chebyshev and Jacobi roots are fairly close. However, the crucial difference is that all the Chebyshev roots are equispaced, which allows us to permute the $\theta_{\ell}$ by just applying a permutation on $[N]$. More precisely for any permutation $\sigma$ on $[N]$, note that $\cos\parens{\parens{\frac{\sigma(\ell)+\frac 12}{N}}\pi}$ is the $\sigma(\ell)$'th root. However, since we do not have a closed form expression for Jacobi polynomial roots, it is not clear how to apply similar tricks. Further, the argument in Appendix~\ref{sec:dct-reduction} very crucially uses that the Chebyshev roots are equally spaced and trying to generalize those arguments to Jacobi polynomial roots breaks down immediately. In summary, we do not see how to directly reduce the Jacobi transforms to DCT or DFT.
\item Like in Section~\ref{sec:jacobi-1sps}, we just try and do some ``error-correction" to obtain cosine evaluations at regular points. Indeed, this is what the trick of $\cos{A}=\frac{\cos(A+B)+\cos(A-B)}{2\cos{B}}$ for random $B$ helped us do for the special case of $k=1$. However, this approach fails even for $k=2$ since we no longer have `common terms' that cancel like they did in proof of Corollary~\ref{cor:jacobi-to-cos}.
\end{enumerate}

\section{Known RIP results}

We begin with some notation. Given a matrix $\vM$, we denote $\max_{i,j} |\vM[i,j]|$ by $\|\vM\|_{\infty}$. Given an $N\times N$ diagonal matrix $\vD$ and a subset $S\subseteq [N]$, $\vD_S$ denotes the diagonal matrix given by
\[\vD_S[i,i]=\begin{cases}
\vD[i,i] & \text{ if } i\in S\\
0 & \text{ otherwise}.
\end{cases}
\]
Finally for any matrix $\vM$, we define $\sqrt{\vM}$ to be the matrix where we apply the $\sqrt{\cdot}$ operator to each entry of $\vM$.

We will need the following well-known result:
\bthm[\cite{RV}]
\label{thm:rip}
Let $\vA$ be an $N\times N$ orthogonal matrix (i.e. $\vA^T\vA=\vI$) and let $\|\vA\|_{\infty} = U$. Let $r\ge 1$ be an integer and $0<\delta<1$ be a real. Define an integer $1\le s\le N$ such that
\[s\ge \Omega\left(U^2N\cdot \frac{r \log{N}}{\delta^2}\cdot \log\left(\frac{r \log{N}}{\delta^2}\right)\log^2{r}\right).\]
Pick a random subset $S\subseteq [N]$ where each element of $[N]$ is picked iid with probability $\frac{s}{N}$. Define the matrix
\[\vB = \sqrt{\frac{N}{s}} \cdot \vI_S\cdot \vA.\]
Then the following holds:
\begin{equation}
\label{eq:rip-condn}
\Avg{\sup_{\|\vz\|_0\le r} \left| \frac{\|\vB\vz\|_2^2}{\|\vz\|_2^2} -1\right|} \le \delta.
\end{equation}
\ethm

We will use the above result for $r=2$ to argue the following well-known corollary that the matrix $\vB$ essentially preserves the inner products of columns of $\vA$:
\bcor
\label{cor:ip-preserve}
Let $s$ be an integer such that
\begin{equation}
\label{eq:s-lb}
s\ge \Omega\left(U^2N\cdot \frac{\log{N}}{\delta^2}\cdot \log\left(\frac{\log{N}}{\delta^2}\right)\right).
\end{equation}
The consider  $\vA$, $S$ and $\vB$ be as defined in Theorem~\ref{thm:rip} with $r=2$. Then we have for any $0<\alpha<1$,
\begin{equation}
\label{eq:ip-discrepancy}
\Prob{ \|\vB^T\vB -\vI\|_{\infty} \le \frac{2\delta}{\alpha}}\ge 1-\alpha.
\end{equation}
\ecor
\begin{proof}
We apply Theorem~\ref{thm:rip} with $r=2$. Then by Markov's inequality, we have that with probability at least $1-\alpha$, we have for any $i,j\in [N]$:
\begin{equation}
\label{eq:e_i+e_j}
\left| \|\vB (\ve_i+\ve_j) \|_2^2-2\right| \le 2\delta'\eqdef \frac{2\delta}{\alpha},
\end{equation}
and
\begin{equation}
\label{eq:e_i}
\left| \|\vB \ve_i \|_2^2-1\right| \le \delta'.
\end{equation}
Note that
\[\|\vB(\ve_i+\ve_j)\|_2^2 = (\ve_i+\ve_j)^T \vB^T\vB (\ve_i+\ve_j) = \ve_i^T \vB^T\vB \ve_i +  \ve_j^T \vB^T\vB \ve_j +2 \ve_i^T \vB^T\vB \ve_j = \|\vB\ve_i\|_2^2 + \|\vB\ve_j\|_2^2  +2\ve_i^T \vB^T\vB \ve_j.\]
Using the above we have
\[\left| \|\vB (\ve_i+\ve_j) \|_2^2-2\right| \ge 2\left|\ve_i^T \vB^T\vB \ve_j\right| - \left| \|\vB \ve_i \|_2^2-1\right| - \left| \|\vB \ve_j \|_2^2-1\right|.\]
The above along with~\eqref{eq:e_i+e_j} and~\eqref{eq:e_i} implies that for every $i,j\in [N]$,
\[ \left|\ve_i^T \vB^T\vB \ve_j-\delta_{i,j}\right| \le 2\delta',\]
which completes the proof (since $\left(\vB^T\vB\right)[i,j]= \ve_i^T \vB^T\vB \ve_j$).
\end{proof}

\section{More on Jacobi polynomials}
\label{app:jacobi}

We begin by collecting more results on Jacobi polynomials and more generally orthogonal polynomials in Appendices~\ref{app:jacobi-more} and~\ref{app:CD-formula} respectively. Finally, we prove Theorem~\ref{cor:jac-U-bound} in Appendix~\ref{app:jac-U-bound}.

\subsection{More results on Jacobi polynomials}
\label{app:jacobi-more}

We start off with the definition of the {\em Bessel function}:
\begin{defn}
\label{def:bessel}
The {\em Bessel function of the first kind} for parameter $\alpha\in\R$ is defined as follows. For every\footnote{The Bessel function is also defined for $z\in\C$ but since we only need it over reals, we only state it for reals.} 
 $z\in\R$:
\[\bessel{\alpha}{z}=\sum_{\nu=0}^{\infty} \frac{(-1)^\nu}{\nu!\Gamma(\nu+\alpha+1)}\cdot \parens{\frac z2}^{2\nu+\alpha}.\]
\end{defn}

We will need the following approximation and bounds for the Bessel function:
\begin{lmm}
\label{lem:bessel-approx}
Let $\alpha> -1$.
For large enough $z$ (compared to $\alpha$), we have
\begin{equation}
\label{eq:bessel-large}
\bessel{\alpha}{z} =\sqrt{\frac 2{\pi z}}\cdot\cos\parens{z-\parens{\alpha+\frac 12}\cdot\frac \pi 2}\pm O\parens{\frac 1{z^{3/2}}}.
\end{equation}
For all fixed $0<\eps_0<\sqrt{3(\alpha+1)}$ and every $z\in\brackets{\eps_0,\min\set{2,\sqrt{3(\alpha+1)}}}$, we have
\begin{equation}
\label{eq:bessel-small-theta}
\abs{\bessel{\alpha}{z}} = \Theta_{\alpha,\eps_0}(1).
\end{equation}
For all fixed $C_0>0$ and every $0\le z\le C_0$, we have
\begin{equation}
\label{eq:bessel-small-bigO}
\abs{\bessel{\alpha}{z}} = O_{\alpha,C_0}(1).
\end{equation}
Finally, for all $z\ge 0$, we have
\begin{equation}
\label{eq:bessel-large-simple}
\abs{\sqrt{z}\bessel{\alpha}{z}} \le O_{\alpha}(1).
\end{equation}
\end{lmm}
\begin{proof}
Equation~\eqref{eq:bessel-large} is
equation (1.71.7) in~\cite{szego}. Further,~\eqref{eq:bessel-large-simple} follows by combining~\eqref{eq:bessel-small-bigO} and~\eqref{eq:bessel-large}. In the rest of the proof, we argue~\eqref{eq:bessel-small-theta} and~\eqref{eq:bessel-small-bigO}.

We begin with~\eqref{eq:bessel-small-bigO}. We will argue that there is a large enough $\nu_0$ (that depends on $C_0$) for which there is a constant $\gamma_0$ (such that $|\gamma_0|<1$) with the following property. For every $\nu\ge \nu_0$, we have that
\begin{equation}
\label{eq:b-nu-ub}
\abs{b_\nu}\le \parens{\gamma_0}^\nu,
\end{equation}
where
\[b_{\nu}=\frac{(-1)^\nu}{\nu!\Gamma(\nu+\alpha+1)}\cdot \parens{\frac z2}^{2\nu+\alpha}.\]
Notice~\eqref{eq:b-nu-ub} implies that
\[\abs{\bessel{\alpha}{z}} =\sum_{\nu=0}^{\nu_0-1}\abs{b_{\nu}} +O(1).\]
Since $\alpha>-1$, every term in the sum is also $O(1)$ and hence~\eqref{eq:bessel-small-bigO} holds. 

Next, we argue~\eqref{eq:b-nu-ub}. Indeed,
\[\abs{b_{\nu}}= \parens{\frac z2}^{\alpha}\parens{\frac{(z/2)^{2\nu}}{\nu!\Gamma(\nu+\alpha+1)}}.\]
By Stirling's approximations of the Gamma function, we have $\nu! \ge \sqrt{2\pi} \nu^{\nu+\frac 12}\cdot e^{-\nu}$ and $\Gamma(\nu+\alpha+1)\ge \Gamma(\nu) \ge \sqrt{2\pi} \nu^{\nu-\frac 12}\cdot e^{-\nu}$ (the first inequality follows since $\alpha>-1$ and the fact that $\Gamma(y)$ is an increasing function for large enough $y$). This along with the above equation implies that
\[\abs{b_{\nu}}\le \parens{\frac z2}^{\alpha}\cdot\frac 1{2\pi}\cdot \parens{\frac {z^2e^2}{4\nu^2}}^{\nu}\le \parens{\gamma_0}^\nu,\]
as desired. In the above, the last inequality follows for $\nu$ being large enough compared to $C_0$ and $\alpha$.

Finally, we argue~\eqref{eq:bessel-small-theta}. Note that~\eqref{eq:bessel-small-bigO} implies we only need to prove a lower bound. Consider for $\nu\ge 2$:
\[\frac{\abs{b_{\nu+1}}}{\abs{b_{\nu}}}=(z/2)^2\cdot \frac{\Gamma(\nu+\alpha+1)}{(\nu+1)\Gamma(\nu+\alpha+2)}\le 1,\]
where the inequality follows by our choice of $z$ and the fact that $\Gamma(y)$ is increasing for $y\ge 2$ (note that by our choice of $\alpha$ and $\nu$, $\nu+\alpha+1\ge 2$). Thus, for every even $\nu\ge 2$, we have:
\[b_{\nu}+b_{\nu+1}\ge 0.\]
This in turn implies that
\[\bessel{\alpha}{z}\ge \parens{\frac z2}^{\alpha}\cdot \parens{\frac 1{\Gamma(\alpha+1)}-\frac{(z/2)^2}{\Gamma(\alpha+2)}}=\parens{\frac z2}^{\alpha}\cdot \frac 1{\Gamma(\alpha+1)}\cdot\parens{1-\frac{(z/2)^2}{\alpha+1}}\ge \parens{\frac z2}^{\alpha}\cdot \frac 1{4\Gamma(\alpha+1)}\ge \Omega_{\alpha,\eps_0}(1),\]
as desired. In the above, the equality follows from the fact that for any $y>0$, $\Gamma(y+1)=y\Gamma(y)$; the second inequality follows from our assumption that $z\le \sqrt{3(\alpha+1)}$ while the final inequality follows since we assumed $z\ge \eps_0$.
\end{proof}

We will use the Bessel function to present a tighter approximation of the Jacobi polynomials than Theorem~\ref{thm:jacobi-approx}. Towards this end, we begin with the following result:

\begin{thm}[\cite{szego}, Theorem 8.21.12]
\label{thm:jacobi-approx-bessel}
Let $\alpha >-1$ and $\beta$ be an arbitrary real. 
We have for every $n\ge 1$:
\[\parens{\sin\frac\theta 2}^{\alpha}\cdot\parens{\cos\frac\theta 2}^{\beta}\cdot P_n^{(\alpha,\beta)}(\cos\theta)=\g{\alpha}{n}\cdot \sqrt{\frac \theta{\sin\theta}}\cdot \bessel{\alpha}{\parens{n+\addjac}\theta} \pm O\parens{\sqrt{\frac\theta{n^3}}},\]
where
\[0<\theta \le \frac \pi 2.\]
Further, we define
\begin{equation}
\label{eq:g}
\g{\alpha}{n}=\frac 1{\parens{n+\addjac}^{\alpha}}\cdot \frac{\Gamma(n+\alpha+1)}{n!}.
\end{equation}
\end{thm}

We will use the following slight re-statement of the above result:
\begin{cor}
\label{cor:jacobi-approx-bessel-acute}
Let $\alpha >-1$ and $\beta$ be an arbitrary real. 
For every $n\ge 1$:
\[P_n^{(\alpha,\beta)}(\cos\theta)=\kap{\alpha}{\beta}{\theta}\parens{\frac{\g{\alpha}{n}}{\sqrt{2\parens{n+\addjac}}}\cdot \sqrt{\pi\parens{n+\addjac} \theta}\cdot \bessel{\alpha}{\parens{n+\addjac}\theta} \pm O\parens{\frac 1{n^{3/2}}}},\]
for any
\[0<\theta \le \frac \pi 2.\]
Further, in the above
$\g{\alpha}{n}$ is as in~\eqref{eq:g} and
\begin{equation}
\label{eq:kap}
\kap{\alpha}{\beta}{\theta}=\frac{1}{\sqrt{\pi}\cdot \parens{\sin{\frac\theta 2}}^{\alphajac+\frac 12}\parens{\cos\frac\theta 2}^{\betajac+\frac 12}}.
\end{equation}
\end{cor}
\begin{proof}
This follows from Theorem~\ref{thm:jacobi-approx-bessel} by using the definition of $\kap{\alpha}{\beta}{\theta}$, the fact that $\sin\theta=2\sin\frac\theta 2\cos\frac\theta 2$ and the observation that $\theta\sin\theta$ is $O(1)$.
\end{proof}

We also need an approximation for the range $\theta\in \parens{\frac \pi 2,\pi}$, which needs the following result:
\begin{lmm}[Equation (4.1.3) in~\cite{szego}]
\label{lem:jacobi-neg-x}
\[P_n^{(\alpha,\beta)}(-x) = (-1)^n\cdot P_n^{(\beta,\alpha)}(x).\]
\end{lmm}

We are now ready to state the equivalent of Corollary~\ref{cor:jacobi-approx-bessel-acute} for obtuse $\theta$:
\begin{cor}
\label{cor:jacobi-approx-bessel-obtuse}
Let $\beta >-1$ and $\alpha$ be an arbitrary real. 
For every $n\ge 1$, $P_n^{(\alpha,\beta)}(\cos\theta)$ is equal to
\[(-1)^n\cdot \kap{\alpha}{\beta}{\theta}\parens{\frac{\g{\beta}{n}}{\sqrt{2\parens{n+\addjac}}}\cdot \sqrt{\pi\parens{n+\addjac} \parens{\pi-\theta}}\cdot \bessel{\beta}{\parens{n+\addjac}\parens{\pi-\theta}} \pm O\parens{\frac 1{n^{3/2}}}},\]
for any
\[\frac \pi 2 \le \theta < \pi.\]
Further, in the above
$\g{\beta}{n}$ is as in~\eqref{eq:g} (with $\alpha=\beta$) and $\kap{\alpha}{\beta}{\theta}$ is as in~\eqref{eq:kap}.
\end{cor}
\begin{proof}
Let
\[\theta = \pi-\theta'.\]
Note that $\cos\theta=-\cos\theta'$ and thus, by Lemma~\ref{lem:jacobi-neg-x}, we have
\[P_n^{(\alpha,\beta)}(\cos\theta)=(-1)^n\cdot P_n^{(\beta,\alpha)}(\cos\theta').\]
The above with Corollary~\ref{cor:jacobi-approx-bessel-acute} applied with $\theta=\theta'$ along with the observation that $\kap{\beta}{\alpha}{\theta'}=\kap{\alpha}{\beta}{\theta}$ completes the proof.
\end{proof}

Finally, we want to argue that $\g{\alpha}{n}$ is no more than a constant:
\begin{lmm}
\label{lem:g-bound}
For large enough $n$,
\[\g{\alpha}{n}=\Theta_{\alpha,\beta}\parens{1}.\]
\end{lmm}
\begin{proof}
This follows from the Stirling's approximation of the $\Gamma$ function. Indeed we have
\begin{align*}
\g{\alpha}{n}&=\frac 1{\parens{n+\addjac}^{\alpha}}\cdot \frac{\Gamma(n+\alpha+1)}{n!}\\
&=\frac 1{\parens{n+\addjac}^{\alpha}}\cdot \frac{(n+\alpha+1)^{n+\alpha+1/2}\cdot e^n}{n^{n+1/2}\cdot e^{n+\alpha+1}}\cdot \parens{1+O\parens{\frac 1n}}\\
&=\frac {(n+\alpha+1)^\alpha}{\parens{n+\addjac}^{\alpha}}\cdot \frac{(n+\alpha+1)^{n+1/2}}{n^{n+1/2}}\cdot e^{-\alpha-1}\cdot \parens{1+O\parens{\frac 1n}}.
\end{align*}
It is easy to check that each of the terms above are $\Theta_{\alpha,\beta}\parens{1}$, which completes the proof.
\end{proof}

We now have the final approximation for the Jacobi polynomial that we use in our proofs:
\begin{cor}
\label{cor:jacobi-approx-bessel-all}
Let $\alpha,\beta >-1$. 
For every $n\ge 1$:
\begin{align*}
\abs{P_n^{(\alpha,\beta)}(\cos\theta)}=\Theta_{\abs{\alpha},\abs{\beta}}&\left(\frac{\abs{\kap{\alpha}{\beta}{\theta}}}{\sqrt{n}}\cdot \left\{\sqrt{\pi\parens{n+\addjac} \parens{\min(\theta_\ell,\pi-\theta_\ell)}}\cdot \bessel{\zeta}{\parens{n+\addjac}\parens{\min(\theta_\ell,\pi-\theta_\ell)}}\right.\right.\\
&\quad\quad\quad \left.\left.\pm O\parens{\frac 1n}\right\}\right),
\end{align*}
for any
\[0< \theta < \pi,\]
where 
\[\zeta=
\begin{cases}
\alpha &\text{ if }\theta\le \frac \pi 2\\
\beta  &\text{otherwise}
\end{cases}.
\]
\end{cor}
\begin{proof}
The proof follows from Corollaries~\ref{cor:jacobi-approx-bessel-acute},~\ref{cor:jacobi-approx-bessel-obtuse} and Lemma~\ref{lem:g-bound}.
\end{proof}

We will also need the following result from~\cite{szego}:
\begin{lmm}[Equation (4.21.7) in~\cite{szego}]
\label{lem:jac-deriv}
For any integer $j\ge 1$:
\[\frac{d}{dX}\set{P_j^{(\alphajac,\betajac)}(X)}=\frac 12\parens{j+\alphajac+\betajac+1}P_{j-1}^{(\alphajac+1,\betajac+1)}(X).\]
\end{lmm}

\subsection{Christoffel-Darboux formula}
\label{app:CD-formula}

In this section, we will only consider OPs that are orthonormal. In such a case the following result is well-known:
\begin{thm}[\cite{szego}, Thm 3.2.1]
\label{thm:OP-recur}
Let $p_0(X), p_1(X),\dots$ be an OP family defined as in~\eqref{eq:OP-cond}. Then for any $j\ge 2$, we have
\begin{equation}
\label{eq:recur}
p_j(X) = (a_jX+b_j)p_{j-1}(X)- c_jp_{j-2}(X),
\end{equation}
where 
\begin{equation}
\label{eq:c-cond}
a_j\ne 0 \text{ and } c_j=\frac {a_j}{a_{j-1}}.
\end{equation}
\end{thm}

We will need the following bound on $a_j$ for (orthonormal version of) Jacobi polynomials:
\begin{rmk}
\label{rem:a_n}
We note that $a_j=2+o(1)$ for $\jac{\alpha}{\beta}{j}(X)$, where this claim follows from~\eqref{eq:jac-recur} and~\eqref{eq:jac-orthonormal}.
\end{rmk}

Next, we will need the following definition:
\begin{defn}
\label{def:kernel}
Given a family of OP $p_0(X),p_1(X),\dots$ as in~\eqref{eq:OP-cond} and an integer $n\ge 0$, define the {\em kernel} polynomial
\[K_n(x,y)=\sum_{i=0}^n p_i(X)\cdot p_i(Y).\]
\end{defn}

The following result is well-known (we provide its proof for the sake of completeness):
\begin{thm}[Christoffel-Darboux formula, Theorem 3.2.2~\cite{szego}]
\label{thm:CD-formula}
For any $\alpha\ne \beta$ and any $n\ge 1$, we have
\[K_{n-1}(\alpha,\beta)=\frac 1{a_n}\cdot\frac{p_n(\alpha)p_{n-1}(\beta)-p_{n-1}(\alpha)p_n(\beta)}{\alpha-\beta}.\]
\end{thm}
\begin{proof}
We will first prove that for any $j\ge 0$:
\begin{equation}
\label{eq:one-prod}
(\alpha-\beta)\cdot p_j(\alpha)p_j(\beta)=\frac 1{a_{j+1}} \parens{p_{j+1}(\alpha)p_{j}(\beta)-p_{j}(\alpha)p_{j+1}(\beta)} - \frac 1{a_j}  \parens{p_{j}(\alpha)p_{j-1}(\beta)-p_{j-1}(\alpha)p_{j}(\beta)},
\end{equation}
where for notational convenience we define $p_i(X)=0$ for any $i<0$. Note that summing the above for $0\le j<n$ proves the claimed result (where we use the fact that $p_{-1}(X)=0$).

To prove~\eqref{eq:one-prod}, we first use~\eqref{eq:recur} with $j+1$ instead of $j$ and evaluate all the polynomials at $\alpha$ and $\beta$ to get
\[p_{j+1}(\alpha) = (a_{j+1}\alpha+b_{j+1})p_{j}(\alpha)- c_{j+1}p_{j-1}(\alpha)\]
and
\[p_{j+1}(\beta) = (a_{j+1}\beta+b_{j+1})p_{j}(\beta)- c_{j+1}p_{j-1}(\beta).\]
Multiplying the above equations by $p_j(\beta)$ and $p_j(\alpha)$ respectively and then subtracting them one gets
\[p_{j+1}(\alpha)p_j(\beta)-p_{j}(\alpha)p_{j+1}(\beta)=a_{j+1}(\alpha-\beta)p_j(\alpha)p_j(\beta) -c_{j+1}\parens{p_{j-1}(\alpha)p_j(\beta)-p_j(\alpha)p_{j-1}(\beta)}.\]
Re-arranging the above, one gets
\[(\alpha-\beta)p_j(\alpha)p_j(\beta) = \frac 1{a_{j+1}}\parens{p_{j+1}(\alpha)p_j(\beta)-p_{j}(\alpha)p_{j+1}(\beta)}-\frac{c_{j+1}}{a_{j+1}}\parens{p_j(\alpha)p_{j-1}(\beta)-p_{j-1}(\alpha)p_j(\beta)}.\]
The above along with~\eqref{eq:c-cond} implies~\eqref{eq:one-prod}.
\end{proof}

We will use the following corollary for the above result:
\begin{cor}[(3.2.4) in~\cite{szego}]
\label{cor:kernel-x=y}
For any $\alpha$ and any $n\ge 1$, we have
\[K_{n-1}(\alpha,\alpha)=\frac{p'_n(\alpha)p_{n-1}(\alpha)-p'_{n-1}(\alpha)p_n(\alpha)}{a_n}.\]
\end{cor}

\subsection{Proof of Theorem~\ref{cor:jac-U-bound}}
\label{app:jac-U-bound}

We are now ready to argue that the OP transform corresponding to Jacobi polynomials for $\alpha,\beta> -1$ is flat, i.e. we prove  Theorem~\ref{cor:jac-U-bound} (which we re-state below for convenience):
\begin{thm}[Theorem~\ref{cor:jac-U-bound}, restated]
\label{cor:jac-flat} 
Let $\alpha,\beta> - 1$.
Let $N\ge 1$ be large enough.
Let $\evalpts_0,\dots,\evalpts_{N-1}$ be the roots of the $N$th Jacobi polynomial. Then define $\vF$ such that for every $0<\ell,j<N$, we have
\[\vF[\ell,j]=\jac{\alpha}{\beta}{j}(\evalpts_{\ell})\cdot\sqrt{w_{\ell}},\]
where
\[w_{\ell}=\frac 1{\sum_{j=0}^{N-1} \jac{\alpha}{\beta}{j}(\evalpts_{\ell})^2}.\]
Then for every $0\le \ell<N$, we have
\[\frac 1{w_\ell} = O\parens{N\cdot\kap{\alpha}{\beta}{\theta_\ell}^2}\]
and
\[\max_{0<\ell,j<N} \abs{\vF[\ell,j]}=\frac{O_{\abs{\alpha},\abs{\beta}}(1)}{\sqrt{N}}.\]
\end{thm}
\begin{proof} We will argue that for any $0\le \ell,j<N$, we have 
\begin{equation}
\label{eq:jac-abs-ub}
\abs{P_j^{(\alpha,\beta)}(\cos\theta_{\ell})}=O_{\abs{\alpha},\abs{\beta}}\parens{\frac{\kap{\alpha}{\beta}{\theta_\ell}}{\sqrt{j}}}.
\end{equation}
The definition of the normalized Jacobi polynomials and Lemma~\ref{lem:norm-fact-bound} implies that
\[\abs{\jac{\alpha}{\beta}{j}\parens{\cos\theta_{\ell}}}=O_{\abs{\alpha},\abs{\beta}}\parens{\kap{\alpha}{\beta}{\theta_\ell}}.\]
The above proves both of the claimed results. Indeed the upper bound on $1/w_\ell$ follows by summing up the square of the above bound for all $0\le j<N$. The above along with the bound of $O_{\abs{\alpha},\abs{\beta}}\parens{\frac{1}{N\kap{\alpha}{\beta}{\theta_\ell}^2}}$ on $w_{\ell}$ from Lemma~\ref{lem:w-l-lb}
would prove our desired bound on entries of $\vF$.

Finally, we note that~\eqref{eq:jac-abs-ub} follows from applying~\eqref{eq:bessel-large-simple} to the bound in Corollary~\ref{cor:jacobi-approx-bessel-all}.
\end{proof}

\begin{lmm}
\label{lem:w-l-lb}
Let $\alpha,\beta> - 1$.
Let $N\ge 1$ be large enough. Let $\evalpts_0,\dots,\evalpts_{N-1}$ be the roots of the $N$th Jacobi polynomial.
Then
\[\sum_{j=0}^{N-1} \jac{\alpha}{\beta}{j}(\cos\theta_\ell)^2 \ge \Omega_{\abs{\alpha},\abs{\beta}}\parens{N\kap{\alpha}{\beta}{\theta_\ell}^2}.\]
\end{lmm}
The above result follows immediately from the following lemma and Lemma~\ref{lem:jacobi-root-bounded}:

\begin{lmm}
\label{lem:w-l-lb-detailed}
Let $\alpha,\beta> - 1$.
Let $N\ge 1$ be large enough. Let $\evalpts_0,\dots,\evalpts_{N-1}$ be the roots of the $N$th Jacobi polynomial.
Then for {\em every} constant $c>0$, the following holds. Let $\theta_\ell$ be such that
\[\frac cN \le \theta_\ell \le \pi-\frac cN.\]
Then
\[\sum_{j=0}^{N-1} \jac{\alpha}{\beta}{j}(\cos\theta_\ell)^2 \ge \Omega_{\abs{\alpha},\abs{\beta}}\parens{N\kap{\alpha}{\beta}{\theta_\ell}^2}.\]
\end{lmm}
\begin{proof}
We will first prove the result for every fixed $c\ge C_0$, where $C_0$ is a constant that we will fix later. Then we will prove the result for any constant $0<c<C_0$.

We start with the case of large $c$. We first observe that by Definition~\ref{def:kernel}, we need to lower bound $K_{N-1}(\evalpts_\ell,\evalpts_\ell)$. Then we use Corollary~\ref{cor:kernel-x=y} and the fact that $\evalpts_\ell$ is a root of the $N$th Jacobi polynomial, to note that
\begin{align*}
K_{N-1}(\evalpts_\ell,\evalpts_\ell)&=\frac 1{a_N}\cdot  \jac{\alpha}{\beta}{N-1}(\evalpts_\ell)\cdot \parens{\frac{d}{dX}\set{\jac{\alphajac}{\betajac}{N}\parens{X}}}_{X\gets \evalpts_\ell}\\
&=
\frac 1{a_N\sqrt{h_N^{\alpha,\beta}\cdot h_{N-1}^{\alpha,\beta}}} \cdot P^{(\alpha,\beta)}_{N-1}(\evalpts_\ell)\cdot \parens{\frac{d}{dX}\set{P^{(\alphajac,\betajac)}_{N}\parens{X}}}_{X\gets \evalpts_\ell}.
\end{align*}
Remark~\ref{rem:a_n} and Lemma~\ref{lem:norm-fact-bound} imply that for large enough $N$, $\frac 1{a_N\sqrt{h_N^{\alpha,\beta}\cdot h_{N-1}^{\alpha,\beta}}} =\Theta_{\abs{\alpha},\abs{\beta}}\parens{N}$. Thus, to complete the proof, we need to argue that
\[E_{\ell,\alpha,\beta,N}\stackrel{\text{def}}{=}P^{(\alpha,\beta)}_{N-1}(\evalpts_\ell)\cdot \parens{\frac{d}{dX}\set{P^{(\alphajac,\betajac)}_{N}\parens{X}}}_{X\gets \evalpts_\ell}\ge \Omega_{\abs{\alpha},\abs{\beta}}\parens{\kap{\alpha}{\beta}{\theta_\ell}^2}.\]
Towards this end, we recall the following identity (which appears as equation (4.5.7) in~\cite{szego}) that holds for any $n\ge 1$:
\begin{align*}
\parens{2n+\alpha+\beta}&\parens{1-X^2}\cdot \parens{\frac{d}{dX}\set{P^{(\alphajac,\betajac)}_{n}\parens{X}}}\\
&=-n\set{(2n+\alpha+\beta)X+\beta-\alpha}\cdot P_n^{(\alpha,\beta)}(X)+2(n+\alpha)(n+\beta)\cdot  P_{n-1}^{(\alpha,\beta)}(X).
\end{align*}
Using the above for $n=N$, substituting $X=\evalpts_\ell$ and noting that $1-\evalpts_\ell^2=\sin^2\theta_\ell$, we get:
\[P^{(\alpha,\beta)}_{N-1}(\evalpts_\ell)= \frac{(2N+\alpha+\beta)\sin^2\theta_\ell}{2(N+\alpha)(N+\beta)}\cdot \parens{\frac{d}{dX}\set{P^{(\alphajac,\betajac)}_{N}\parens{X}}}_{X\gets \evalpts_\ell}.
\]
This implies that
\begin{align*}
E_{\ell,\alpha,\beta,N}&= \frac{(2N+\alpha+\beta)\sin^2\theta_\ell}{2(N+\alpha)(N+\beta)}\brackets{\parens{\frac{d}{dX}\set{P^{(\alphajac,\betajac)}_{N}\parens{X}}}_{X\gets \evalpts_\ell}}^2\\
&= \frac{(2N+\alpha+\beta)(N+\alpha+\beta+1)^2\cdot\sin^2\theta_\ell}{8(N+\alpha)(N+\beta)}\cdot \set{P^{(\alphajac+1,\betajac+1)}_{N-1}\parens{\evalpts_\ell}}^2,
\end{align*}
where the second inequality follows from Lemma~\ref{lem:jac-deriv} (with $j=N$).

Thus to show $E_{\ell,\alpha,\beta,N} \ge \Omega_{\abs{\alpha},\abs{\beta}}\parens{\kap{\alpha}{\beta}{\theta_\ell}^2}$, it is enough to argue
\[ \abs{P^{(\alpha+1,\beta+1)}_{N-1}(\evalpts_\ell)} \ge \Omega_{\abs{\alpha},\abs{\beta}}\parens{\frac{\kap{\alpha}{\beta}{\theta_\ell}}{\sin\theta_\ell\cdot\sqrt{N}}}=\Omega_{\abs{\alpha},\abs{\beta}}\parens{\frac{\kap{\alpha+1}{\beta+1}{\theta_\ell}}{\sqrt{N}}},\]
where the equality follows from the fact that $\kap{\alpha+1}{\beta+1}{\theta}=\frac{\kap{\alpha}{\beta}{\theta}}{\sin\frac \theta 2\cos\frac \theta 2}=\frac{2\kap{\alpha}{\beta}{\theta}}{\sin\theta}$.
Corollary~\ref{cor:jacobi-approx-bessel-all} (with $n=N-1$ and where $\alpha \gets \alpha + 1$ and $\beta \gets \beta + 1$) implies that the above is true if
\[\abs{\sqrt{\pi\parens{N+\addjac} \parens{\min(\theta_\ell,\pi-\theta_\ell)}}\cdot \bessel{\zeta+1}{\parens{N+\addjac}\parens{\min(\theta_\ell,\pi-\theta_\ell)}}} \ge \Omega_{\abs{\alpha},\abs{\beta}}\parens{1},\]
where in the above we have used the fact that $N$ is large enough and
where $\zeta$ is as defined in Corollary~\ref{cor:jacobi-approx-bessel-all}.
For notational simplicity, define
\[z=\parens{N+\addjac}\parens{\min(\theta_\ell,\pi-\theta_\ell)}.\]
Let $z$ be large enough so that~\eqref{eq:bessel-large} holds. I.e. there exist  constant $C_0$ such that
\[\frac {C_0}N\le \theta_\ell\le \pi-\frac{C_0}N\]
and
\[\abs{\sqrt{\pi z}\bessel{\zeta}{z}}= \abs{\sqrt{2}\cdot\cos\parens{z-\parens{\zeta+\frac 12}\cdot\frac \pi 2-\frac \pi 2}}\pm O\parens{\frac 1z}=\abs{\sqrt{2}.\sin\parens{z-\parens{\zeta+\frac 12}\cdot\frac \pi 2}}\pm O\parens{\frac 1z}
\]
We will argue that
\begin{equation}
\label{eq:sin-lb-at-root}
\abs{\sin\parens{z-\parens{\zeta+\frac 12}\cdot\frac \pi 2}}\ge 1-O\parens{\frac 1z}.
\end{equation}
The above implies that it is enough to show that $1-O\parens{\frac 1z}\ge \Omega_{\abs{\alpha},\abs{\beta}}\parens{1}$, which is true if we pick $C_0$ to be large enough. 

Now, we argue~\eqref{eq:sin-lb-at-root}. Noting that $P_N^{(\alpha,\beta)}(\evalpts_\ell)=0$, Corollary~\ref{cor:jacobi-approx-bessel-all} implies that 
\[\abs{\sqrt{\pi z}\bessel{\zeta}{z}} \le O\parens{\frac 1N}.\]
Since we have assumed $z$ is large enough and~\eqref{eq:bessel-large} holds, we have that
\[\abs{\sqrt{2}\cdot\cos\parens{z-\parens{\zeta+\frac 12}\cdot\frac \pi 2}}\pm O\parens{\frac 1z}\le O\parens{\frac 1N}.\]
In other words, we have
\[\abs{\cos\parens{z-\parens{\zeta+\frac 12}\cdot\frac \pi 2}}\le O\parens{\frac 1z},\]
where in the above we used the fact that $N\ge \Omega(z)$. The above, along with the fact that $\abs{\sin A}\ge 1-\abs{\cos A}$, implies~\eqref{eq:sin-lb-at-root}.

If the constant $c$ in the lemma statement satisfies $c\ge C_0$, then we are done. So let us assume that $c<C_0$ and WLOG assume $\theta_\ell=\frac cN$. We will use the other bounds from Lemma~\ref{lem:bessel-approx} to prove the lemma for this case. Towards that end, define
\[\eps_0=\frac 12\cdot\min\set{c,\sqrt{3(\alpha+1)},2},\]
and for any $0\le j<N$,
\[z_j=\parens{j+\addjac}\parens{\min(\theta_\ell,\pi-\theta_\ell)}.\]
Then by our choice of $\eps_0$ and $c$, for some small enough constant $\mu>0$, there exists a subset $S\subseteq [0,N-1]$ with
\begin{equation}
\label{eq:size-S}
\abs{S}\ge \parens{\frac{\eps_0}{c}-\mu}\cdot N-2
\end{equation}
such that for every $j\in S$, we have
$j\ge \mu N$ and:
\[z_j\in \brackets{\eps_0,\min\set{2,\sqrt{3(\alpha+1)}}}.\]
(Indeed, the above range is of size at least $\eps_0$ and every increment in $j$ increases $z_j$ by at most $\frac cN$, which means that are at least $\frac {\eps_0 N}{c}-2$ $z_j$'s in the above range. We lose at most a further factor of $\mu N$ to ensure that every $j\in S$ satisfies $j\ge \mu N$.)
The above along with Lemma~\ref{lem:bessel-approx}, implies that for every $j\in S$:
\[\sqrt{\pi z_j}\bessel{\zeta}{z_j}\ge \Omega_{\abs{\alpha},\abs{\beta}}\parens{1}.\]
The above with Corollary~\ref{cor:jacobi-approx-bessel-all} implies that for any $j\in S$ and large enough $N$:
\begin{equation}
\label{eq:P_j-lb-small-c}
\abs{P_j^{(\alpha,\beta)}(\cos\theta_\ell)} \ge \frac{\abs{\kap{\alpha}{\beta}{\theta_\ell}}}{\sqrt{N}}\cdot \Omega_{\abs{\alpha},\abs{\beta}}\parens{1}.
\end{equation}
Thus, we have
\begin{align*}
K_{N-1}\parens{\evalpts_\ell,\evalpts_\ell}&\ge \sum_{j\in S} \parens{\jac{\alpha}{\beta}{\cos\theta_\ell}}^2\\
&=\sum_{j\in S} \frac{1}{h_j^{\alpha,\beta}}\cdot P_j^{(\alpha,\beta)}(\cos\theta_\ell)^2\\
&\ge \Omega_{\abs{\alpha},\abs{\beta}}\parens{N}\sum_{j\in S} P_j^{(\alpha,\beta)}(\cos\theta_\ell)^2\\
&\ge \Omega_{\abs{\alpha},\abs{\beta}}\parens{N\parens{\kap{\alpha}{\beta}{\theta_\ell}}^2},
\end{align*}
as desired. In the above, the second inequality follows from Lemma~\ref{lem:norm-fact-bound} and the fact that $j\ge \Omega(N)$ and the final inequality follows from~\eqref{eq:P_j-lb-small-c} and the fact that $\abs{S}\ge\Omega(N)$ (which in turn follows from~\eqref{eq:size-S} and our choice of parameters). The proof is complete.
\end{proof}

\section{A number theory problem}
\label{sec:mod-1}

Let $N$ be some large enough integer. In this section we will consider sequences of reals: $0\le y_1\le y_2\le \dots\le `y_N<1$ that have certain nice proprieties when multiplied by integers.

We begin with some notation: given $z\in\R$, let $\set{z}$ denote the fractional part of $z$, i.e.
\[\set{z}=z-\floor{z}.\]
We will also denote $\set{z}$ as $z\mod{1}$. We will also need the notation $\inner{z}$, which is its distance from the closest integer to $z$. Equivalently,
\[\inner{z}=\min\parens{\set{z},1-\set{z}}.\]

We are now ready to define the kind of sequences we will encounter in our work:
\begin{defn}
\label{def:scatter}
A  sequences of reals: $0\le y_1\le y_2,\dots,y_N<1$ is called $s$-scattered (for any integer $s\ge 1$) if for any of the intervals $\left[\frac iN,\frac{i+1}N\right)$ (with $i\in \Z_N$) has at most $s$ elements from the sequence in it.
\end{defn}

We next define our notion of when a real number of {\em good}:
\begin{defn}
\label{def:eps-bad}
Let $0\le \eps\le 1$. We call a real $y\in [0,1)$ to be $\eps$-{\em good} if for any reals $0\le \ell<r<1$, we have that
\[\abs{\frac{\abs{\set{x\in\Z_N| xy\mod{1}\in [\ell,r]}}}{N}- (r-\ell)} \le \eps.\]
If the above is not satisfied then we call $y$ to be $\eps$-{\em bad}.
\end{defn}

We are interested in bounding how many bad $y_i$'s can be there in an $O(1)$-scattered sequence:
\begin{lmm}
\label{lem:bad-in-spread-sequence}
Let $0\le y_1\le y_2,\dots,y_N<1$ be an $O(1)$-scattered sequence. Define 
\[B=\set{i\in\Z_N| y_i\text{ is } \eps-\text{bad}}.\]
Then, we have
\[\abs{B}\le O\parens{\frac 1{\eps^2}}.\]
\end{lmm}

\subsection{Proof of Lemma~\ref{lem:bad-in-spread-sequence}}

We thank Stefan Steinerberger for showing us the following proof and kindly allowing us to use it.

We begin with Dirichlet's approximation theorem, which implies the following result.

\begin{lmm}
\label{lem:dirichlet}
For every $i\in \Z_N$, we have that there exists integers $0\le p_i\le q_i\le N$ with $q_i\ge 1$ such that $\gcd(p_i,q_i)=1$ and
\[\abs{y_i-\frac{p_i}{q_i}}\le \frac 1{q_iN}.\]
\end{lmm}

The above immediately implies the following:
\begin{cor}
\label{cor:rational-approx}
Let $i\in \Z_N$ be such that $q_i\ge \ceils{\frac 4\eps}$. Then we have for every $x\in\Z_N$
\[\abs{x\cdot y_i-x\cdot \frac{p_i}{q_i}}\le \frac \eps 4.\]
\end{cor}

The above basically says that we can essentially look at the goodness of rationals. In particular,
\begin{lmm}
\label{lem:good-rational}
Consider a rational $\frac ab$ with $1\le a\le b\le N$ with $\gcd(a,b)=1$. Further let $b\ge  \ceils{\frac 4\eps}$. Then $\frac ab$ is $\frac \eps 2$-good.
\end{lmm}
\begin{proof}
Since $\gcd(a,b)=1$ (and $a\ne 0$), we have that the values $x\cdot \frac ab$ over all $x\in\Z_N$, take the values in $\frac 1b,\frac 2b,\dots\frac {b-1}b,1$ between $\floor{\frac Nb}$ and $\ceils{\frac Nb}$ times. It is not too hard to see that out of these $b$ values $(r-\ell)\cdot b \pm 2$ values can fall in the range $[\ell,r]$.\footnote{Basically, in the worst case one can get two ``extra" elements at the end points of $[\ell,r]$ for the upper bound. For the lower bound the closest two points to $\ell$ and $r$ might be just outside of $[\ell,r]$.} This implies that
\[\abs{\frac{\abs{\set{x\in\Z_N| x\frac ab\mod{1}\in [\ell,r]}}}{N}- (r-\ell)} \le \frac 2b\le \frac \eps 2,\]
as desired.
\end{proof}

Now Corollary~\ref{cor:rational-approx}, Lemma~\ref{lem:good-rational} and the triangle inequality implies that
\begin{cor}
\label{cor:good-points}
Let $i\in \Z_N$ be such that $q_i\ge \ceils{\frac 4\eps}$. Then $y_i$ is $\frac {3\eps}4\le \eps$-good. 
\end{cor}

We are now pretty much done. By the above result, we have that all $\eps$-bad $y_i$ have $q_i\le \ceils{\frac 4\eps}$. Further, since the sequence $\set{y_i}_{i\in\Z_N}$ is $O(1)$-scattered, Lemma~\ref{lem:dirichlet} also implies that any rational $\frac ab$ with $0\le a\le b\le N$ and $\gcd(a,b)=1$, is ``assigned" at most $O(1)$ many $y_i$.\footnote{Indeed Lemma~\ref{lem:dirichlet} implies that for any $i\in \Z_N$, we have $\abs{y_i-\frac{p_i}{q_i}}\le \frac 1N$ (since we must have $q_i\ge 1$).} The number of rationals $\frac {p_i}{q_i}$ with $q_i\le \ceils{\frac 4\eps}$ is trivially at most $\parens{\ceils{\frac 4\eps}}^2$. Thus, the overall number of $\eps$-bad $y_i$'s is at most $O\parens{\frac 1{\eps^2}}$, as desired. Note that this argument implies the following result, which will be needed in our algorithms:
\begin{cor}
\label{cor:bad-y-intervals}
Given a sequence $\set{y_i}_{i\in\Z_N}$ that is $O(1)$-scattered, all the $\eps$-bad $y_i$'s are contained in $O\parens{\frac 1{\eps^2}}$ intervals, each of size at most $\frac 2N$.
\end{cor}

\end{document}